\newcommand{\authorstrut}{\vphantom{$^{*\alpha\dagger}$}} %
\newcommand{\A}[1]{\authorstrut #1}

\makeatletter
\renewcommand*{\@fnsymbol}[1]{%
    \ensuremath{%
        \ifcase#1\or *\or \alpha\or \mathsection\or \ddagger\or
            \dagger\or \mathparagraph\or \|\or **\or
            \dagger\dagger \or \ddagger\ddagger \else\@ctrerr\fi}}
\makeatother

\title{Determination of the fifth Busy Beaver value}
\author{
The bbchallenge Collaboration\thanks{\url{https://bbchallenge.org}}\and
\A{Justin Blanchard}\thanks{Alphabetical ordering.}\and
\A{Daniel Briggs}\and
\A{Konrad Deka}\and
\A{Nathan Fenner}\and
\A{Yannick Forster}\and
\A{Georgi Georgiev (Skelet)}\and
\A{Matthew L. House}\and
\A{Rachel Hunter}\and
\A{Iijil}\and
\A{Maja Kądziołka}\and
\A{Pavel Kropitz}\and
\A{Shawn Ligocki}\and
\A{mxdys}\and
\A{Mateusz Na\'{s}ciszewski}\and
\A{savask}\and
\A{Tristan Stérin}\thanks{Corresponding author: \texttt{tristan@prgm.dev}.}\and
\A{Chris Xu}\and
\A{Jason Yuen}\and
\A{Théo Zimmermann}
}

\documentclass[a4paper,british]{article}

\usepackage{babel}
\usepackage[utf8]{inputenc}
\usepackage[T1]{fontenc}
\usepackage[margin=1in]{geometry}
\usepackage[
    colorlinks=true,
    linkcolor=black,        %
    urlcolor=blue!50!black, %
    citecolor=blue!50!black %
]{hyperref}
\usepackage{caption}
\usepackage{floatpag}
\usepackage{subcaption}
\usepackage{tikz}
\usepackage{mwe} %

\usepackage{wrapfig}
\usepackage{refcount}
\usetikzlibrary{calc}
\usetikzlibrary{arrows.meta}

\usepackage{algorithm}
\usepackage[noend]{algpseudocode}
\usepackage{xcolor,colortbl,tikz}

\definecolor{graySymb1}{RGB}{85,85,85}
\definecolor{graySymb2}{RGB}{170,170,170}

\usepackage{bold-extra}%

\usepackage{graphicx}
\usepackage{mathtools}

\usepackage{booktabs}

\usepackage{amsmath,amsfonts,amssymb,amsthm}

\theoremstyle{definition} %
\newtheorem{theorem}{Theorem}[section]
\newtheorem{conjecture}{Conjecture}[section]
\newtheorem{problem}{Problem}[section]
\newtheorem{definition}{Definition}[section]
\newtheorem{lemma}{Lemma}[section]

\numberwithin{equation}{section}

\theoremstyle{definition} %

\newtheorem{example}{Example}[section]

\usepackage{xassoccnt}
\DeclareCoupledCountersGroup{theorems}
\DeclareCoupledCounters[name=theorems]{theorem,definition,lemma,proposition,corollary,remark,observation,example}

\usepackage{microtype,xspace,wrapfig,multicol}
\usepackage[textsize=tiny,color=lightgray]{todonotes}
\usepackage[normalem]{ulem} %
\usepackage{stmaryrd}

\newcommand{\ts}[1]{{\color{red}#1}}

\newcommand{\tabi}{\hspace{\algorithmicindent}}

\newcommand{\N}{\mathbb{N}}
\newcommand{\Z}{\mathbb{Z}}

\newcommand{\tm}[1]{\href{https://bbchallenge.org/#1}{\texttt{\nolinkurl{#1}}}}

\usepackage{xcolor}
\usepackage{adjustbox}

\definecolor{colorA}{RGB}{255,0,0}
\definecolor{colorB}{RGB}{255,128,0}
\definecolor{colorC}{RGB}{0,0,255}
\definecolor{colorD}{RGB}{0,255,0}
\definecolor{colorE}{RGB}{255,0,255}

\usepackage{epigraph}

\newcommand{\stateA}{{\textcolor{colorA}{A}}\xspace}
\newcommand{\stateB}{{\textcolor{colorB}{B}}\xspace}
\newcommand{\stateC}{{\textcolor{colorC}{C}}\xspace}
\newcommand{\stateD}{{\textcolor{colorD}{D}}\xspace}
\newcommand{\stateE}{{\textcolor{colorE}{E}}\xspace}

\newcommand{\stateAx}{{\textcolor{colorA}{A}}}
\newcommand{\stateBx}{{\textcolor{colorB}{B}}}
\newcommand{\stateCx}{{\textcolor{colorC}{C}}}
\newcommand{\stateDx}{{\textcolor{colorD}{D}}}
\newcommand{\stateEx}{{\textcolor{colorE}{E}}}

\usepackage{enumitem}

\newcommand{\szero}{\texttt{0}\xspace}
\newcommand{\sone}{\texttt{1}\xspace}

\def\@fnsymbol#1{\ensuremath{\ifcase#1\or *\or \dagger\or \ddagger\or
            \mathsection\or \mathparagraph\or \|\or **\or \dagger\dagger
        \or \ddagger\ddagger \else\@ctrerr\fi}}
\newcommand{\ssymbol}[1]{^{\@fnsymbol{#1}}}

\newcommand{\BBtheFourth}{107}

\newcommand{\BBtheFourthTNF}{858{,}909}

\newcommand{\partialto}{\to}

\newcommand{\BBtheFifth}{47{,}176{,}870}
\newcommand{\BBtheFifthTNF}{181{,}385{,}789}
\newcommand{\SigmaTheFifth}{4{,}098}

\newcommand{\BBTxF}{3{,}932{,}964}
\newcommand{\BBTxFTNF}{2{,}154{,}217}

\newcommand{\numBBsholdouts}{1{,}214}

\newcommand{\radofull}{Tibor Rad\'o\xspace}
\newcommand{\rado}{Rad\'o\xspace}

\newcommand{\cycler}{Cycler\xspace}
\newcommand{\cyclers}{Cyclers\xspace}
\newcommand{\TC}{Translated Cycler\xspace}
\newcommand{\TCs}{Translated Cyclers\xspace}

\newcommand{\headpos}{head-position\xspace}
\newcommand{\headposs}{head-positions\xspace}

\newcommand{\states}{\mathcal{S}}
\newcommand{\alphabet}{\mathcal{A}}
\newcommand{\balphabet}{\left\{\szero,\sone\right\}}
\newcommand{\symbolzero}{\texttt{0}}

\newcommand{\numSporadic}{13\xspace}

\newcommand{\ssp}{state-symbol pair\xspace}
\newcommand{\ssps}{state-symbol pairs\xspace}

\newcommand{\HALT}{\texttt{HALT}\xspace}
\newcommand{\NONHALT}{\texttt{NONHALT}\xspace}
\newcommand{\UNKNOWN}{\texttt{UNKNOWN}\xspace}

\newcommand{\CoqBB}{Coq-BB5\xspace}

\newcommand{\TMstep}{\to}

\usepackage{tocloft}

\usepackage{thm-restate}
\usepackage{makecell}

\begin{document}
\date{}

\maketitle

\vspace{-1em}
\begin{abstract}
    The Busy Beaver value $S(n)$ is the maximum~number of steps that an n-state 2-symbol Turing machine can perform from the all-zero tape before halting. $S$ was historically introduced by \radofull in 1962 as one of the simplest examples of an uncomputable function.

    We prove that $S(5) = 47,176,870$ using the Coq proof assistant. The proof enumerates 181,385,789 Turing machines with 5 states and, for each machine, decides whether it halts or not.
    Our result marks the first determination of a new Busy Beaver value in over 40 years and the first Busy Beaver value ever to be formally verified, attesting to the effectiveness of massively collaborative online research (\url{bbchallenge.org}).
\end{abstract}

\setcounter{tocdepth}{2}
\begingroup
\renewcommand{\baselinestretch}{0.85}\normalsize
\tableofcontents
\endgroup

\newpage

\begin{center}

\end{center}
\vspace{-20pt}
\setlength{\epigraphwidth}{0.7\textwidth}
\epigraph{In any case, even though skilled mathematicians and experienced programmers attempted to evaluate $\Sigma(3)$ and S(3), there is no evidence that any presently known approach will yield the answer, even if we avail ourselves of high-speed computers and elaborate programs. As regards $\Sigma(4)$, $S(4)$ the situation seems to be entirely hopeless at present.}{\radofull, 1963 \cite{Rado_1963}}
\vspace{-1em}
\epigraph{\textit{Prediction 5}. It will never be proved that $\Sigma(5) = \SigmaTheFifth$ and $S(5) = \BBtheFifth$.}{Allen H. Brady, 1990 \cite{BradyMeaningOfLife}}
\vspace{2em}

\vspace{-30pt}

\section{Introduction}\label{sec:intro}

\subsection{Main Result}\label{sec:intro:mainresults}

\newcommand{\noncomput}{uncomputable\xspace}
\newcommand{\BBfull}{Busy Beaver\xspace}
\newcommand{\Coq}{Coq\xspace}
\newcommand{\CoqProofReleaseURL}{\url{https://github.com/ccz181078/Coq-BB5}}

\newcommand{\ie}{i.e.~}
\newcommand{\eg}{e.g.~}

Are there simple \noncomput functions? What is the \textit{smallest} open problem in mathematics? What do algorithms look like, \textit{in the wild}?

Introduced by \radofull in 1962, \textit{the \BBfull game} gives a framework to answer these seemingly independent questions, starting with the first one: \rado's original goal was to ``present some very simple instances of non-computable functions'' \cite{Rado_1962}. The game is as such: (i) run all $n$-state 2-symbol Turing machines (see Section~\ref{sec:TMs}) from the all-zero tape; (ii) consider the set of machines that eventually halt; (iii) the winner of the game is the halting machine that has the most \sone symbols on its tape when it halts. This maximum number of \sone symbols on final tape among $n$-state halting machines is called $\Sigma(n)$. \rado also introduced $S(n)$, the maximum number of steps made by a halting $n$-state Turing machine from the all-zero tape.\footnote{We avoid using notation $\text{BB}(n)$ in this work as it historically meant $\Sigma$ \cite{Rado_1962, 10.5555/1151785.1151794} and later shifted to mean $S$ \cite{BusyBeaverFrontier,sterin_2022_14955828}.} Both functions $\Sigma$ and $S$ are \noncomput and this is most obvious in the case of $S$: if an $n$-state machine runs for more than $S(n)+1$ steps, we know it will never halt, giving an algorithm to decide Turing's halting problem\footnote{In the variant where machines are given no input and instead start from the all-zero tape.} if $S$ were computable. Because of this tight link between $S$ and the halting problem, we take the liberty to focus our work on $S$.

While there is no algorithm to compute $S$ for \textit{all} $n$, we can certainly try to compute $S$ for \textit{some} $n$. Prior to this work, only the first four values of $S$ had been proved: $S(1)=1$, $S(2)=6$ \cite{Rado_1962}, $S(3) = 21$ \cite{Lin1963}, and $S(4) = 107$ \cite{Brady83}. With some early attempts in the 1960s and 1970s, the $S(5)$ quest started in earnest in 1983 with a 2-day competition organised at the University of Dortmund\footnote{Report of the competition: \url{https://docs.bbchallenge.org/other/lud20.pdf}.} with the sole goal of finding new 5-state champions -- \ie 5-state machines achieving higher step-count than any previously known machines \cite{PMichel_website,michel2019busy}. The winning machine in Dortmund, found by Uwe Schult, achieved $134{,}467$ steps, establishing $S(5) \geq 134{,}467$. In 1989, significant progress was made when Heiner Marxen and J\"urgen Buntrock found a new champion achieving $\BBtheFifth$ steps \cite{Marxen_1990}, showing $S(5) \geq \BBtheFifth$; this machine is given in Figure~\ref{fig:bb5win}. However, it remained unknown if no other machine could beat it, \ie whether Marxen and Buntrock's machine was the actual 5-state Busy Beaver winner or not. In 2020, based on the lack of a new 5-state champion in 30 years, Scott Aaronson conjectured that it was the winner, and thus, that $S(5) = \BBtheFifth$ \cite{BusyBeaverFrontier}.

Our main result is to prove this conjecture, using the \Coq proof assistant \cite{the_coq_development_team_2024_14542673}, see Theorem~\ref{th:BB5}. The \Coq proof is called \CoqBB and is available at \href{https://github.com/ccz181078/Coq-BB5}{\texttt{github.com/ccz181078/Coq-BB5}} \cite{mxdys_2025_17061968}. We also prove $\Sigma(5) = \SigmaTheFifth$; see Section~\ref{sec:results}, Theorem~\ref{th:Sigma5}. The goal of this paper is to serve as a ``human readable'' version of \CoqBB.

\begin{restatable}[\CoqBB: \texttt{Lemma BB5\_value}]{theorem}{thBBTheFifth}\label{th:BB5}
    $S(5) = \BBtheFifth$.
\end{restatable}

The function $S$ can naturally be extended to Turing machines using more than two alphabet symbols \cite{BradyMeaningOfLife}; for instance, $S(2,3) = 38$ is the value of $S$ for 2-state, 3-symbol machines \cite{BradyMeaningOfLife, MICHEL200445, LafittePapazian2007}. We prove, using Coq, that $S(2,4) = \BBTxF$, see Theorem~\ref{th:BB2x4} and Figure~\ref{fig:bb2x4}:

\begin{restatable}[\CoqBB: \texttt{Lemma BB2x4\_value}]{theorem}{thBBTxF}\label{th:BB2x4}
    $S(2,4) = \BBTxF$.
\end{restatable}

\CoqBB provides formal proofs for $S(5)$ and $S(2,4)$ --- as well as for previously known $S(2),\,S(3),\,S(4)$ and $S(2,3)$. The lists of all the Turing machines enumerated by these proofs, together with their Coq-verified halting status, are available at \url{https://docs.bbchallenge.org/CoqBB5_release_v1.0.0/}.

As a result of our work, we now have a clearer view of the landscape of small Busy Beaver values; see Table~\ref{table:landscape}.

\setlength{\fboxrule}{1.2pt}
\begin{table}[h]
    \centering
    \small
    \renewcommand{\arraystretch}{1.3}
    \setlength{\tabcolsep}{5pt}  %
    \begin{tabular}{c|ccccc}
        \hline
        \textbf{Symbols} & \textbf{2-State}                                               & \textbf{3-State} & \textbf{4-State} & \textbf{5-State} & \textbf{6-State} \\
        \hline
        2                & \cellcolor{green!20}$S(2) = 6$ \cite{Rado_1962}
                         & \cellcolor{green!20}$S(3) = 21$ \cite{Lin1963}
                         & \cellcolor{green!20}$S(4) = 107$ \cite{Brady83}
                         & \cellcolor{green!50}{$S(5) = \BBtheFifth$}
                         & \cellcolor{orange!50}{$S(6) > 2 \uparrow \uparrow \uparrow 5$}                                                                             \\
        \hline
        3                & \cellcolor{green!20}$S(2,3) = 38$ \cite{LafittePapazian2007}
                         & \cellcolor{orange!50}{$S(3,3) > 10^{17}$}
                         & \cellcolor{orange!20}$S(4,3) > 10 \uparrow^4 4 $
                         & --                                                             & --                                                                        \\
        \hline
        4                & \cellcolor{green!50}{$S(2,4) = \BBTxF$}
                         & \cellcolor{orange!20}$S(3,4) > 2 \uparrow^{15} 5 $
                         & --                                                             & --               & --                                                     \\
        \hline
        5                & \cellcolor{orange!50}{$S(2,5) > 10 \uparrow \uparrow 4$}
                         & --                                                             & --               & --               & --                                  \\
        \hline
    \end{tabular}
    \caption{Landscape of small Busy Beaver values.
        \cellcolor{green!20}Cells highlighted in green (that is, those strictly to the left of $S(6)$, $S(3,3)$, and $S(3,4)$) correspond to values for which we provide \Coq proofs. Bright green indicates the new results: $S(5)$ and $S(2,4)$, original to this work.
        All remaining highlighted cells are in orange and indicate the existence of a Cryptid (\ie machines whose halting problem is currently open and we believe to be mathematically hard; see Section~\ref{sec:intro:discuss} and Appendix~\ref{app:cryptids}). Lighter orange means that the existence of a Cryptid comes trivially from reusing a known 3-state 3-symbol Cryptid and ignoring the available additional state or symbol. Lower bounds (see Appendix~\ref{app:lowerbounds}), which come from exhibiting Turing machines achieving at least the given step-counts, are expressed using Knuth's ``up-arrow notation'', which is a way to express iterated exponentiation:
        $a \uparrow b = a^b$ is exponentiation,
        $a \uparrow\uparrow b = a^{a^{\dots^a}}$, called \textit{tetration}, is a tower of powers of $a$ with $b$ occurrences of $a$,
        and higher arrows indicate further levels of iteration \cite{Knuth1976Coping}; \eg three arrows is referred to as \textit{pentation}, so $2 \uparrow \uparrow \uparrow 5 = 2 \uparrow\uparrow \bigl( 2 \uparrow\uparrow ( 2 \uparrow\uparrow ( 2 \uparrow\uparrow 2 ) ) \bigr)$.
    }
    \label{table:landscape}
\end{table}

\vspace{-2ex}
\paragraph{Challenges.} Proving $S(5) = \BBtheFifth$ required analysing the behaviour of $\BBtheFifthTNF$ Turing machines\footnote{Knowing this exact number is only possible after the proof is done: it is as hard as computing $S(5)$; see Section~\ref{sec:enum}.} -- evidently requiring computer assistance. The challenge for analysing a halting machine occurs when it halts after a number of steps that is too enormous to be simulated step-by-step (e.g. the current 6-state champion halts after more than $2 \uparrow \uparrow \uparrow 5$ steps, see Appendix~\ref{app:lowerbounds}); this challenge was not encountered for 5-state machines since they halt in at most $\BBtheFifth$ steps -- this could not have been known for certain in advance but was believed -- which is easy to simulate on modern computers. The challenge for analysing a nonhalting machine is that proving that it does not halt can be hard. How hard?

Dauntingly, any $\Pi_1^0$ mathematical statement\footnote{A $\Pi_1^0$ statement is a statement of first-order logic of the form ``$\forall x, \phi(x)$'' where $\phi$ is a sentence using only bounded quantifiers, implying that for a given $x$, $\phi(x)$ can always be verified by a computer in finite time.} can be encoded as the halting problem of a Turing machine (from all-zero tape). Such statements are common in mathematics and include famous open problems such as Goldbach's conjecture and the Riemann hypothesis \cite{Davis-Matiyasevich-Robinson-1976}. Goldbach's conjecture, formulated in 1742, is one of the oldest open problems in mathematics and states that ``every even positive integer greater than 2 is the sum of two prime numbers''. We can build a Turing machine that halts iff the conjecture is false: by enumerating all even positive integers, and for each, testing all pairwise sums of smaller primes and halting iff we cannot express it as such a sum. A machine performing this procedure has been built using only 25 states, and the construction was formally verified using the Lean theorem prover \cite{GoldbachTM27, GoldbachTM25, DeMouraKongAvigadVanDoornvonRaumer}.

This means that proving the value of $S(25)$ is at least as hard as solving Goldbach's conjecture for two reasons: (i) assuming $S(25)$ is known, we could, impractically, simulate the machine for $S(25)$ steps to see if it has halted to settle the conjecture -- this is unrealistic because $S(25) > f^2_{\omega^2}(4 \uparrow \uparrow 341)$, where $f$ refers to the \textit{Fast Growing Hierarchy} \cite{wikiChampions,wainer1970classification} and (ii) intuitively, determining $S(25)$ requires arguments justifying the halting status of each 25-state machine, including this particular one. Similarly, the Riemann hypothesis has been encoded in a 744-state machine \cite{RiemannTM,Yedidia2016,BusyBeaverFrontier}. As few as 15 states are enough to encode a hard conjecture in number theory by Erd\H{o}s \cite{BB15}. Worse, the consistency of common axiomatic systems such as Peano Arithmetic (PA) or Zermelo–Fraenkel set theory (ZF) is also $\Pi_1^0$ since one can enumerate proofs in these systems until the proof of a contradiction is found. This has been done in practice for ZF, using $748$ states \cite{BB748Thesis,Yedidia2016,BusyBeaverFrontier}. By G\"odel's second incompleteness theorem, this means that proving the value of $S(748)$ cannot be done using ZF. This result has known further improvements to 636 states \cite{RidenourZF636}, and even 432 states \cite{wade2025alignment}, pending verification. Aaronson conjectures that as low as $S(20)$ cannot be proved in ZF and $S(10)$ cannot be proved in PA \cite{BusyBeaverFrontier}.

\enlargethispage{\baselineskip}
Hence, while 5-state halting machines were not feared, it remained unknown how hard proving 5-state machines nonhalting could get. This article settles the question: the smallest open problem in mathematics (on the Busy Beaver scale) does not arise from 5-state machines -- but we have good contenders among 6-state machines; see Section~\ref{sec:intro:discuss} and Cryptids (Appendix~\ref{app:cryptids}).

\paragraph{Related Work.} In 1983, Brady published the proof of $S(4) = \BBtheFourth$ \cite{Brady83}. One of the proof's main innovations was the introduction of a method to solve the halting problem of a category of machines the author calls \textit{Xmas Trees} based on a conversation he had with Shen Lin who had proved $S(3) = 21$ together with \rado. A caveat of the proof resides in the following quote from the paper: ``All of the remaining holdouts were examined by means of voluminous printouts of their histories along with some program extracted features. It was determined to the author's satisfaction that none of these machines will ever stop.'' A \textit{holdout} is a machine still needing a proof of halting/nonhalting. Using Coq, we bring additional confirmation that $S(4) = \BBtheFourth$; see Theorem~\ref{th:BB4}.

\newcommand{\SkeletHoldoutsSporadic}{\ts{XX}\xspace}

We know of two attempts at solving $S(5)$: in 2003, Georgiev (Skelet) published the program \texttt{bbfind} \cite{Skelet_bbfind} which enumerates and decides the halting behaviour of 5-state Turing machines, claiming to leave unsolved 43 holdouts.\footnote{\url{https://skelet.ludost.net/bb/nreg.html} and \url{https://bbchallenge.org/skelet}} Additionally, \texttt{bbfind} left no 4-state holdouts and computed $S(4) = 107$, agreeing with \cite{Brady83}. However, attesting to the validity of these results is difficult as Skelet's program consisted of about $6{,}000$ lines of undocumented Pascal code. That said, it turned out to be instrumental to solving $S(5)$, as \texttt{bbfind}'s ``Closed Position Set'' technique (see Section~\ref{sec:n-gramCPS}) was used, simplified, and improved in order to decide slightly more than 99.87\% of the 5-state Turing machines excluding loops (see Section~\ref{sec:loops}). Also, all of our \numSporadic \textit{Sporadic Machines}, \ie machines for which we needed individual proofs of nonhalting, were either among Skelet's 43 holdouts or claimed to have been manually solved by him -- Section~\ref{sec:sporadic} is dedicated to these longstanding holdouts. Some of Skelet's 43 holdouts were analysed by hand by Briggs starting in 2010 \cite{DanBriggs}. The second known attempt at solving $S(5)$ was in 2008 with Joachim Hertel's ``Symbolic induction prover'', which claimed to leave only $1{,}000$ holdouts, $900$ of which were manually proved not to halt, allegedly leaving only $100$ proper holdouts \cite{Hertel}. In contrast to Skelet's work, the method is documented but, to the best of our knowledge, neither the code nor the 900 claimed manual proofs are made available, making verification of the result difficult apart from attempting to reproduce it from scratch. Arguably, verification would still be tedious if the 900 manual proofs were given.

The \BBfull problem was also studied in models of computation other than \rado's (see Section~\ref{sec:TMs}), including (i) the \textit{quadruple} variation of Turing machines where each transition may move or write a new symbol, but not both \cite{Ross2003,Ross2005}; (ii) \textit{turmites}, which are Turing machines that operate in 2D \cite{BradyMeaningOfLife}; and (iii) lambda calculus \cite{tromp_oeis}. For additional historical perspective on the \BBfull problem, we refer the reader to Pascal Michel's survey and website \cite{michel2019busy,PMichel_website}.

\paragraph{Structure of the proof.} The proof of our main result, Theorem~\ref{th:BB5}, is given in Section~\ref{sec:results}. The structure of the proof is as follows: machines are enumerated arborescently in \textit{Tree Normal Form} (TNF) \cite{Brady64} -- which drastically reduces the search space's size: from $16{,}679{,}880{,}978{,}201$ 5-state machines to ``only'' $\BBtheFifthTNF$; see Section~\ref{sec:enum}. Each enumerated machine is fed through a \textit{pipeline} of proof techniques, mostly consisting of \textit{deciders}, which are algorithms trying to decide whether the machine halts or not. Because of the uncomputability of the halting problem, there is no \textit{universal} decider and all the craft resides in creating deciders able to decide large families of machines in reasonable time. Almost all of our deciders are instances of an abstract interpretation framework that we call \textit{Closed Tape Language} (CTL), which consists in approximating the set of configurations visited by a Turing machine with a more convenient superset, one that contains no halting configurations and is closed under Turing machine transitions (see Section~\ref{sec:deciders-overview}). The $S(5)$ pipeline is given in Table~\ref{tab:pipelineBB5} -- see Table~\ref{tab:pipelineBB2x4} for $S(2,4)$. All the deciders in this work were crafted by The bbchallenge Collaboration; see Section~\ref{sec:deciders}.

In the case of 5-state machines, \numSporadic \textit{Sporadic Machines} were not solved by deciders and required individual proofs of nonhalting, see Section~\ref{sec:sporadic}. These machines include surprising behaviours, such as eventually reaching an infinite loop after more than $5.41\times10^{51}$ steps of chaos (machine ``\href{https://bbchallenge.org/1RB1RD_1LC0RC_1RA1LD_0RE0LB_---1RC}{Skelet \#1}''), base-Fibonacci double counter (machine ``\href{https://bbchallenge.org/1RB0RA_0LC1RA_1RE1LD_1LC0LD_---0RB}{Skelet \#10}''), or obfuscated Gray code (machine ``\href{https://bbchallenge.org/1RB---_0LC1RE_0LD1LC_1RA1LB_0RB0RA}{Skelet \#17}'', \cite{xu2024skelet17fifthbusy}) and they are beautiful examples of \textit{algorithms in the wild}: non-human-engineered algorithms that, like deep sea life, were only found by means of exploration. In that spirit, a coarse \textit{zoology} of 5-state Turing machines is proposed in Section~\ref{sec:zoo}.

\newpage
\paragraph{Collaboratively solving the problem: bbchallenge.org.} In 2022, Stérin created ``The Busy Beaver Challenge'', \url{bbchallenge.org}, an online platform dedicated to collaboratively solving ``$S(5)=\BBtheFifth$'' \cite{sterin_2022_14955828}. Collaboration was motivated by the great amount of Turing machines to study in order to solve the problem. The \texttt{bbchallenge} platform essentially consists of the website, an instant chat \textit{Discord} server\footnote{\url{https://discord.gg/wuZhtTvYU3}}, and a wiki.\footnote{\url{wiki.bbchallenge.org}} The website serves as an entry point to the problem and Turing machine visualisation tools both for studying purposes and for piquing the curiosity of visitors with ``eye candy''. The website also provided a browsable \textit{seed database}\footnote{\url{https://github.com/bbchallenge/bbchallenge-seed}} containing a sufficient set of 5-state Turing machines to prove nonhalting in order to solve $S(5)$. Using this database, the task of contributors was to design deciders (see above). For both the seed database and deciders, trust in the results required a strict validation process: (i) any algorithm had to be reproduced at least once by an independent contributor, with matching results; (ii) a proof of correctness had to be provided (in natural language, as in a regular mathematical article). The use of proof assistants, such as Coq (see after), was merely dreamed of when the project started.

With the surprise release of \CoqBB in the spring of 2024 (see after), both the seed database and the validation process described above were made obsolete (because both the enumeration of machines and the verification of deciders were performed directly by \Coq), but almost all the collaborative work performed on \url{bbchallenge.org} during these 2 years was embedded in the Coq proof -- \CoqBB also contains many original innovations. Also, although now obsolete, the seed database provided during these 2 years a clear indicator of progress with the number of its machines remaining to be decided, which clearly stimulated collaboration. Deciders that were developed by The bbchallenge Collaboration but that were not used in the Coq proof\footnote{Or, in the case of FAR (Section~\ref{sec:FAR}), only partially used.} have been described in \cite{bbchallenge_part1}.

The bbchallenge Collaboration roughly comprises all who participated in the discussions across all our channels (Discord chat, forum, wiki, GitHub, emails), who total hundreds of people; collaborators whose contributions were key to solving $S(5)$ co-author this work, and we acknowledge many others; see Appendix~\ref{app:contribs}.  Some of them are anonymous. We believe that welcoming anonymity played a part in making bbchallenge a place where contributors felt at ease. Most contributors have no academic affiliation and have software engineering related positions or are students. The community seems to be relatively balanced between three geographical zones: North America, Europe, and Asia. A majority of contributors seem to be below the age of 30 but the 30+ age bracket is also well represented. Most contributors never met \textit{In Real Life}. Similarly to the ``build in public'' philosophy in software, our research happened in public with no withholding of information, enabling full reproducibility of the results. Motivated newcomers were able to build on the existing results, finding where they wanted to contribute. In contrast, many newcomers also reported being overwhelmed by the entropic nature of our collaboration.

The bbchallenge Collaboration has a hub-and-spoke structure: typically, single contributors or small subsets of contributors made discoveries (mainly, new deciders) and shared their results on the bbchallenge platform (mainly, on our Discord server). Over the span of two years, collaborators organically joined the project and contributed to deciders: more than 20 independent GitHub repositories of deciders\footnote{Some are listed here: \url{https://wiki.bbchallenge.org/wiki/Code_repositories}.} were shared on the bbchallenge platform, spanning a vast range of languages -- C, C++, Go, Rust, Haskell, \Coq, Dafny, Lean, Python, PHP, etc. One core principle of our collaboration was to welcome contributions in any programming language or technology. That way, the use of \Coq to solve $S(5)$ was by no means imposed but came because of the taste and experience of the collaborators who mainly pushed the formal verification effort: mxdys (\CoqBB) and Kądziołka (\texttt{busycoq} \cite{busycoq}).

As mentioned, most of the collaboration was \textit{discussion-driven}, happening continuously, day and night, on our rather entropic Discord server. At the time of this writing, the Discord server has about $1{,}300$ members, of which about 100 are active monthly, and about $135{,}000$ messages have been exchanged by about 400 people in total since launch in March 2022. Keeping track of the knowledge produced by the collaboration was a challenge and required dedicated ``research maintainers'' whose responsibilities very much resembled those of open-source software project maintainers. We did not have to deal with ``trolls'' and little moderation was necessary on our channels. Our website \href{https://bbchallenge.or}{bbchallenge.org} has had more than 45,000 unique visitors since launch and currently an average of 60 unique visitors per day, with main historical spikes of traffic generated by Hacker News \cite{bmc7505_2023_busybeaver}, Quanta Magazine's article \cite{brubaker2024_bb5_quanta} and YouTube video \cite{quanta2025_busybeaver_video} about our project, as well as references from blogs\footnote{Mainly \url{https://scottaaronson.blog/} and \url{https://sligocki.com}.} and other news reports in national journals (France \cite{larousserie2024_castor}, Austria \cite{derstandard2025_fleissiger_biber}).

\paragraph{Proof assistants, \Coq, and \CoqBB.} Proof assistants are software tools able to express and verify formal proofs; some popular proof assistants are: Agda, Coq\footnote{Now renamed the Rocq Prover: \url{https://rocq-prover.org/about\#Name}}, F*, HOL (HOL4/HOL Light), Isabelle, Lean, Metamath. Milestones in the use of proof assistants include: the Coq proofs of the four-color theorem \cite{gonthier2008formal, gonthier2023computer} and Feit-Thompson theorem on odd orders finite groups \cite{gonthier2013machine}; the HOL Light proof of the Kepler Conjecture \cite{3fdf2c9143a54629bba02e34f24c7698}; the liquid tensor experiment in Lean \cite{2309.14870,LeanCommunityLTE2022}. Proof assistants often come with a fully-fledged programming language allowing to run computations within proofs and/or to write proofs about computer programs (such as correctness proofs), which have both been extensively used in \CoqBB. For any proof assistant, trust in its implementation is required in order to accept as true the results it verifies. The open-source nature of most proof assistants facilitates bug discovery and resolution.

In our case, with hundreds of millions of Turing machines to study, the use of a proof assistant immensely facilitates scientific consensus on the correctness of the proof: for instance, we are assured that no Turing machine was forgotten in the study and that our proofs of nonhalting are correct. Also, the 13 individual proofs for Sporadic Machines contain technical and error-prone arguments that would be harder to verify and trust if not formalised -- \eg a standalone article was dedicated to machine ``Skelet \#17'' \cite{xu2024skelet17fifthbusy} and its \Coq proof is almost $7{,}000$ lines long, see Section~\ref{sec:sporadic}. As mentioned above, \texttt{bbchallenge} was originally not a formal verification project: formal verification happened as an unpredicted event which started as early as 2022 when Fenner verified some deciders using Dafny \cite{dafny_fenner, Dafny}. Not long after, Kądziołka started verifying deciders in Coq as well as providing, together with Yuen, 12 out of the 13 individual proofs of nonhalting for Sporadic Machines; see  \texttt{busycoq} \cite{busycoq}, which were reused in \CoqBB.

\CoqBB, is written in Coq \cite{the_coq_development_team_2024_14542673}, which is a proof assistant and programming language based on the Calculus of Inductive Constructions \cite{CoC} whose development started in 1989. Distinctively, \CoqBB's objects of study are Turing machines (\ie algorithms) instead of more traditional mathematical objects. \CoqBB is available at \url{https://github.com/ccz181078/Coq-BB5} \cite{mxdys_2025_17061968}. \CoqBB was released\footnote{\CoqBB was released in four stages: (i) \href{https://discuss.bbchallenge.org/t/proving-bb-5-in-coq/225}{82 holdouts}; (ii) 1 holdout, \href{https://bbchallenge.org/1RB---_0LC1RE_0LD1LC_1RA1LB_0RB0RA}{Skelet \#17}; (iii) \href{https://discuss.bbchallenge.org/t/july-2nd-2024-we-have-proved-bb-5-47-176-870/237}{full proof}; (iv) \href{https://github.com/ccz181078/Coq-BB5}{optimised proof}.} in the spring of 2024 by contributor mxdys who embedded two years of work done by the bbchallenge collaboration as well as improving and introducing new deciders to finish the proof: \CoqBB is not an \textit{a posteriori} formalisation of existing work and contains many original innovations without which even a nonformalised $S(5)$ proof would have been a lot more complex. \CoqBB totals $27{,}274$ lines of \Coq and $638$ lemmas; plus an additional $10{,}553$ lines of \Coq and $319$ lemmas including imported \texttt{busycoq} proofs.

\CoqBB is one of the most compute-intensive formal proof to date as it implements both the TNF enumeration of the \BBtheFifthTNF\ Turing machines and the deciders (see proof structure above) directly in Coq. The algorithms are proven correct, and the proof runs them and uses their verified outputs to obtain the result. Proofs such as \CoqBB, which rely on computing algorithmic outputs \cite{vmcompute,nativecompute}, are known as ``proofs by reflection'' \cite{boutin1997using}, and the \Coq proof of the four-color theorem also fits in that category \cite{gonthier2010introduction}. As main output of The Busy Beaver Challenge, the list\footnote{Available at \url{https://docs.bbchallenge.org/CoqBB5_release_v1.0.0/}} of all the Turing machines enumerated by the \Coq proof, together with their halting status and decider, were \textit{extracted} from the proof using \Coq's OCaml extraction capabilities; see Section~\ref{sec:results}.

This comprehensive use of \Coq in order to solve $S(5)$, where the TNF enumeration of the $\BBtheFifthTNF$ Turing machines itself is implemented and run in \Coq, positively shocked the community. Indeed, before \CoqBB, consensus within The bbchallenge Collaboration  was that formally verifying the TNF enumeration was arduous and running it too computationally intensive to be implemented within a proof assistant; at best, the belief was that formal verification efforts would rely on an external, trusted enumeration, such as bbchallenge's seed database (see above).

\enlargethispage{\baselineskip}
\CoqBB initially compiled in about 13 hours on a standard laptop, but the use of \Coq's faster computing engine \texttt{native\_compute} \cite{nativecompute} and parallelisation of the proof (see Section~\ref{sec:enum}) brought compile time down to about 45 minutes on 13 cores. \CoqBB strongly benefited from Coq’s fast computing abilities, currently arguably more developed than those of any other proof assistant. Trusting \Coq, the only elements to check in order to trust \CoqBB's results are the main theorem statement and the definitions it uses, which have been isolated in an individual file, \texttt{BB5\_Statement.v}, which is documented with comments and only 121 lines long and requires no \Coq expertise to be read. Since its release, \Coq experts joined our team (see Appendix~\ref{app:contribs}) to review and validate it as well as ruling out the possibility that the proof would try to exploit a potential \Coq bug to falsely claim the results. Finally, after compilation, the proof prints the only axiom that it uses, called \texttt{functional\_extensionality\_dep} from \Coq's standard library, which claims that two functions are equal if they are equal at all points. This axiom is widely accepted, consistent in \Coq, and true in common set-theoretic foundations of mathematics.\footnote{It could be removed, at the cost of unnecessarily complicating the proof.}

\subsection{Discussion}\label{sec:intro:discuss}

\paragraph{Cryptids.} The Busy Beaver game is a concrete attempt at identifying the frontier between the knowable and the unknowable: what is the highest $n$ for which we can prove the value of $S(n)$? Knowing that for $n > 5$, as mentioned in Section~\ref{sec:intro:mainresults}, the proof may involve solving arbitrarily difficult problems or worse, simply be outside of PA or ZF. The Busy Beaver game is great at generating open problems in mathematics!

For instance, in all the unsolved Turing machine classes (orange highlight in Table~\ref{table:landscape}), we have found what we call ``Cryptids'', which are loosely defined as Turing machines whose halting problem from the all-zero tape is believed to be mathematically hard (Appendix~\ref{app:cryptids}). For instance, with 6 states, Antihydra (\tm{1RB1RA_0LC1LE_1LD1LC_1LA0LB_1LF1RE_---0RA}; see Turing machine notation, Section~\ref{sec:TMs}) is a machine that does not halt from the all-zero tape if and only if the following Collatz-reminiscent conjecture (which was discovered from analysing the machine) holds:

\begin{conjecture}[Antihydra does not halt]
    Consider the Collatz-like map $H: \mathbb{N} \to \mathbb{N}$ defined by $H(x) = 3\frac{x}{2}$ if $x$ is even and $H(x) = \frac{3x-1}{2}$ if $x$ is odd. Iterating $H$ from $x=8$, there are never (strictly) more than twice as many odd numbers as even numbers.
\end{conjecture}

Using Conway's terminology, the conjecture \textit{probviously} (portmanteau of the words probabilistic and obvious) holds because a probabilistic analysis\footnote{See \url{https://wiki.bbchallenge.org/w/index.php?title=Antihydra}.} of the Turing machine suggests that its probability of ever halting is minuscule, namely, $\smash{\left(\frac{\sqrt{5}-1}{2}\right)}^{1073720885}\approx 4.84\times 10^{-224394395}$. However, properly proving that Antihydra does not halt is believed mathematically hard, because of resemblance to the notoriously hard Collatz conjecture \cite{LagariasCollatz}: Antihydra is a Cryptid. Also, the map $H$ was studied before Antihydra was discovered and is known to not generate Sturmian words \cite{DUBICKAS_2009}.

Antihydra is the \textit{smallest}\footnote{We refrained from saying the \textit{simplest}, since there are certainly machines with simpler, yet open, halting problems.} open problem in mathematics, on the Busy Beaver scale. Or, to be exact, one of the smallest, since a dozen other 6-state Cryptids have been identified to date\footnote{See \url{https://wiki.bbchallenge.org/wiki/BB(6)\#Cryptids}}, such as the following, jokingly called, \href{https://wiki.bbchallenge.org/wiki/Beaver_Math_Olympiad}{Beaver Math Olympiad} (BMO) problem:

\begin{problem}[BMO Problem 1]
Let $(a_n)_{n \ge 1}$ and $(b_n)_{n \ge 1}$ be two sequences such that $(a_1, b_1) = (1, 2)$ and
$$(a_{n+1}, b_{n+1}) = \begin{cases}
        (a_n-b_n, 4b_n+2) & \text{if }a_n \ge b_n \\
        (2a_n+1, b_n-a_n) & \text{if }a_n < b_n
    \end{cases}$$

\noindent for all positive integers $n$. Does there exist a positive integer $i$ such that $a_i = b_i$?
\end{problem}

This problem is a reformulation of  ``does \tm{1RB1RE_1LC0RA_0RD1LB_---1RC_1LF1RE_0LB0LE} halt?'' -- the machine halts if there is $i$ such that $a_i = b_i$. Similarly to Antihydra, the machine is probviously nonhalting\footnote{See analysis: {\scriptsize \url{https://wiki.bbchallenge.org/wiki/1RB1RE_1LC0RA_0RD1LB_---1RC_1LF1RE_0LB0LE}}. }, but nonetheless, the problem is still open.

We also know of some probviously halting Cryptids, such as the 3-state 3-symbol machine \tm{1RB2LC1RC_2LC---2RB_2LA0LB0RA}, which probabilistically  has a 100\% chance of halting\footnote{See analysis: {\scriptsize \url{https://wiki.bbchallenge.org/wiki/1RB2LC1RC_2LC---2RB_2LA0LB0RA}}.} and would become the new 3-state 3-symbol champion (considerably extending the current $10^{17}$ bound) if proved to halt. Other Cryptids are neither probviously halting or nonhalting, for instance \tm{1RB1LE_0LC0LB_1RD1LC_1RD1RA_1RF0LA_---1RE} is estimated to have a 3/5 chance of nonhalting and a 2/5 chance of halting.\footnote{See analysis: {\scriptsize \url{https://wiki.bbchallenge.org/wiki/1RB1LE_0LC0LB_1RD1LC_1RD1RA_1RF0LA_---1RE}}.} Note that all these \textit{probvious} arguments are open to dispute by nature: for instance, bad probabilistic models can predict halting with probability 1 of machines that were proved nonhalting.\footnote{Such as for this machine: \url{https://wiki.bbchallenge.org/wiki/1RB0LE_1LC1RA_---1LD_0RB1LF_1RD1LA_0LA0RD}.} In hindsight, it is surprising (and lucky!) that there are no 5-state Cryptids.

Cryptids illustrate the ease with which the Busy Beaver game generates small, non-trivial, open mathematical problems, challenging our intelligence and the limits of mathematical knowledge. The open problems generated by our work have been included in a dataset of Lean-formalised conjectures, designed to serve as future challenges for AI reasoning tools \cite{google-deepmind-formal-conjectures}. More generally, given the sheer amount of problems, with broad variety of difficulty (from very easy to near-impossible), that the Busy Beaver game can generate, the ability to prove known Busy Beaver values or to make progress on unknown ones would be an ambitious benchmark for machine intelligence. \CoqBB has already started being used with that goal in mind: AI systems have been tested on proving its first 100 lemmas of the $S(4)$ proof, with 58\% success so far~\cite{teodorescu2024nlir}.

\paragraph{Trends in collaborative research.} Massively collaborative online research in mathematics and theoretical computer science is a relatively new phenomenon, pioneered by the Polymath Project, which collaboratively solved several problems in mathematics \cite{Gowers2009}. A few differences with \texttt{bbchallenge} come to mind such as (i) our use of instant messaging instead of blog comments as main communication medium and (ii) the non-academic affiliation of most of our contributors -- see Section~\ref{sec:intro:mainresults} -- but the essential philosophy of leveraging collective, distributed, intelligence to solve complex problems is the same. Online communities similar to \texttt{bbchallenge} in structure and size at the time of this writing include \textit{ConwayLife}\footnote{\url{https://conwaylife.com/}}, which performs research on John Conway's Game of Life (and other cellular automata) and \textit{Googology}\footnote{\url{https://googology.fandom.com/wiki/Googology_Wiki} and \url{https://googology.miraheze.org/wiki/Main_Page}}, which performs research on large numbers.

In many ways, \texttt{bbchallenge}'s effort is closer to experimental open-source software project development than to traditional research in mathematics or theoretical computer science: we mainly develop algorithms (deciders), and as a consequence, \CoqBB itself is essentially a collection of algorithms, proved correct.

The current increasing popularity of proof assistants such as \Coq (now renamed the Rocq Prover) or Lean within the mathematical and computer science communities naturally accelerates this convergence between collaborative research and open source software development: researchers can now leverage the same tools and processes as developers (version control, \textit{pull requests}, \textit{issues}, etc.) to massively collaborate on proofs in a scalable way -- \ie not requiring humans to check or trust proofs.

As a concrete example, shortly after we announced our $S(5)$ result, Terence Tao launched a collaborative pilot project in universal algebra, called the ``Equational Theories Project'' (ETP), requiring the proof or disproof of $22{,}028{,}942$ implications. The ETP leveraged GitHub and Lean as their means of collaboration and a \textit{Zulip} chat server for communication \cite{TaoBlog, swh-dir-426b52b, ETPpaper}. In contrast with \texttt{bbchallenge} which is a rather baroque assembly of many technologies with a ``late'' use of theorem provers, the ETP focused efforts on formal verification and Lean \textit{by design} and from the start. The project was extremely successful, attracted more than 50 contributors, and was completed in only a few months.\footnote{See Tao's personal log: \url{https://github.com/teorth/equational_theories/wiki/Terence-Tao's-personal-log}.} In this context, AI seems to have a bright future: either as an enhanced project knowledge base or as a collaborator itself \cite{Trinh2024, wu2024internlm25stepproveradvancingautomatedtheorem}; although its use in projects such as the ETP had limited success~\cite{ETPpaper}.

\subsection{Future Work}\label{sec:fw}

Solving $S(5)$ does not solve all questions about $5$-state Turing machines, for instance we are interested in the following problems: (i) characterising all 5-state counters (see Section~\ref{sec:zoo}); (ii) finding the biggest loop among 5-state machines with no halting transitions\footnote{The TNF enumeration discards machines with no halting transitions as they are obviously nonhalting; see Section~\ref{sec:enum}.} -- we already know of some that are way bigger than Sporadic Machine ``Skelet \#1'' \cite{ligocki2022motherofgiants}; (iii) slightly less related but, is there a universal Turing machine with 5 states (which also brings the question of studying 5-state machines on other tapes than all-zero)?

Progress is ongoing in all unsolved Turing machine classes (orange highlight in Table~\ref{table:landscape}): new deciders are being developed to tackle $S(3,3)$, $S(2,5)$ and $S(6)$, including a generalisation of all the regular CTL deciders presented in this paper (see Section~\ref{sec:deciders}). Most of these new deciders have been formalised using \Coq.\footnote{See: \url{https://github.com/ccz181078/busycoq/tree/BB6/verify}} There currently remain only $4$ holdouts for $S(3,3)$\footnote{\url{https://wiki.bbchallenge.org/wiki/BB(3,3)\#Holdouts}}, including the probviously halting suspected new champion (see Section~\ref{sec:intro:discuss}), and only $60$ holdouts for $S(2,5)$.\footnote{\url{https://wiki.bbchallenge.org/wiki/BB(2,5)\#Holdouts}}

Concerning $S(6)$, the TNF enumeration of 6-state machines, which contains about 33 billion machines, has also been partially implemented in \Coq, and, together with the deciders and about $2{,}000$ individual proofs of nonhalting, ``only'' about $\numBBsholdouts$ holdouts remain to date. Importantly, among these holdouts we have:
\begin{itemize}[label=--]
    \item Antihydra and other Cryptids (see Section~\ref{sec:intro:discuss} and Appendix~\ref{app:cryptids}), which are, most likely, extremely hard problems to solve.
    \item The possible existence of a halting machine that exceeds the current $2\uparrow\uparrow\uparrow5$ champion (see Table~\ref{table:landscape} and Appendix~\ref{app:lowerbounds}), which could require significant analysis or accelerated simulator improvements to be detected.
\end{itemize}

Hence, in perpetuating a longstanding tradition of hope about Busy Beaver values, we predict that $S(6)$ will never be proved.\footnote{Nonetheless, \texttt{bbchallenge.org} welcomes all new contributors interested in the Busy Beaver game!}

\newpage
\section{Turing machines}\label{sec:TMs}

\vspace{-1.4em}

\begin{figure}[ht]
    \centering
    \renewcommand{\arraystretch}{1.3}
    \setlength{\tabcolsep}{6pt}

    \begin{subfigure}[b]{0.28\textwidth}
        \centering
        \begin{tabular}{ccc}
            \toprule
                    & \textbf{0} & \textbf{1} \\
            \midrule
            \stateA & 1R\stateB  & 1L\stateC  \\
            \stateB & 1R\stateC  & 1R\stateB  \\
            \stateC & 1R\stateD  & 0L\stateE  \\
            \stateD & 1L\stateA  & 1L\stateD  \\
            \stateE & ---        & 0L\stateA  \\
            \bottomrule
        \end{tabular}
        \caption{5-state 2-symbol \BBfull winner. This machine was discovered by Marxen and Buntrock in 1989 \cite{Marxen_1990}.}
        \label{table:bb5}
    \end{subfigure}
    \hfill
    \begin{subfigure}[b]{0.31\textwidth}
        \centering
        \includegraphics[width=0.60\linewidth]{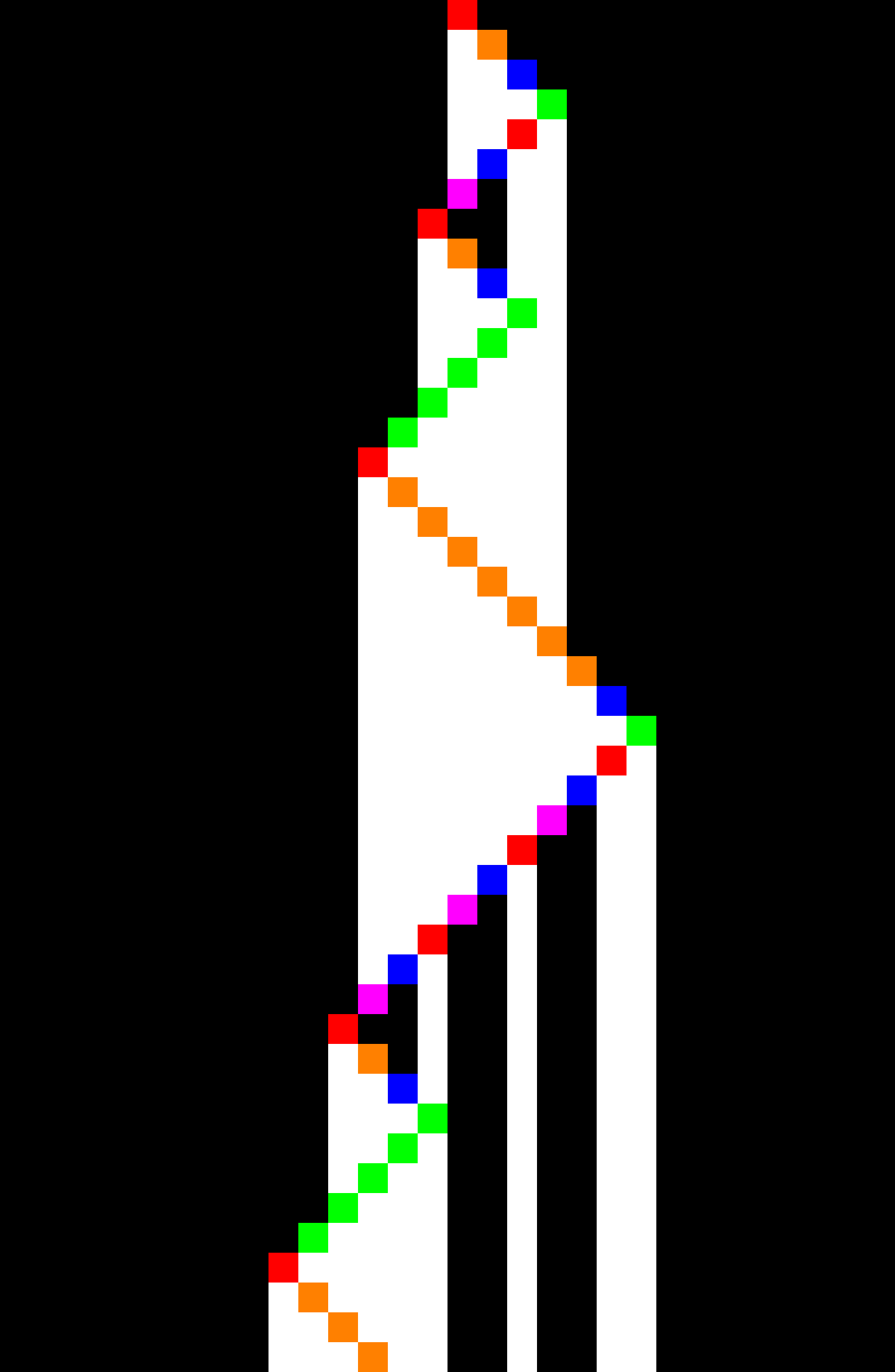}
        \caption{45-step space-time diagram of the 5-state winner. Head position is coloured to indicate state, see (a).}
        \label{fig:bb5-diagram}
    \end{subfigure}
    \hfill
    \begin{subfigure}[b]{0.3\textwidth}
        \centering
        \includegraphics[width=0.8\linewidth]{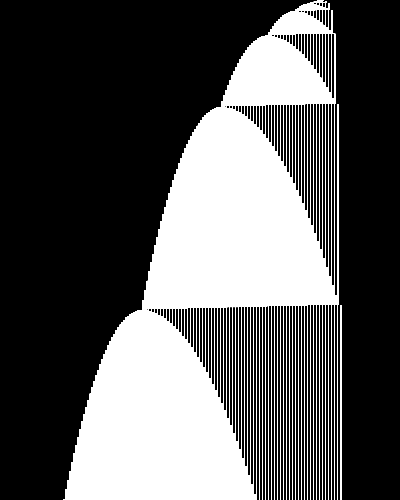}
        \caption{20,000-step space-time diagram of the 5-state winner.}\label{fig:bb5-diagram-zoomout}
    \end{subfigure}

    \caption{Transition table and space-time diagrams of the 5-state 2-symbol \BBfull winner, which halts after 47,176,870 steps. See
        \tm{1RB1LC_1RC1RB_1RD0LE_1LA1LD_---0LA}.}\label{fig:bb5win}
\end{figure}

In this work, $\N = \{0,1,\dots\}$ and $\N^+ = \{1,2,3\dots\}$.

We consider Turing machines that use a single, discrete, bi-infinite tape, \ie the tape can be thought as a function $\tau: \mathbb{Z} \to \alphabet$, where $\alphabet$ is the alphabet of symbols used by the machine. Machine transitions are either undefined (the machine halts if it ever reaches an undefined transition) or given by (i) a symbol of $\alphabet$ to write; (ii) a direction to move (right or left); and (iii) a state to go to. More precisely, the transition table of a Turing machine is a partially defined function $\delta: \states \times \alphabet \partialto \alphabet \times \{\text{L},\text{R}\} \times \states $, with $\states$ the set of states, \eg $\{\stateA,\stateB,\stateC,\stateD,\stateE\}$ for 5-state machines. Figure~\ref{fig:bb5win}(a) gives the transition table of the 5-state 2-symbol \BBfull winner. The machine halts after 47,176,870 steps (starting from all-0 tape) when it reads a \szero in state \stateE for the first time (undefined transition). Allowing for undefined transitions is a small, consequenceless but useful (see Section~\ref{sec:enum}) deviation from \rado's original setup.

In the \BBfull context, machines are always executed from the all-0 tape and starting in state~\stateA. Execution goes as follows: at each step, the machine which is in state $s$ looks at which symbol $\sigma$ is present on the tape cell the head is currently on and then, if defined, executes the instruction given by its transition table, \eg $\delta(s, \sigma) = \szero\text{L}\text{\stateE}$ means that the machine will write a \szero, move the head one cell to the left and switch to state \stateE. If $\delta(s, \sigma)$ is not defined, the machine halts.

A \textit{configuration} (also known as \textit{execution state}) of a Turing machine is defined by the 3-tuple: (i) state; (ii) position of the head on the tape; (iii) content of the tape. As mentioned above, here, \textit{the initial configuration} of a machine is always (i) state is A, i.e. the first state to appear in the machine's description; (ii) head's position is 0; (iii) the initial tape is all-0 -- i.e. each tape cell is containing 0. We write $c_1 \TMstep_\mathcal{M} c_2$ if a configuration $c_2$ is obtained from $c_1$ in one computation step of machine $\mathcal{M}$. We omit $\mathcal{M}$ if it is clear from context. We let $c_1 \TMstep^s c_2$ denote a sequence of $s$ computation steps, and let $c_1 \TMstep^* c_2$ denote zero or more computation steps. %

\vspace{-1ex}
\paragraph{Halting and step count convention.} Halting happens when the machine attempts to run an undefined transition. We write $c \TMstep \bot$ to signify that the machine halts after attempting to run one step from configuration $c$. The number of steps $s\in \N^+$ run by a halting Turing machine includes the final halting step, \eg $s = 1$ for a machine where $\delta(\stateA,\symbolzero)$ is not defined.

\vspace{-1ex}
\paragraph*{Turing machine format.} We often communicate Turing machines using the following linear format: \\ \verb|1RB1LC_1RC1RB_1RD0LE_1LA1LD_---0LA| represents the transition table of Figure~\ref{fig:bb5win}(a), where \texttt{\_} is used to separate states and transitions are given in read-symbol order. Note that, historically, the undefined transition reached by a halting machine was represented using \texttt{1RZ}, hence our format allows the use of any letter outside of the state space to represent halting, \eg \texttt{1RB1LC\_1RC1RB\_1RD0LE\_1LA1LD\_1RZ0LA}, the use of \texttt{1RZ} instead of \verb|---| means that \textit{we know} that the transition is reached and, thus, that the machine halts. Multi-symbol machines are represented in the same way, \eg the 2-state 4-symbol \BBfull winner is \verb|1RB2LA1RA1RA_1LB1LA3RB---| (also given by \texttt{1RB2LA1RA1RA\_1LB1LA3RB1RZ}); see Figure~\ref{fig:bb2x4}. This format can be used as URL on \url{bbchallenge.org} to display space-time diagrams and known information about the machine, e.g. \url{https://bbchallenge.org/1RB1LC\_1RC1RB\_1RD0LE\_1LA1LD\_---0LA}.

\vspace{-1ex}
\paragraph*{Space-time diagrams.} We use space-time diagrams to give a visual representation of the behaviour of a given machine. The space-time diagram of machine $\mathcal{M}$ is an image where the $i^\text{th}$ row of the image gives:
\begin{enumerate}
    \item The content of the tape after $i$ steps (for 2-symbol machines, black is 0 and white is 1, while for $n$ symbols, black is 0, white is symbol $n-1$ and linear grey-scaling is used in between, \eg~Figure~\ref{fig:bb2x4}).
    \item The position of the head is coloured to give state information using the following colours for 5-state machines: \textcolor{colorA}{A},  \textcolor{colorB}{B},  \textcolor{colorC}{C},  \textcolor{colorD}{D},  \textcolor{colorE}{E} (one has to look at the row above to deduce what symbol the head is reading, unless it is the initial row, where a \szero is read).
\end{enumerate}

Figure~\ref{fig:bb5win}(b) gives a 45-step space-time diagram for the 5-state 2-symbol \BBfull winner. We often use \textit{zoomed-out} space-time diagrams without state-coloring information, such as Figure~\ref{fig:bb5win}(c), which gives the first 20,000 steps of the 5-state 2-symbol \BBfull winner. Zoomed-out space-time diagrams depicted in this work use a tape of 400 cells unless stated otherwise; the initial cell is generally at the center of the tape but sometimes offset to the right or left. Figure~\ref{fig:bb2x4} showcases the 2-state 4-symbol \BBfull winner.

\begin{figure}[ht]
    \centering
    \renewcommand{\arraystretch}{1.3}
    \setlength{\tabcolsep}{6pt}

    \begin{subfigure}[b]{0.28\textwidth}
        \centering
        \scalebox{0.93}{
            \begin{tabular}{ccccc}
                \toprule
                        & \textbf{0}~\tikz\fill[black, draw=black] (0,0) rectangle (0.20,0.20);     & \textbf{1}~\tikz\fill[graySymb1, draw=black] (0,0) rectangle (0.20,0.20);

                        & \textbf{2}~\tikz\fill[graySymb2, draw=black] (0,0) rectangle (0.20,0.20); & \textbf{3}~\tikz\fill[white, draw=black] (0,0) rectangle (0.20,0.20);                             \\
                \midrule
                \stateA & 1R\stateB                                                                 & 2L\stateA                                                                 & 1R\stateA & 1R\stateA \\
                \stateB & 1L\stateB                                                                 & 1L\stateA                                                                 & 3R\stateB & ---       \\
                \bottomrule
            \end{tabular}
        }
        \caption{Transition table of the 2-state 4-symbol \BBfull winner found by Ligocki and Ligocki in 2005 \cite{PMichel_website}.}
        \label{table:bb2x4}
    \end{subfigure}
    \hfill
    \begin{subfigure}[b]{0.31\textwidth}
        \centering
        \includegraphics[width=0.60\linewidth]{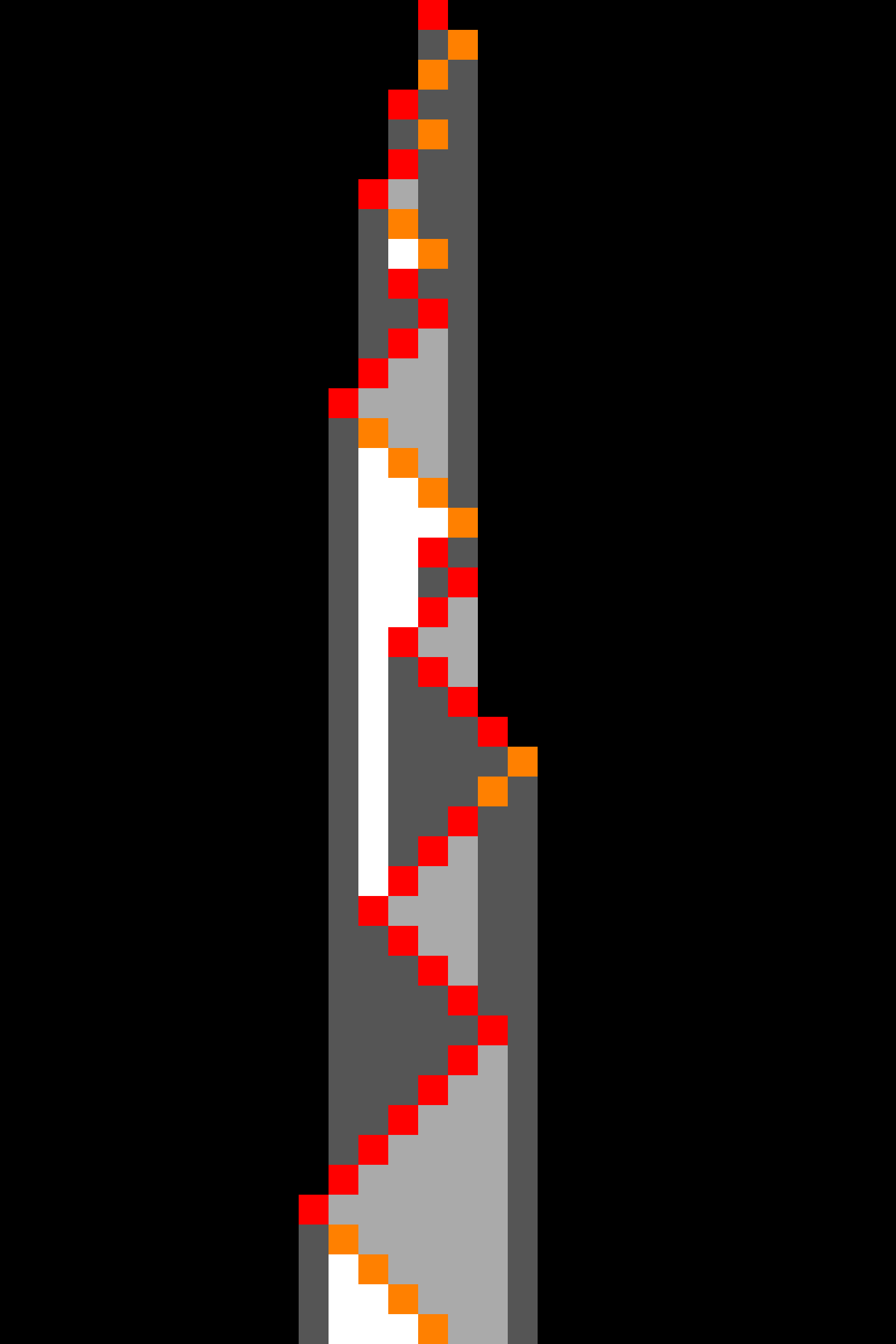}
        \caption{45-step space-time diagram of the 2-state 4-symbol. Head position is coloured to indicate state, see (a).}
        \label{fig:bb2x4-diagram}
    \end{subfigure}
    \hfill
    \begin{subfigure}[b]{0.32\textwidth}
        \centering
        \includegraphics[width=0.8\linewidth]{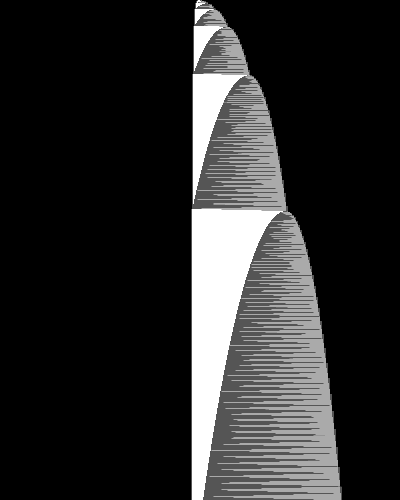}
        \caption{20,000-step space-time diagram of the 2-state 4-symbol winner.}
        \label{fig:bb2x4-diagram-zoomout}
    \end{subfigure}

    \caption{Transition table and space-time diagrams of the 2-state 4-symbol \BBfull winner, which halts after $\BBTxF$ steps. See
        \tm{1RB2LA1RA1RA\_1LB1LA3RB---}.}
    \label{fig:bb2x4}
\end{figure}

\newpage
\section{Enumerating Turing machines in Tree Normal Form (TNF)}\label{sec:enum}
\vspace{-10pt}
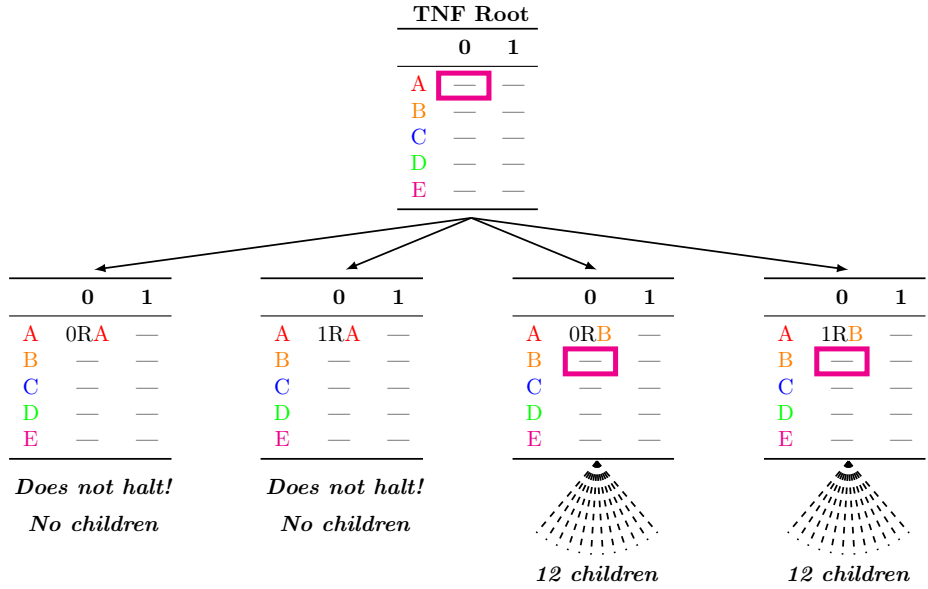
\begin{figure}[ht]
    \centering
    \resizebox{0.8\textwidth}{!}{ %
        \begin{tikzpicture}[
                level distance=45mm, %
                sibling distance=70mm, %
                every node/.style={align=center},
                edge from parent/.style={draw, -latex, thick}
            ]

            \newcommand{\turingTable}[2]{
                \adjustbox{valign=t}{
                    \begin{tabular}{ccc}
                        \toprule
                                            & \textbf{0} & \textbf{1} \\
                        \midrule
                        {\color{red} A}     & #1         & #2         \\
                        {\color{orange} B}  & ---        & ---        \\
                        {\color{blue} C}    & ---        & ---        \\
                        {\color{green} D}   & ---        & ---        \\
                        {\color{magenta} E} & ---        & ---        \\
                        \bottomrule
                    \end{tabular}
                }
            }

            \node (root) at (0,0) {\turingTable{---}{---}};

            \node (child1) at (-6,-4) {\turingTable{0R\stateA}{---}};
            \node (child2) at (-2,-4) {\turingTable{1R\stateA}{---}};
            \node (child3) at (2,-4) {\turingTable{0R\stateB}{---}};
            \node (child4) at (6,-4) {\turingTable{1R\stateB}{---}};

            \draw[-latex, thick] (root.south) -- (child1.north);
            \draw[-latex, thick] (root.south) -- (child2.north);
            \draw[-latex, thick] (root.south) -- (child3.north);
            \draw[-latex, thick] (root.south) -- (child4.north);

            \draw[magenta, thick, line width=0.8mm]
            ($(root.north west) + (0.9cm,-1.25cm)$)
            rectangle
            ($(root.north west) + (1.7cm,-0.85cm)$);

            \draw[magenta, thick, line width=0.8mm]
            ($(child3.north west) + (1.08cm,-1.67cm)$)
            rectangle
            ($(child3.north west) + (1.88cm,-1.27cm)$);

            \draw[magenta, thick, line width=0.8mm]
            ($(child4.north west) + (1.08cm,-1.67cm)$)
            rectangle
            ($(child4.north west) + (1.88cm,-1.27cm)$);

            \node at ($(root.north) + (0,0.1cm)$) {\textbf{TNF Root}};

            \node at ($(child1.south) + (0,-0.3cm)$) {\textbf{\textit{Does not halt!}}};
            \node at ($(child2.south) + (0,-0.3cm)$) {\textbf{\textit{Does not halt!}}};

            \foreach \i in {230.0,237.27,244.55,251.82,259.09,266.36,273.64,280.91,288.18,295.45,302.73,310.0} {
                    \draw[dashed, thick] ($(child3.south) + (0,0.1cm)$) -- ++(\i:1.5cm);
                }

            \foreach \i in {230.0,237.27,244.55,251.82,259.09,266.36,273.64,280.91,288.18,295.45,302.73,310.0} {
                    \draw[dashed, thick] ($(child4.south) + (0,0.1cm)$) -- ++(\i:1.5cm);
                }

            \node at ($(child3.south) + (0,-1.7cm)$) {\textbf{\textit{12 children}}};
            \node at ($(child4.south) + (0,-1.7cm)$) {\textbf{\textit{12 children}}};

            \node at ($(child1.south) + (0,-0.9cm)$) {\textbf{\textit{No children}}};
            \node at ($(child2.south) + (0,-0.9cm)$) {\textbf{\textit{No children}}};

        \end{tikzpicture}
    }
    \caption{First-level children of the Tree Normal Form (TNF) enumeration of 5-state 2-symbol Turing machines: each node is a Turing machine, nonhalting machines are leaves of the tree. Internal nodes are halting machines, \ie machines eventually reaching an undefined transition (highlighted in magenta), and their children correspond to all the ways to define this undefined transition. By symmetry, at the first level of the TNF tree, we can ignore machines taking a left move. The two halting machines at the first level of the tree each have 12 children, corresponding to all the choices in $\{\szero,\sone\}\times\{\text{R},\text{L}\}\times\{\text{\stateA},\text{\stateB},\text{\stateC}\}$ for defining the magenta transition. Note that, in this case, the choice of states is reduced from $\{\text{\stateA},\text{\stateB},\text{\stateC},\text{\stateD},\text{\stateE}\}$ to $\{\text{\stateA},\text{\stateB},\text{\stateC}\}$ in order to prevent constructing machines that are only a permutation of one another.}\label{fig:TNF}
\end{figure}

Syntactically, as defined in Section~\ref{sec:TMs}, there are $(2ns + 1)^{ns}$ Turing machines with $n$ states and $s$ symbols. This gives $21^{10} \simeq 1.67\times10^{13} \simeq 16.7 \text{ trillion}$ possible 5-state 2-symbol Turing machines. However, naively counting Turing machines this way does not account for two phenomena:
\begin{enumerate}
    \item \textbf{Unreachable transitions.} Take the 5-state 2-symbol machine where only the first transition is defined as \texttt{0RA} -- the leftmost machine in Figure~\ref{fig:TNF}. This machine is the archetypal Turing machine equivalent of a ``while True'' infinite loop: the machine will never leave the transition, indefinitely drifting to the right of the tape. Hence, none of the $21^9$ machines obtained by defining the other 9 transitions are relevant since these transitions are never reached.
    \item \textbf{State/symbol permutations.} Permuting non-\stateA states and non-zero symbols (\stateA and 0 are special because the initial configuration is the all-0 tape in state \stateA) creates identical machines up to renaming, hence studying the halting of only one of them is enough. State/symbol permutation divides the syntactic space size by a factor of $(n-1)! (s-1)!$.
\end{enumerate}

Tree Normal Form (TNF) enumeration, introduced by Brady in 1963 in his PhD thesis \cite{Brady64} and illustrated in Figure~\ref{fig:TNF}, solves both of these problems: Turing machines are recursively \textit{discovered} starting from the machine with no transitions defined (TNF root). Each enumerated machine is processed by a \textit{pipeline of deciders} (see Section~\ref{sec:deciders}) which will output either \HALT, \NONHALT or \UNKNOWN for each machine:
\begin{itemize}
    \item \HALT. If the machine halts, such as the rightmost machine in Figure~\ref{fig:TNF}, it means that it has met an undefined transition and children of the machine correspond to all the possible ways of defining that undefined transition (highlighted in magenta in Figure~\ref{fig:TNF}). Avoiding redundant state/symbol permutations is dealt with at this point by imposing an order on the yet-to-be-seen states/symbols, \eg children of the rightmost machine in Figure~\ref{fig:TNF} will choose between states $\{\text{\stateA},\text{\stateB},\text{\stateC}\}$ instead of $\{\text{\stateA},\text{\stateB},\text{\stateC},\text{\stateD},\text{\stateE}\}$ since $\text{\stateC}$ is the next unseen state.

    \item  \NONHALT. If the machine does not halt, all its remaining undefined transitions are unreachable and the machine is a leaf of the TNF tree.
    \item \UNKNOWN. If the halting status of a machine remains unknown, it is put in the basket of \textit{holdouts}, \ie machines that remain to be decided. Having solved $S(5)$ means that there are no more 5-state holdouts.
\end{itemize}

Hence, by design, TNF enumeration avoids machines with unreachable transitions and state/symbol permutations. One further optimization in the TNF algorithm is,
at the first level of the TNF tree (see Figure~\ref{fig:TNF}), to avoid machines that make a first move to the left, as they can be symmetrised to go to the right instead, \eg for 5-state 2-symbol machines, this makes the TNF root have 4 children instead of 8. It is also known that only considering machines that first write a 1 is enough to conclude the value $S$, but, for simplicity, this is not used in our work \cite{Marxen_1990,busycoq}. In practice, and in the counts of Table~\ref{tab:TNF-numbers}, leaves of the TNF tree, that have all their transitions defined, are not enumerated because they are obviously nonhalting -- hence not relevant for computing $S$.

\begin{table}[h!]
    \centering
    \begin{tabular}{|l|r|r|r|r|}
        \hline
        $S(n)$ & Nonhalt     & Halt       & Total       & Syntactic/TNF ratio \\
        \hline
        $S(2)$ & 42          & 19         & 61          & $107$               \\
        $S(3)$ & 3,645       & 1,772      & 5,417       & $891$               \\
        $S(4)$ & 609,216     & 249,693    & 858,909     & $8,121$             \\
        $S(5)$ & 133,005,895 & 48,379,894 & 181,385,789 & $91{,}958$          \\
        \hline
    \end{tabular}

    \caption{TNF metrics for $S(2),\dots,S(5)$: number of nonhalting and halting machines the TNF tree, total number of TNF-enumerated machines and ratio between $(4n+1)^{2n}$, which is the syntactic number of $n$-state 2-symbol machines, and the number of machines in the TNF enumeration.}\label{tab:TNF-numbers}
\end{table}

TNF is unreasonably effective (Table~\ref{tab:TNF-numbers}): for 5-state 2-symbol machines, it reduces the search-space from 16 trillion to just $\BBtheFifthTNF$ machines. The variant of TNF that only enumerates machines that start by writing a $1$ was implemented (with a step limit instead of deciders) by Marxen and Buntrock in 1989 to find the fifth \BBfull winner (Figure~\ref{fig:bb5win}), taking about 10 days to run at the time \cite{Marxen_1990}. This same variant (using a $\BBtheFifth$-step limit, no deciders) was implemented in 2022 for \url{bbchallenge.org}’s seed database (see Section~\ref{sec:intro:mainresults}), which yielded 88,664,064 holdouts  \cite{sterin_2022_14955828}. This database, which had to be trusted at the time, became obsolete with the release of \CoqBB. %SIA

% TNF is unreasonably effective, as shown in Table~\ref{tab:TNF-numbers}. In the case of 5-state 2-symbol Turing machines, the total number of machines in the TNF enumeration is $\BBtheFifthTNF$, which is $91{,}958$ times smaller than $16$ trillion, the number of synctatically correct machines. The variant of TNF that only enumerates machines that start by writing a $1$ (and with a step limit instead of deciders) was implemented in 1989 by Heiner and Buntrock in order to find the fifth \BBfull winner (Figure~\ref{fig:bb5}); at the time, this enumeration took about 10 days \cite{Marxen_1990}. This same variant (also without deciders, using a simulation limit of $\BBtheFifth$ steps) was implemented in 2022 for the release of \url{bbchallenge.org}'s seed database \cite{sterin_2022_14955828} (see Section~\ref{sec:intro:mainresults}) which provided a set of 88,664,064 holdouts to prove nonhalting -- this database, which had to be trusted, became obsolete with the release of \CoqBB.

\paragraph{\CoqBB TNF implementation.} TNF enumeration, as described here, is implemented in \CoqBB for the proofs of $S(2),\dots,S(5)$; see file \texttt{TNF.v}. A \texttt{SearchQueue} abstraction with DFS capabilities is implemented, see function \texttt{SearchQueue\_upds}. The search queue is initialised with the TNF root (this is most obvious in the proofs of $S(<5)$, see file \texttt{BB4\_TNF\_Enumeration.v}), and deciders (see Section~\ref{sec:deciders}) are run on the enumerated Turing machines. Halting machines' children are added to the queue: the goal of the proof is to empty the queue. Importantly, \texttt{Lemma\ SearchQueue\_upd\_spec} ensures that $S$ can be computed considering only TNF-enumerated machines.

Taking advantage of the tree structure, compilation of the $S(5)$ proof was parallelised by isolating the 12 children of the rightmost machine in Figure~\ref{fig:TNF} in separate, independent files; see folder \texttt{BB5\_TNF\_Enumeration\_Roots/}. Parallelising the compilation made the proof compile in 3 hours (on 13 cores) instead of 13 hours. Switching to \Coq's more powerful \texttt{native\_compute} engine \cite{nativecompute} further brought parallel compilation time down to 45 minutes. This compilation time could be improved further by splitting the tree in even more files with only RAM and number of cores as limiting factors.

\CoqBB's proof of $S(2,4)$ implements TNF almost exactly in the way described here but for the fact that, for simplicity, it does not impose an order on non-0 symbols meaning that identical machines up to symbol renaming are enumerated. In this ``quasi-TNF'' setup, the proof of $S(2,4)$ enumerates 2,154,217 machines of which 1,432,880 are nonhalting.

Being able to perform the TNF enumeration of 5-state 2-symbol Turing machines directly in \Coq came as a surprise for most \url{bbchallenge.org} collaborators when \CoqBB was released. Results of the enumeration (\ie list of machines with halting status and decider) were extracted from the proof and made available at \url{https://docs.bbchallenge.org/CoqBB5_release_v1.0.0/}.

\paragraph{TNF normalisation.} To compute the TNF-equivalent of an arbitrary Turing machine, states are reordered by first-visit order, and all R moves are swapped with L (and vice versa) if the initial move is L. This process is not computable in general, as we cannot determine in advance which states will be visited. However, for $n$-state machines, knowing $S(n-2)$ makes TNF normalisation computable.\footnote{Apart from being able to remove unreachable transitions from the initial machine; knowing $S(n)$ solves this. } %SIA

% \paragraph{TNF Normalisation.} Given an arbitrary Turing machine, finding its TNF-equivalent machine is done by ordering its state in the order they are visited by the machine and changing all R moves to L and vice versa if the first transition uses L instead of R.  TNF Normalisation is not computable because we don't know if some states will ever be visited, hence we can't always reorder them without additional knowledge on the machine's behaviour. In practice, knowing $S(n-2)$ makes TNF Normalisation computable for $n$-state machines.

\section{Deciders}\label{sec:deciders}

\subsection{Pipelines}\label{sec:pipelines}

In this work, we call a \textit{decider} a program that takes as input a Turing machine $\mathcal{M}$ and that returns in finite time either \HALT, \NONHALT, or \UNKNOWN depending on whether it was able to detect the machine's halting status or not.\footnote{Hence, we do not use the word ``decider'' in the traditional sense of theoretical computer science since, although our deciders finish, they are partial.}

A \textit{pipeline} consists of applying different proof techniques in sequence, mostly consisting of deciders: a machine is tested by each decider successively until one of them outputs \HALT or \NONHALT. We call \textit{Sporadic Machines} the \numSporadic machines that were not solved by any decider but using individual proofs of nonhalting instead, see Section~\ref{sec:sporadic}. Table~\ref{tab:pipelineBB5} gives an approximation of the pipeline implemented in \CoqBB in order to prove $S(5) = \BBtheFifth$, see Theorem~\ref{th:BB5}. Similarly, Table~\ref{tab:pipelineBB2x4} and Table~\ref{tab:pipelineBB4} respectively give approximations of the pipelines leading to $S(2,4) = \BBTxF$ and $S(4) = \BBtheFourth$ -- the latter confirming the result for $S(4)$ originally given in \cite{Brady83}.

The exact pipelines are provided in Appendix~\ref{app:pipelines}. They differ mainly in the specific parameters and, occasionally, the algorithmic variants used for each decider. In some cases, deciders are interleaved -- for example, the loop decider is initially invoked with a small step-count parameter, followed by other deciders, and then called again later with a higher step-count. %SIA

% The exact pipelines are give in Appendix~\ref{app:pipelines} and only differ in that the precise parameters and sometimes algorithmic variants are given for each decider which are sometimes interleaved with each other (\eg the decider for loops is first called with a small step-count parameter, then other deciders are applied and later on it is called again with higher step-count parameter).

\begin{table}[h!]
  \centering
  \resizebox{\textwidth}{!}{%
    \begin{tabular}{|l|rrr|}
      \hline
      $S(5)$ pipeline                                                                      & Nonhalt                         & Halt                           & Total decided \\
      \hline
      1. Loops, see \textbf{Section~\ref{sec:loops}}                                       & 126,994,099                     & 48,379,711                     & 175,373,810   \\
      2. $n$-gram Closed Position Set (NGramCPS), see \textbf{Section~\ref{sec:n-gramCPS}} & 6,005,142                       & 0                              & 6,005,142     \\
      3. Repeated Word List (RepWL), see \textbf{Section~\ref{sec:RepWL}}                  & 6,577                           & 0                              & 6,577         \\
      4. Finite Automata Reduction (FAR), see \textbf{Section~\ref{sec:FAR}}               & 23                              & 0                              & 23            \\
      5. Weighted Finite Automata Reduction (WFAR), see \textbf{Section~\ref{sec:WFAR}}    & 17                              & 0                              & 17            \\
      6. Long halters (simulation up to $\BBtheFifth$ steps)                               & 0                               & 183                            & 183           \\
      7. Sporadic machines, individual proofs, see \textbf{Section~\ref{sec:sporadic}}     & 13                              & 0                              & 13            \\
      8. \texttt{1RB}-reduction, see \textbf{Section~\ref{sec:deciders-overview}}          & 24                              & 0                              & 24            \\ \hline
      Total                                                                                & \multicolumn{1}{r}{133,005,895} & \multicolumn{1}{r}{48,379,894} & 181,385,789   \\ \hline
    \end{tabular}
  }
  \caption{Approximation of the $S(5)$ pipeline as implemented in \CoqBB. All the $\BBtheFifthTNF$ enumerated 5-state machines are decided by this pipeline, which solves $S(5) = \BBtheFifth$, see Theorem~\ref{th:BB5}. The exact pipeline, with deciders parameters is given in Appendix~\ref{app:pipelines}. }
  \label{tab:pipelineBB5}
\end{table}

\begin{table}[h!]
  \centering
  \begin{tabular}{|l|rrr|}
    \hline
    $S(2,4)$ pipeline                                                                    & Nonhalt   & Halt    & Total decided \\
    \hline
    1. Loops, see \textbf{Section~\ref{sec:loops}}                                       & 1,263,302 & 721,313 & 1,984,615     \\
    2. $n$-gram Closed Position Set (NGramCPS), see \textbf{Section~\ref{sec:n-gramCPS}} & 163,500   & 0       & 163,500       \\
    3. Repeated Word List (RepWL), see \textbf{Section~\ref{sec:RepWL}}                  & 6,078     & 0       & 6,078         \\
    4. Long halters (simulation up to $\BBTxF$ steps)                                    & 0         & 24      & 24            \\
    \hline
    Total                                                                                & 1,432,880 & 721,337 & 2,154,217     \\ \hline
  \end{tabular}
  \caption{Approximation of the $S(2,4)$ pipeline as implemented in \CoqBB. All the $\BBTxFTNF$ enumerated 2-state 4-symbol machines are decided by this pipeline, which solves $S(2,4) = \BBTxF$, see Theorem~\ref{th:BB2x4}. The exact pipeline, with deciders parameters is given in Appendix~\ref{app:pipelines}. }\label{tab:pipelineBB2x4}
\end{table}

\begin{table}[h!]
  \centering
  \begin{tabular}{|l|rrr|}
    \hline
    $S(4)$ pipeline                                                                      & Nonhalt & Halt    & Total decided \\
    \hline
    1. Loops, see \textbf{Section~\ref{sec:loops}}                                       & 588,373 & 249,693 & 838,066       \\
    2. $n$-gram Closed Position Set (NGramCPS), see \textbf{Section~\ref{sec:n-gramCPS}} & 20,841  & 0       & 20,841        \\
    3. Repeated Word List (RepWL), see \textbf{Section~\ref{sec:RepWL}}                  & 2       & 0       & 2             \\
    \hline
    Total                                                                                & 609,216 & 249,693 & 858,909       \\
    \hline
  \end{tabular}
  \caption{Approximation of the $S(4)$ pipeline as implemented in \CoqBB. All the $\BBtheFourthTNF$ enumerated 4-state machines are decided by this pipeline, which brings confirmation that $S(4) = \BBtheFourth$ \cite{Brady83}, see Theorem~\ref{th:BB4}. The exact pipeline, with deciders parameters is given in Appendix~\ref{app:pipelines}. }\label{tab:pipelineBB4}
\end{table}

\newpage

\subsection{Deciders overview}\label{sec:deciders-overview}

Five deciders are used in \CoqBB to solve $S(5)$, see Table~\ref{tab:pipelineBB5}: Loops, $n$-gram Closed Position Set (NGramCPS), Repeated Word List (RepWL), Finite Automata Reduction (FAR) and Weighted Finite Automata Reduction (WFAR). They are individually described in Sections~\ref{sec:loops} to \ref{sec:WFAR}. To the best of our knowledge, all these deciders are original.\footnote{There existed a previous algorithm to decide loops \cite{Lin1963}, but we present a new one.} Solving $S(2,4)$ (and previously known $S(4)$ \cite{Brady83}) only used  a subset of these deciders (Tables~\ref{tab:pipelineBB2x4}~and~\ref{tab:pipelineBB4}) and required a lot less compute and overall effort -- \eg no individual nonhalting proofs, in contrast with 5-state Sporadic Machines, Section~\ref{sec:sporadic}.

All these deciders can be expressed\footnote{In practice, all deciders but Loops (for which the CTL framework is not useful) \textit{are} expressed using the CTL framework.} in the same general framework, known as Closed Tape Language (CTL) which is an abstract interpretation idea attributed to Marxen and Buntrock and was first documented by Ligocki \cite{ShawnCTL}:
\vspace{-1ex}
\paragraph{General framework: Closed Tape Language (CTL).} For a given Turing machine, assume there is a set $C$ of configurations such that:
\begin{enumerate}
  \item $C$ contains the initial all-0 configuration.
  \item $C$ is \textit{closed} under transitions: for any $c \in C$, the configuration one step later belongs to $C$.
  \item $C$ does not contain any halting configuration.
\end{enumerate}

Then, the machine does not halt from any configuration of $C$ and, in particular, from the initial all-0 configuration.

\vspace{-1ex}
\paragraph{Regularity.} All our deciders but WFAR are instances of \textit{regular} CTL, meaning that set $C$ is a regular language -- \ie $C$ is described using a Finite State Automaton (FSA) or, equivalently, a regular expression. Said otherwise, regular CTL approximates the set of configurations of a Turing machine using a regular language that is larger than the machine's set of configurations but on which it is easier (in practice, trivial) to ensure CTL conditions, \ie (i) membership of the all-0 configuration (ii) closure under Turing machine steps and (iii) absence of halting configurations. NGramCPS and RepWL each focus on specific types of regular languages and are good introductory examples to illustrate this method, while FAR generalises to arbitrary regular languages. We say that a machine is \textit{regular} if it can be proven nonhalting (from all-0 tape) using regular CTL, otherwise we say it is \textit{irregular}. Other regular CTL deciders were developed by The bbchallenge Collaboration, but were not used to solve $S(5)$ \cite{BruteforceCTL, SymbolicTM, regexy, bbchallenge_part1}.

\vspace{-1ex}
\paragraph{Irregularity.} WFAR is an analogue of FAR, leveraging CTL using \textit{Weighted} FSAs which are a nonregular generalisation of FSAs, see Section~\ref{sec:WFAR}. Not all machines solved by WFAR are necessarily irregular but we strongly suspect that the 17 machines it solves in the $S(5)$ pipeline (Table~\ref{tab:pipelineBB5}) are irregular because intensive search did not allow finding regular CTL solutions for them. Informal irregularity arguments have been given for sporadic machines ``Finned~\#3'' and ``Skelet~\#17'' \cite{irregularFinned3, irregularSk17}, see Section~\ref{sec:sporadic}.

Interestingly, only regular deciders are needed to solve $S(4)$ and $S(2,4)$, see Tables~\ref{tab:pipelineBB2x4}~and~\ref{tab:pipelineBB4}, which, assuming $S(5)$ irregularity arguments are correct, draws a conceptual line separating $S(4)$ from $S(5)$.

\vspace{-1ex}
\paragraph{Deciders vs. verifiers.} While we put all the methods under the decider umbrella it is worthwhile to mention that, in \CoqBB, only the \textit{verifier} part of FAR and WFAR are implemented. This means that instead of searching for a CTL set $C$ (see above), it is given and verified correct: \CoqBB hardcodes a total of $40$ FAR /
WFAR certificates which are FSAs / Weighted FSAs describing CTL sets $C$. See files \texttt{Verifier\_FAR\_Hardcoded\_Certificates.v} and \texttt{Verifier\_WFAR\_Hardcoded\_Certificates.v} in \CoqBB's folder \texttt{BB5\_Deciders\_Hardcoded\_Parameters}.

\enlargethispage{\baselineskip}
\paragraph{\texttt{1RB}-reduction.} Any TNF-enumerated machine (Section~\ref{sec:enum}) whose initial transition writes a \szero and which has at least one transition writing a \sone (TNF guarantees the transition is reachable), can be transformed into a machine that starts by writing a \sone and visits the same configurations (up to state-renaming) but for the first few that wrote a \szero, see \CoqBB's function \texttt{TM\_to\_NF}. We call this transformation \texttt{1RB}-reduction (\texttt{1RB} is the first transition of the new machine). Hence, any machine whose reduction to \texttt{1RB} already has a proof of nonhalting can be decided using the same argument! This is proven in \CoqBB's \texttt{Lemma TM\_to\_NF\_spec}. For instance, TNF-enumerated machine \texttt{0RB0LD\_1RC1RE\_1LA1RC\_1LC1LD\_---0RB} reduces that way to sporadic machine ``Finned \#3'' (Section~\ref{sec:sporadic}), and is decided using the same proof. In the $S(5)$ pipeline (Table~\ref{tab:pipelineBB5}) this argument is used to decide 24 machines whose \texttt{1RB}-reduction correspond to 23 FAR/WFAR certificates (see above) and one sporadic machine, ``Finned \#3''.

\newpage
\subsection{Loops}\label{sec:loops}

\begin{figure}[h!]
    \centering
    \includegraphics[width=0.25\textwidth]{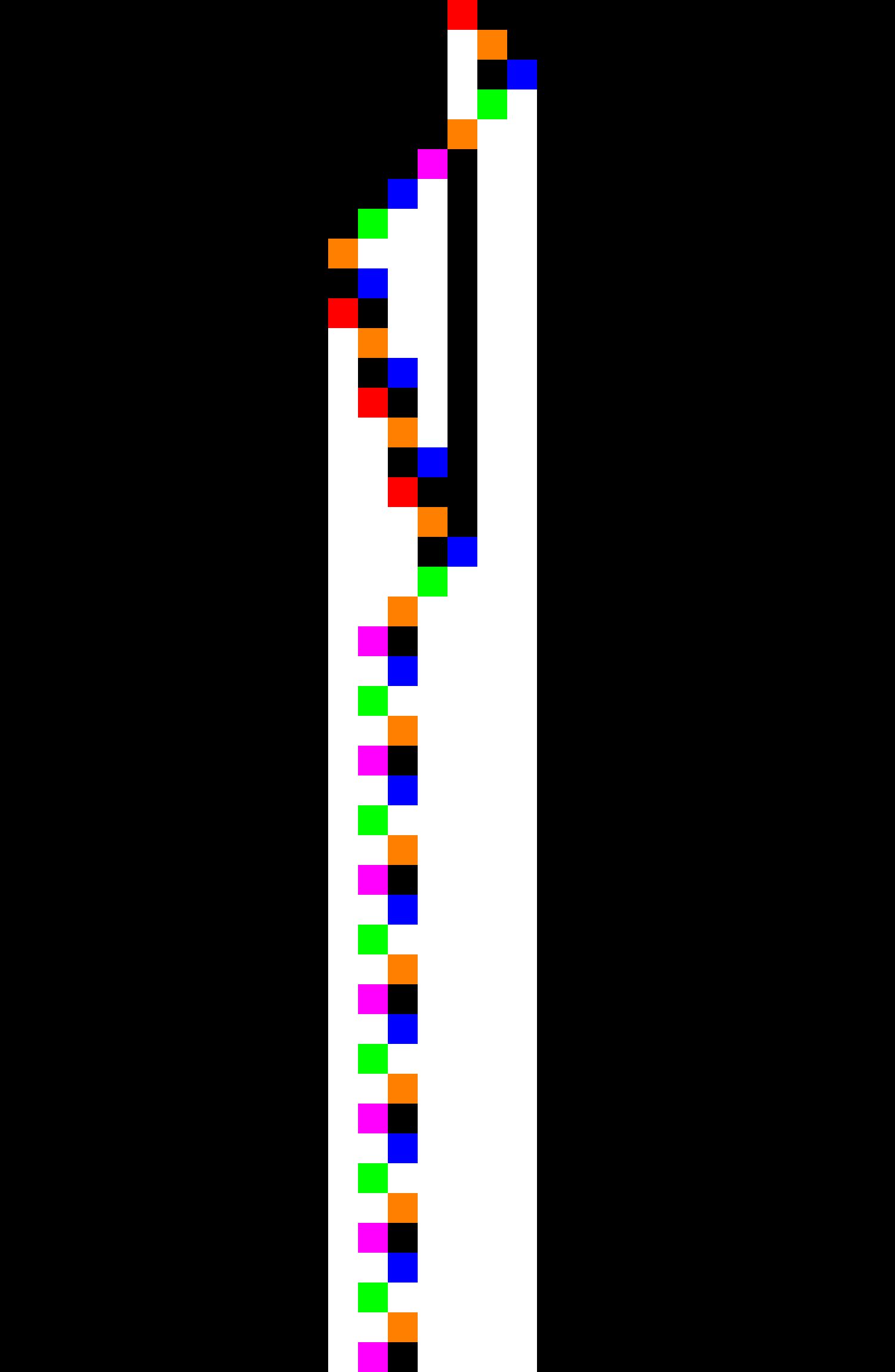}
    \hspace{10ex}
    \includegraphics[width=0.25\textwidth]{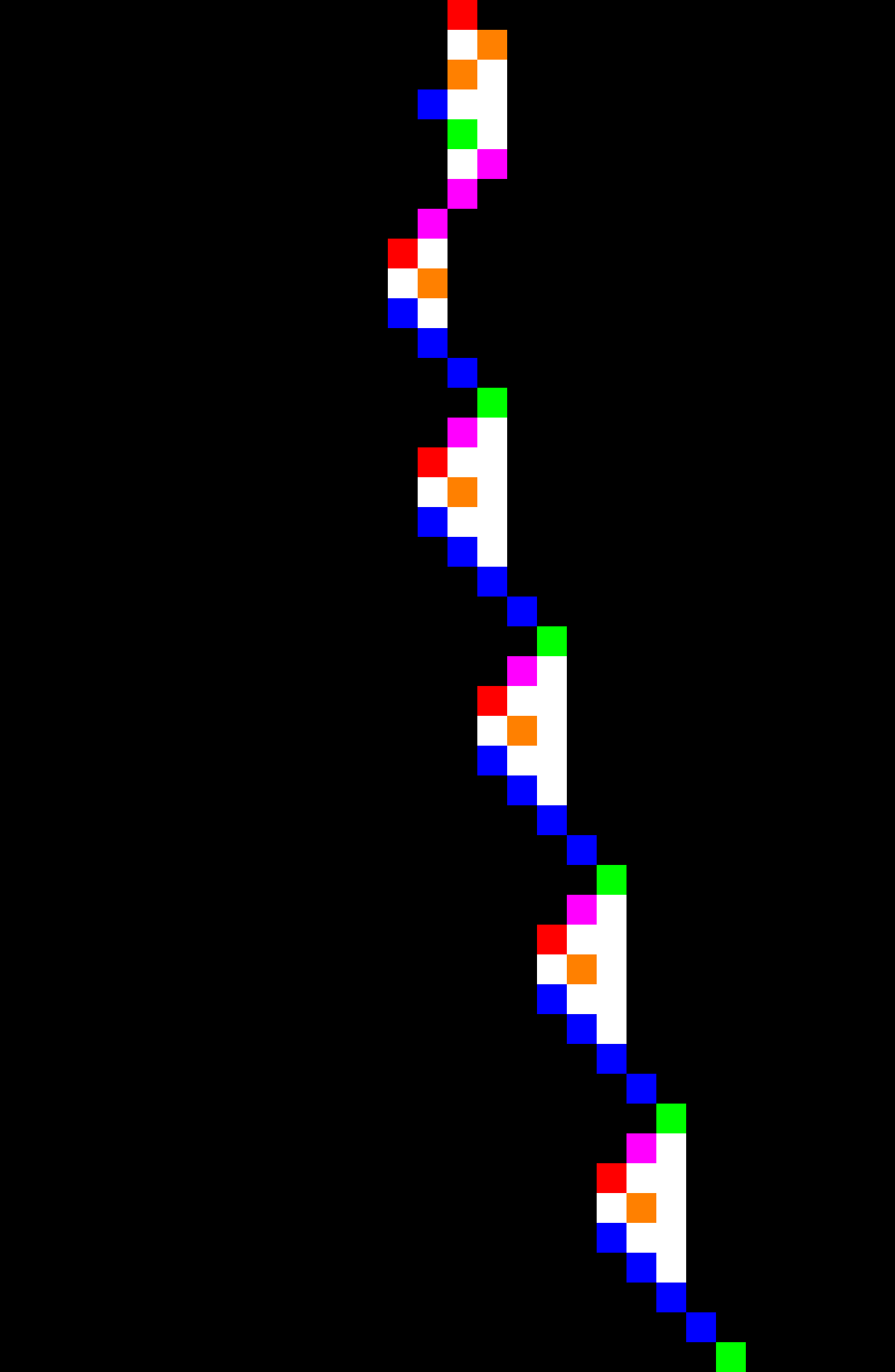}
    \caption{Space-time diagrams of the first 45 steps of a \textit{\cycler}, \tm{1RB---_0RC0LE_1LD0LA_1LB1RB_1LC1RC} (left) and of the first 45 steps of a \textit{\TC}, \tm{1RB---_1LB1LC_0RD0RC_1LE1RE_1LA0LE} (right). \cyclers are machines that eventually repeat the same configuration forever. \TCs are machines that eventually repeat the same configuration forever, but translated in space. We refer to these two types of machines as \textit{loops}.}\label{fig:loops}
\end{figure}

Arguably, one of the most elementary arguments to prove that a Turing machine does not halt on a given input is to show that it enters a \textit{loop} by eventually repeating the same configuration, \ie same tape content, head position, and state, see Figure~\ref{fig:loops}~left. Slightly less obvious, but extremely common (see Section~\ref{sec:loops:results}), is the case where a configuration is repeated but translated in space (analogously to \textit{gliders} in cellular automata), also leading to nonhalting, see Figure~\ref{fig:loops}~right. We respectively call these types of machines: (i) \textit{\cyclers} and (ii) \textit{\TCs}, which we regroup under the umbrella term of loops.

Deciding Cyclers reduces to the well-known mathematical problem of detecting the cycles of a function and standard detection algorithms exist \cite{wiki:Cycle_detection}, the simplest one consisting in memorising each successive configuration of the machine until encountering one that has been already seen. Translated Cyclers, also known as \textit{Lin's recurrence}, have first been described and decided in Shen Lin's 1963 PhD thesis \cite{Lin1963}, other similar algorithms to detect them have been developed since then.\footnote{\url{https://discuss.bbchallenge.org/t/decider-translated-cyclers/34}}

Here, we develop a completely different algorithm (Algorithm~\ref{alg:loops}) for deciding both \cyclers and \TCs. The particularity of this algorithm is that it detects loops only by analysing the history of state, read-symbol and \headposs visited by the machine, instead of considering entire configurations (\ie with full tape content information). Hence, in theory, Algorithm~\ref{alg:loops} can be implemented to use less memory than previously-known algorithms.\footnote{In practice, for simplicity, the Coq implementation (see Section~\ref{sec:loops:results}) stores entire tapes (whereas it only uses the head's information), hence it does not avail itself of this potential memory optimisation.} This algorithm was introduced with \CoqBB.

\subsubsection{Algorithm}\label{sec:loops:algo}

Let us call the \textit{transcript} of a machine the list of successive \ssps visited by the machine from the all-0 tape. For instance, the transcript of the \cycler in Figure~\ref{fig:loops}~(left) starts with \texttt{A0 B0 C0 D0 B1 E0 C0 D0 B0 C1 A0} and the transcript of the \TC in Figure~\ref{fig:loops}~(right) starts with \texttt{A0 B0 B1 C0 D1 E1 E1 E0 A0 B1 C1}. Surprisingly, it turns out that in order to detect loops, we only have to track when a transcript repeats the same sequence twice back-to-back, for instance, in the case of the Cycler in Figure~\ref{fig:loops}: \texttt{A0 B0 C0 D0 B1 E0 C0 D0 B0 C1 A0 B0 C1 A0 B0 C1 A0 B0 C0 D0 \textbf{\underline{B1 E1 C0 D1}} \textbf{\underline{B1 E1 C0 D1}}}. When such a repetition occurs, we use the extra information of \headpos to conclude:

\begin{enumerate}
    \item If when entering the second repetition the head is at the same position it was at the beginning of the first repetition, then we have detected a \cycler, \eg for the \cycler in Figure~\ref{fig:loops}~(left), here is the end of the transcript with extra head-position information given after each \ssp: \texttt{\underline{\textbf{B1(-2)} E1(-3) C0(-2) D1(-3)} \underline{\textbf{B1(-2)} E1(-3) C0(-2) D1(-3)}}.

    \item If, at the beginning of both repetitions, the head is at the same extremity of the tape (\ie both positions are either both a local maximum or both a local minimum) then we have detected a \TC. For the \TC in Figure~\ref{fig:loops}~(right):
          \begin{align*}
               & \texttt{A0(0)* B0(1)* B1(0) C0(-1) D1(0) E1(1)* E1(0) E0(-1) A0(-2) B1(-1) C1(-2) C1(-1) C0(0)} \\
               & \texttt{\underline{\textbf{D0(1)* E0(0) A0(-1) B1(0) C1(-1) C1(0) C1(1)* C0(2)*}}}              \\
               & \texttt{\underline{\textbf{D0(3)* E0(2) A0(1) B1(2) C1(1) C1(2) C1(3)* C0(4)*}}}
          \end{align*}
          \noindent where \texttt{*} indicates steps where the machine is at the right-extremity of the tape (\headpos local maximum).

\end{enumerate}

Algorithm~\ref{alg:loops} first computes $L$ transcript steps, then, starting from the end of the transcript, tries to detect two consecutive size-$l$ transcript repetitions, for all $l \leq L/2$. Once it has found two such repetitions, it tries to ``rewind'' them until they both start on positions that are on the (same) extremity of the tape. We prove that Algorithm~\ref{alg:loops} is correct in Theorem~\ref{th:loops}.

\begin{algorithm}
    \caption{{\sc decider-Loops}, reformulates the algorithm \texttt{loop1\_decider} of Coq-BB5.}\label{alg:loops}

    \begin{algorithmic}[1]
        \State{\textbf{Input:} A Turing machine `$\mathcal{M}$', a step-limit parameter $L$.}
        \State{\textbf{Output:} \NONHALT if the decider detects that the machine is a loop, \HALT if the machine halts and \UNKNOWN otherwise.}

        \State
        \State Simulate $\mathcal{M}$ for $L$ steps and save the history of each consecutive state, read-symbol and position, \ie consecutive $T_i = (s_i,m_i,d_i)\in\states\times\alphabet\times\Z$ for $0 \leq i \leq L$ and $T_0 = (\stateA,\symbolzero,0)$.

        \State \If{the machine has halted before $L$ steps}
        \State \Return HALT \label{alg:loops:halt}
        \EndIf
        \State \For{$l$ \textbf{in} $[1,L/2]$ } \Comment{$l$ is the length of the potential loop}
        \State $K = L-l$
        \State $\text{offset} = 0$ \Comment{Offset in case we need to keep looking for extremal tape position}
        \State \For{$i$ \textbf{in} $[0,l+\text{offset}[$ }
        \Comment{$\text{Offset}$ and therefore loop bounds may vary between iterations}
        \If{$K-i < 0$}
        \State \textbf{break}
        \State \EndIf
        \State $s,m,d = T_{L-i}$
        \State $s',m',d' = T_{K-i}$
        \State
        \If{$s\neq s'$ \textbf{or} $m \neq m'$}\label{alg:loops:testeq} \Comment{Comparing \ssp equality at each step}
        \State \textbf{break}
        \EndIf

        \State \If{$i = l+\text{offset}-1$}
        \If{$d = d'$}\label{alg:loops:cycler}
        \State \Return \NONHALT \Comment{We have detected a Cycler}
        \EndIf
        \State \If{$d = \text{max} \{d_j \, | \, j < L-i \}$ \textbf{and} $d' = \text{max} \{d_j \, | \, j < K-i \}$}\label{alg:loops:tcplus}
        \State \Return \NONHALT \Comment{We have detected a (positive) Translated Cycler}
        \EndIf
        \State \If{$d = \text{min} \{d_j \, | \, j < L-i \}$ \textbf{and} $d' = \text{min} \{d_j \, | \, j < K-i \}$}\label{alg:loops:tcminus}
        \State \Return \NONHALT \Comment{We have detected a (negative) Translated Cycler}
        \State \EndIf
        \State $\text{offset} = \text{offset} + 1$
        \State \EndIf
        \EndFor
        \EndFor

        \State \Return \UNKNOWN

    \end{algorithmic}

\end{algorithm}

\subsubsection{Correctness}
Proving the correctness of this decider is surprisingly nontrivial. Let us represent the space-time diagram of a given Turing machine $\mathcal{M}$ (Section~\ref{sec:TMs}) using partially defined functions $f_\mathcal{M},g_\mathcal{M},h_\mathcal{M}$, and, $F_\mathcal{M}$:
\begin{align*}
    f_\mathcal{M} & : \N \partialto \Z \to \alphabet\text{, tape content}                                \\
    g_\mathcal{M} & : \N \partialto \Z\text{, head position}                                             \\
    h_\mathcal{M} & : \N \partialto \states\text{, head state}                                           \\
    F_\mathcal{M} & : \N \partialto \Z \to (\alphabet\times \states)\text{, tape content and head state}
\end{align*}

\newcommand{\baref}{F}

At time $t\in\N$, $f_\mathcal{M}(t)$ gives the tape content (as a total function $\Z \to \alphabet$), $g_\mathcal{M}(t)$ gives the head position and $h_\mathcal{M}(t)$, the head state, assuming in each case that $\mathcal{M}$ has not halted before time $t$, otherwise values are not defined. For brevity, when the Turing machine is clear from context, we may write $f,g,h$. In the case of $f$, we will also use notation $f(t,z)$ or $f(p)$ with $p \in \N \times \Z$ seen as a vector, instead of $f(t)(z)$ when $f(t)$ is defined. Using same extended notations as for $f$, we also define $\baref_\mathcal{M}(t,z) = (f(t,z),h(t))$ which gives tape content with head state information at each time when $f$ and $h$ are defined.
Finally, when assuming/claiming $f(x) = f(y)$ for some $x,y \in \N \times \Z$ we also implicitly assume/claim that $f$ is defined at $x$ and $y$; same for $\baref$, $g$, and, $h$.

Transcripts as defined in Section~\ref{sec:loops:algo} correspond to sequences of the form $F_\mathcal{M}(t,g(t))$ with $t\in\N$.

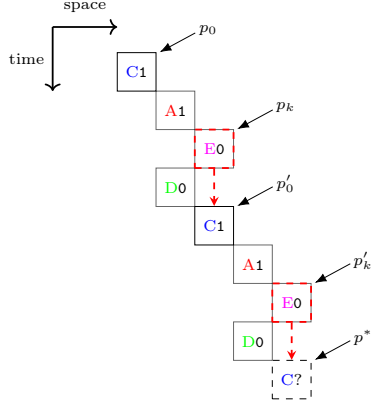
\begin{figure}[h!]
    \noindent
    \begin{minipage}[t]{1\textwidth}
        \centering
        \scalebox{0.85}{
            \begin{tikzpicture}[every node/.style={minimum size=6mm, draw, font=\scriptsize}, scale=1]

                \draw[->, thick] (-1.3,2.7) -- (-1.3,1.7) node[draw=none, midway, left] {\scriptsize time};
                \draw[->, thick] (-1.3,2.7) -- (-0.3,2.7) node[draw=none, midway, above] {\scriptsize space};

                \node (C1) at (0,2) {\stateCx\sone};
                \node[draw=black, draw opacity=0.6] (A1) at (0.6,1.4) {\stateAx\sone};
                \node[draw=black, draw opacity=0.6] (B0) at (1.2,0.8) {\stateEx\szero};
                \node[draw=black, draw opacity=0.6] (D0) at (0.6,0.2) {\stateDx\szero};

                \node[dashed, draw=red, thick] (red) at (1.2,0.8) {};

                \node (C1') at (1.2, -0.4) {\stateCx\sone};
                \node[draw=black, draw opacity=0.6] (A1') at (1.8, -1.0) {\stateAx\sone};
                \node[draw=black, draw opacity=0.6] (B0') at (2.4, -1.6) {\stateEx\szero};
                \node[draw=black, draw opacity=0.6] (D0') at (1.8, -2.2) {\stateDx\szero};

                \node[dashed, draw=red, thick] (red') at (2.4,-1.6) {};

                \node[dashed] (C1'') at (2.4,-2.8) {\stateCx?};

                \draw[dashed, draw=red, thick, ->, >=stealth] (red.south) -- (C1'.north);
                \draw[dashed, draw=red, thick, ->, >=stealth] (red'.south) -- (C1''.north);

                \node[draw=none] at ($(C1.north east)+(0.6,0.3)+(0.2,0.05)$) {$p_0$};
                \draw[-{Latex[length=1.8mm]}] ($(C1.north east)+(0.6,0.3)$) -- ($(C1.north east)-(-0.05,0.00)$);

                \node[draw=none] at ($(C1'.north east)+(0.6,0.3)+(0.2,0.05)$) {$p_0'$};
                \draw[-{Latex[length=1.8mm]}] ($(C1'.north east)+(0.6,0.3)$) -- ($(C1'.north east)-(-0.05,0.00)$);

                \node[draw=none] at ($(C1''.north east)+(0.6,0.3)+(0.2,0.05)$) {$p^*$};
                \draw[-{Latex[length=1.8mm]}] ($(C1''.north east)+(0.6,0.3)$) -- ($(C1''.north east)-(-0.05,0.00)$);

                \node[draw=none] at ($(B0'.north east)+(0.6,0.3)+(0.2,0.05)$) {$p'_k$};
                \draw[-{Latex[length=1.8mm]}] ($(B0'.north east)+(0.6,0.3)$) -- ($(B0'.north east)-(-0.05,0.00)$);

                \node[draw=none] at ($(B0.north east)+(0.6,0.3)+(0.2,0.05)$) {$p_k$};
                \draw[-{Latex[length=1.8mm]}] ($(B0.north east)+(0.6,0.3)$) -- ($(B0.north east)-(-0.05,0.00)$);

            \end{tikzpicture}
        }
    \end{minipage}
    \caption{Illustration of Lemma~\ref{lem:vector}: head-only space-time diagram showing transcript repetition \texttt{\textbf{\underline{C1 A1 E0 D0}} \textbf{\underline{C1 A1 E0 D0}}}. Coordinates $p_0 = (t_0,g(t_0))$ correspond to the beginning of the first repetition, and $p'_0 = (t_0',g(t_0'))$ of the second. We can easily show that the machine state is the same state, \stateC, at position $p^*$ and as at position $p_0'$, Lemma~\ref{lem:vector}~Point~\ref{lem:vector:pt2}. Less immediate, in this case, using $p_k$ and $p'_k$, we can show that positions $p^*$ and $p_0'$ also share same read-symbol (depicted as \texttt{?} to signify that it is less immediate to show), which is the symbol outputted after $p_k$ and $p_k'$, Lemma~\ref{lem:vector}~Point~\ref{lem:vector:pt4}. }\label{fig:loop-lemma}
\end{figure}

\begin{lemma}\label{lem:vector} Let $\mathcal{M}$ be a Turing machine.
    Assume there is $ t_0 \in \N$ and $l \in \N^+$ such that
    for all $0 \leq i < l$: $$F_\mathcal{M}(p_i)   = F_\mathcal{M}(p'_i)$$

    \noindent with $t_i = t_0 + i$, $t_i' = t_i+l$ and $p_i = (t_i, g(t_i))$, $p_i' = (t_i', g(t_i'))$. Call $p^*=(t^*,g(t^*))$ with $t^*=t'_0 + l$. Then:
    \begin{enumerate}
        \item For all $0 \leq i < l$ we have: $p'_i - p_i = p_0' - p_0$. Also, $g(t^*)$ is defined and we have $p^* - p_0' = p_0' - p_0$. \label{lem:vector:pt1}
        \item We have $h(t^*) = h(t'_0)$.\label{lem:vector:pt2}
        \item For all  $0 \leq i < l$, $g(t_i) = g(t_0') \Leftrightarrow g(t_i') = g(t^*)$.\label{lem:vector:pt3}
        \item If there is $0 \leq k < l$ such that $g(t_k') = g(t^*)$ then $f(p^*) = f(p_0')$.\label{lem:vector:pt4}
    \end{enumerate}
    Figure~\ref{fig:loop-lemma} illustrates this Lemma and Point~\ref{lem:vector:pt4} in particular.
\end{lemma}
\begin{proof}
    First note that by definition, all $p_i$ and $p_i'$ correspond to head positions and the condition $F_\mathcal{M}(p_i)   = F_\mathcal{M}(p'_i)$ means that there is a repetition in the machine's transcript (see Section~\ref{sec:loops:algo}).
    \begin{enumerate}
        \item Because $F(p_0) = F(p'_0)$ we know that at times $t_0$ and $t_0'$ the head is in same state, reading the same symbol, hence the same transition executes and, in particular, both heads move in the same direction $m \in \{-1,1\}$, giving the existence of $u = (1, m)$ such that $p_1 = p_0 + u$ and $p_1' = p_0' + u$. Hence, $p_1' - p_1 = p_0' + u - p_0 - u = p_0' - p_0$. Repeating the same argument for each $i < l$ gives $p_i' - p_i = p_0' - p_0$. Finally, applying the same argument with $F(p_{l-1}) = F(p'_{l-1})$ gives that $g(t^*)$ is defined: the machine has not halted between steps $t_{l-1}$ and $t'_0$ hence it does not halt between steps $t'_{l-1}$ and $t^*$. Furthermore, $p^*-p_0' = p'_0 - p_0$ as $p^*$ corresponds to one time step after $p'_{l-1}$ and $p'_0$ one time step after $p_{l-1}$.
        \item Similarly to above, $F(p_{l-1}) = F(p'_{l-1})$ implies that the machine will transition to the same state after $p_{l-1}$ and $p'_{l-1}$, giving $h(t^*) = h(t'_0)$.

        \item Let $0 \leq i < l$ such that $g(t_i) = g(t_0')$. Using Point~\ref{lem:vector:pt1}, we have $p_i' = p_i + (p_0' - p_0)$, meaning $g(t_i) = g(t_0') \Leftrightarrow g(t_i') = g(t_i) + g(t_0') - g(t_0)$ which rewrites as $g(t_i) = g(t_0') \Leftrightarrow g(t_i') = g(t_0') + g(t_0') - g(t_0) = g(t^*)$ because $g(t^*) = g(t_0') + g(t_0') - g(t_0)$ using Point~\ref{lem:vector:pt1} with $p^* = p_0' + (p_0' - p_0)$. In the end, we get $g(t_i) = g(t_0') \Leftrightarrow g(t_i') = g(t^*)$ as needed.

        \item Without loss of generality, let us assume that $k$ is maximal. Using Point~\ref{lem:vector:pt3} we get that $g(t_k) = g(t_0')$ and, by maximality of $k$, there is no $s > k$ with $s <l$ such that $g(t_s) = g(t_0')$. Because $F(p_k) = F(p_k')$, the same symbol is outputted after $p_k$ and $p_k'$, giving $f(t_k+1,g(t_k)) = f(t_k'+1,g(t_k'))$ and by maximality of $k$, the cells are not revisited respectively before $p_0'$ and $p^*$, giving $f(p_0') = f(t_k+1,g(t_k))$ and $f(p^*) = f(t_k'+1,g(t_k'))$. Hence we have $f(p^*) = f(p_0')$ as needed.

    \end{enumerate}

\end{proof}

\begin{definition}[Loops]\label{def:loops}
    Let $M$ be a Turing machine and let $F_M$ and $g$ be the functions defined as above. Let $l \in \N^+$, $t_0 \in \N$.
    We say that $M$ is a loop of period $l$ and pre-period $t_0$ if for all $t > t_0$, $F_M(t,g(t)) = F_M(t_0 + j,g(t_0 + j))$ where $j$ is the remainder in the division of $t - t_0$ by $l$.
\end{definition}

\begin{lemma}\label{lem:loopdonthalt}
    Loops do not halt.
\end{lemma}
\begin{proof} By Definition~\ref{def:loops}, the space-time diagram of a loop is infinite, the machine does not halt.
\end{proof}

\begin{theorem}[Loops]\label{th:loops:theory} Let $\mathcal{M}$ be a Turing machine.
    Assume there is $ t_0 \in \N$ and $l \in \N^+$ such that
    for all $0 \leq i < l$: $$F_\mathcal{M}(p_i)   = F_\mathcal{M}(p'_i)$$
    with $t_0' = t_0+l$ and $p_i = (t_0+i, g(t_0+i))$ and $p_i' = (t_0'+i, g(t_0'+i))$. Then, three cases:
    \begin{enumerate}
        \item If $g(t_0') = g(t_0)$, then $\mathcal{M}$ is a loop, more specifically called a \textit{Cycler}.\label{th:case1}
        \item If $g(t_0') \geq g(t_0)$ and the tape content to the right of $p_0$ is the same as to right of $p_0'$, \ie for all nonnegative integer $z$, $f(t_0,g(t_0)+z) = f(t_0',g(t_0')+z)$, then $\mathcal{M}$ is a loop, more specifically called a (positive) \textit{Translated Cycler}.\label{th:case2}
        \item If $g(t_0') \leq g(t_0)$ and the tape content to the left of $p_0$ is the same as to left of $p_0'$, \ie for all nonnegative integer $z$, $f(t_0,g(t_0)-z) = f(t_0',g(t_0')-z)$, then $\mathcal{M}$ is a loop, more specifically called a (negative) \textit{Translated Cycler}.\label{th:case3}
    \end{enumerate}
    In these three cases, $\mathcal{M}$ has period $l$, pre-period $t_0$, and does not halt.
\end{theorem}

\begin{proof}

    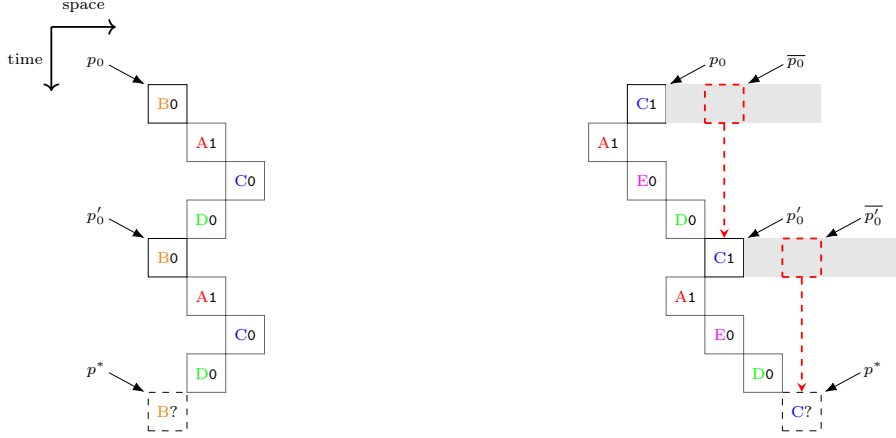
\begin{figure}[h!]
        \noindent
        \begin{minipage}[t]{0.5\textwidth}
            \centering
            \scalebox{0.85}{
                \begin{tikzpicture}[every node/.style={minimum size=6mm, draw, font=\scriptsize}, scale=1]

                    \draw[->, thick] (-1.8,3.2) -- (-1.8,2.2) node[draw=none, midway, left] {\scriptsize time};
                    \draw[->, thick] (-1.8,3.2) -- (-0.8,3.2) node[draw=none,midway, above] {\scriptsize space};
                    \node (B0) at (0,2) {\stateBx\szero};
                    \node[draw=black, draw opacity=0.6] (A1) at (0.6,1.4) {\stateAx\sone};
                    \node[draw=black, draw opacity=0.6] at (1.2,0.8) {\stateCx\szero};
                    \node[draw=black, draw opacity=0.6] at (0.6,0.2) {\stateDx\szero};

                    \node (B0') at (0,-0.4) {\stateBx\szero};
                    \node[draw=black, draw opacity=0.6] (A1') at (0.6,-1.0) {\stateAx\sone};
                    \node[draw=black, draw opacity=0.6] at (1.2,-1.6) {\stateCx\szero};
                    \node[draw=black, draw opacity=0.6] at (0.6,-2.2) {\stateDx\szero};

                    \node[dashed] (B0'') at (0,-2.8) {{\stateBx}?};

                    \node[draw=none] at ($(B0.north west)-(0.6,-0.3)+(-0.2,0.05)$) {$p_0$};
                    \draw[-{Latex[length=1.8mm]}] ($(B0.north west)-(0.6,-0.3)$) -- ($(B0.north west)-(0.05,0.00)$);

                    \node[draw=none] at ($(B0'.north west)-(0.6,-0.3)+(-0.2,0.05)$) {$p_0'$};
                    \draw[-{Latex[length=1.8mm]}] ($(B0'.north west)-(0.6,-0.3)$) -- ($(B0'.north west)-(0.05,0.00)$);

                    \node[draw=none] at ($(B0''.north west)-(0.6,-0.3)+(-0.2,0.05)$) {$p^*$};
                    \draw[-{Latex[length=1.8mm]}] ($(B0''.north west)-(0.6,-0.3)$) -- ($(B0''.north west)-(0.05,0.00)$);

                \end{tikzpicture}
            }
        \end{minipage}
        \hfill
        \begin{minipage}[t]{0.5\textwidth}
            \centering
            \scalebox{0.85}{
                \begin{tikzpicture}[every node/.style={minimum size=6mm, draw, font=\scriptsize}, scale=1]

                    \node (C1) at (0,2) {\stateCx\sone};
                    \node[draw=black, draw opacity=0.6] (A1) at (-0.6,1.4) {\stateAx\sone};
                    \node[draw=black, draw opacity=0.6] (B0) at (0,0.8) {\stateEx\szero};
                    \node[draw=black, draw opacity=0.6] (D0) at (0.6,0.2) {\stateDx\szero};

                    \foreach \i in {1,...,4} {
                            \node[minimum size=6mm, draw=none, fill=gray!20] at ($(C1)+(0.6*\i,0)$) {};
                        }

                    \node[dashed, draw=red, thick] (red) at (1.2,2) {};

                    \node (C1') at (1.2, -0.4) {\stateCx\sone};
                    \node[draw=black, draw opacity=0.6] (A1') at (0.6, -1.0) {\stateAx\sone};
                    \node[draw=black, draw opacity=0.6] (B0') at (1.2, -1.6) {\stateEx\szero};
                    \node[draw=black, draw opacity=0.6] (D0') at (1.8, -2.2) {\stateDx\szero};

                    \foreach \i in {1,...,4} {
                            \node[minimum size=6mm, draw=none, fill=gray!20] at ($(C1')+(0.6*\i,0)$) {};
                        }

                    \node[dashed, draw=red, thick] (red') at (2.4,-0.4) {};

                    \node[dashed] (C1'') at (2.4,-2.8) {\stateCx?};

                    \draw[dashed, draw=red, thick, ->, >=stealth] (red.south) -- (C1'.north);
                    \draw[dashed, draw=red, thick, ->, >=stealth] (red'.south) -- (C1''.north);

                    \node[draw=none] at ($(red.north east)+(0.6,0.3)+(0.2,0.05)$) {$\overline{p_0}$};
                    \draw[-{Latex[length=1.8mm]}] ($(red.north east)+(0.6,0.3)$) -- ($(red.north east)-(-0.05,0.00)$);

                    \node[draw=none] at ($(red'.north east)+(0.6,0.3)+(0.2,0.05)$) {$\overline{p_0'}$};
                    \draw[-{Latex[length=1.8mm]}] ($(red'.north east)+(0.6,0.3)$) -- ($(red'.north east)-(-0.05,0.00)$);

                    \node[draw=none] at ($(C1.north east)+(0.6,0.3)+(0.2,0.05)$) {$p_0$};
                    \draw[-{Latex[length=1.8mm]}] ($(C1.north east)+(0.6,0.3)$) -- ($(C1.north east)-(-0.05,0.00)$);

                    \node[draw=none] at ($(C1'.north east)+(0.6,0.3)+(0.2,0.05)$) {$p_0'$};
                    \draw[-{Latex[length=1.8mm]}] ($(C1'.north east)+(0.6,0.3)$) -- ($(C1'.north east)-(-0.05,0.00)$);

                    \node[draw=none] at ($(C1''.north east)+(0.6,0.3)+(0.2,0.05)$) {$p^*$};
                    \draw[-{Latex[length=1.8mm]}] ($(C1''.north east)+(0.6,0.3)$) -- ($(C1''.north east)-(-0.05,0.00)$);

                \end{tikzpicture}
            }
        \end{minipage}
        \caption{Illustration of Theorem~\ref{th:loops:theory}: Case 1, Cycler (left) and Case 2, positive Translated Cycler (right). In both cases, we show a head-only space-time diagram with one transcript repetition (see Section~\ref{sec:loops:algo}). Coordinates $p_0 = (t_0,g(t_0))$ correspond to the beginning of the first repetition, and $p'_0 = (t_0',g(t_0'))$ of the second. In both cases, showing that cell at position $p^*$ shares the same state as cell at position $p_0'$ is rather easy (Lemma~\ref{lem:vector}, Point~\ref{lem:vector:pt2}) while showing that they share same read-symbol (depicted as \texttt{?}) requires more work, Theorem~\ref{th:loops:theory}. In the case of Translated Cyclers, this is either done using Lemma~\ref{lem:vector}, Point~\ref{lem:vector:pt4}, depicted in Figure~\ref{fig:loop-lemma} or, using the assumption that the tape content after $p_0$ and $p_0'$ is the same, here symbolised using grey shading (right). In that case, using $\overline{p_0}$ and $\overline{p_0'}$ we show that the read-symbol at $p^*$ is the same as at $\overline{p_0'}$.}\label{fig:loop-proof}
    \end{figure}

    Consider $p^* = (t^*, g(t^*))$ with $t^*=t_0'+l$ which corresponds to one time step after $p'_{l-1}$. Figure~\ref{fig:loop-proof} illustrates the situation for the theorem's cases~\ref{th:case1} (on the left) and~\ref{th:case2} (on the right): $p^*$ is the coordinates of the black dashed cell.

    We show that $F(p^*) = F(p'_0)$. By Lemma~\ref{lem:vector} Point~\ref{lem:vector:pt2}, we know that $h(t^*) = h(t_0')$ hence we have to show that the read symbol at $p^*$ and $p'_0$ are also the same, \ie $f(p^*) = f(p'_0)$. Three cases:

    \begin{enumerate}
        \item Case $g(t_0') = g(t_0)$, illustrated in Figure~\ref{fig:loop-proof}~(left).  By Lemma~\ref{lem:vector}, Point~\ref{lem:vector:pt1}, we know that $p^*-p_0' = p'_0 - p_0$. Because $g(t_0') = g(t_0)$, the space component of $p^*-p_0'$ is 0 and we have that $g(t^*) = g(t_0')$. Hence, we know that the cell at tape position $g(t^*)$ has been visited at least once between time steps $t_0'$ and $t_0'+l-1$, which, by Lemma~\ref{lem:vector}, Point~\ref{lem:vector:pt4}, gives that $f(p^*) = f(p'_0)$.

        \item Case $g(t_0') > g(t_0)$, illustrated in Figure~\ref{fig:loop-proof}~(right). If there is a time step between $t_0'$ and $t_0'+l-1$ such that tape position $g(t^*)$ has been visited, by Lemma~\ref{lem:vector}, Point~\ref{lem:vector:pt4}, we get that $f(p^*) = f(p'_0)$.

              If there is no such time step, we have $f(p^*) = f(\overline{p_0'})$ with $\overline{p_0'} = (t_0',g(t^*))$, dashed red in Figure~\ref{fig:loop-proof}~(right). By Lemma~\ref{lem:vector}, Point~\ref{lem:vector:pt3}, there is no time step between $t_0$ and $t_0+l-1$ such that tape position $g(t_0')$ has been visited. Hence, we  get $f(p'_0) = f(\overline{p_0})$ with $\overline{p_0} = (t_0,g(t_0'))$. By hypothesis, tape content to the right of $p_0$ is the same as to the right of $p_0'$, and $p^*-p_0' = p'_0 - p_0$ (Lemma~\ref{lem:vector}, Point~\ref{lem:vector:pt1}) implies that $g(t^*)-g(t_0') = g(t_0') - g(t_0)$, hence $f(\overline{p_0}) = f(\overline{p_0'})$ and, finally, $f(p^*) = f(p'_0)$.\label{thm:proof:pt2}
        \item Case $g(t_0') < g(t_0)$, handled symmetrically to~\ref{thm:proof:pt2}.
    \end{enumerate}
    From there, we get $F(p^*) = F(p'_0)$, and the argument can be repeated starting with $p_1,\, \dots,\, p_{l-1},\, p'_0$ acting as new $p_0,\,\dots,\,p_{l-1}$ and $p'_1,\, \dots,\, p_{l-1}',\, p^*$ as new $p_0',\,\dots,\,p_{l-1}'$ by noting that for cases~\ref{th:case2} and~\ref{th:case3}, the fact that the tape is the same after $p_0$ and $p_0'$ implies that it is also the same after $p_1$ and $p_1'$. Hence, inductively, $\mathcal{M}$ is a loop and, by Lemma~\ref{lem:loopdonthalt}, it does not halt.
\end{proof}

\begin{theorem}[Coq-BB5: \texttt{Lemma loop1\_decider\_WF}]\label{th:loops}
    Let $\mathcal{M}$ be a Turing machine and $L \in \N^+$ a step-limit. \textsc{decider-loops}($\mathcal{M}$, $L$) terminates and its result is correct -- see Algorithm~\ref{alg:loops}:
    \begin{itemize}
        \item If the result is \texttt{HALT} then $\mathcal{M}$ halts from the all-0 tape.
        \item If the result is \texttt{NONHALT} then $\mathcal{M}$ does not halt from the all-0 tape.
    \end{itemize}
\end{theorem}
\begin{proof}
    The call to \textsc{decider-loops}($\mathcal{M}$, $L$) terminates as all loops are bounded. The call returns \texttt{HALT} if and only if $\mathcal{M}$ halts within $L$ steps from the all-0 tape, see Algorithm~\ref{alg:loops}, l.\ref{alg:loops:halt}, hence if the call returns \texttt{HALT} we know that the machine halts. The interesting case is the loop-detection leading to \texttt{NONHALT}.

    Algorithm~\ref{alg:loops} finds $t_0$ and $l$ satisfying the hypotheses of Theorem~\ref{th:loops:theory}: $t_0 = K-l-1-o$
    and where $F_\mathcal{M}(p_i)   = F_\mathcal{M}(p'_i)$ is guaranteed thanks to Algorithm~l.\ref{alg:loops:testeq}. Theorem~\ref{th:loops:theory} Cases~1, 2, and, 3 are respectively handled by Algorithm~l.\ref{alg:loops:cycler}, l.\ref{alg:loops:tcplus}, and, l.\ref{alg:loops:tcminus}. In Case 2/Case 3, the condition of having tapes be the same to the right/left of $p_0$ and $p_0'$ is handled by making sure the head is at the maximum/minimum seen position of the tape in both cases, ensuring that there are only 0s to the right/left, and therefore satisfying the condition. Hence, we get that $\mathcal{M}$ does not halt.
\end{proof}

\subsubsection{Implementations and results}\label{sec:loops:results}

\begin{table}[h!]
  \centering
  \begin{tabular}{|l|rrr|}
    \hline
    Step-limit parameter $L$ & Nonhalt     & Halt       & Total decided \\
    \hline
    130                      & 126,950,828 & 48,367,435 & 175,318,263   \\
    4100                     & 43,269      & 12,276     & 55,545        \\
    1,050,000                & 2           & 0          & 2             \\ \hline
    Total                    & 126,994,099 & 48,379,711 & 175,373,810   \\
    \hline
  \end{tabular}
  \caption{Machines decided by using the loop deciders (Algorithm~\ref{alg:loops}) in the $S(5)$ pipeline (Table~\ref{tab:pipelineBB5}) per step-limit parameter $L$.}\label{tab:paramsLoops}
\end{table}

The decider for loops, Algorithm~\ref{alg:loops}, is implemented as part of Coq-BB5 (function \texttt{loop1\_decider}). As advertised in the $S(5)$ pipeline (Table~\ref{tab:pipelineBB5}), it decides a very important proportion of the enumerated 5-state Turing machines: 95.48\% of the nonhalting machines and more than 99.99\% of the halting ones and this with fairly low step-limit parameters, see Table~\ref{tab:paramsLoops}. This means, for instance, that 99.99\% of the enumerated 5-state halting machines halt before $4{,}100$ steps.

\begin{figure}
  \centering
  \includegraphics[width=0.35\textwidth]{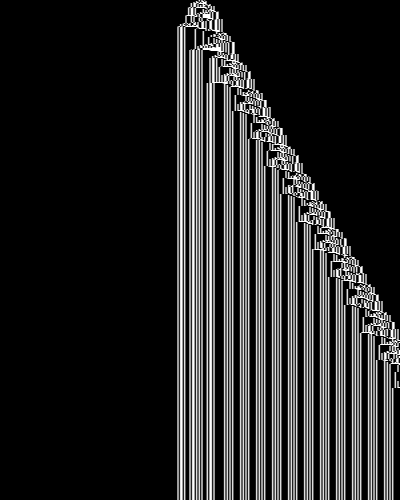}

  \caption{10,000-step space-time diagram of a Translated Cycler not decided by the decider for loops in \CoqBB (it is decided by NGramCPS, see Section~\ref{sec:n-gramCPS}). See \tm{1RB0LE_1LC0RD_---1LD_1RE0LA_1LA0RE}.}\label{fig:translated-cyclers-more}
\end{figure}

The number of nonhalting machines decided by this decider in \CoqBB (\ie $126{,}994{,}099$, see Table~\ref{tab:paramsLoops}) is a lower bound of the actual number of 5-state loops, for instance here are two loops decided by other means:

\begin{enumerate}
  \item Figure~\ref{fig:translated-cyclers-more} gives a \TC that is decided by the $n$-gram Closed Position Set (NGramCPS) decider, see Section~\ref{sec:n-gramCPS}. Higher step-limit $L$ would have been needed to be detected by Algorithm~\ref{alg:loops}.
  \item The Sporadic machine (\ie machine which required an individual proof of nonhalting) named ``Skelet \#1'', see Section~\ref{sec:sporadic}, is a \TC but with enormous parameters: it does not start looping before $5.41 \times 10^{51}$ steps and has a period of more than 8 billion steps \cite{ShawnSkelet1}. There is no reasonable step-limit $L$ for which this machine would have been decided by the decider for loops, neither in fact by any of the deciders presented in this work which all more or less rely on step-by-step simulation. An individual proof of nonhalting was required \cite{busycoq}.
\end{enumerate}

In this sample of $126{,}994{,}099$ nonhalting loops we find approximately 86\% \TCs and 14\% \cyclers which suggests that, in general, \TCs are much more common than \cyclers.

\textbf{Other implementation.} Algorithm~\ref{alg:loops} also has a Python implementation.\footnote{\url{https://github.com/bbchallenge/bbchallenge-deciders/tree/main/decider-loops-reproduction}}

\newcommand{\ngramcps}{NGramCPS\xspace}

\subsection{$n$-gram Closed Position Set (\ngramcps)}\label{sec:n-gramCPS}

\newcommand{\leftngram}{left\xspace}
\newcommand{\rightngram}{right\xspace}
\newcommand{\middlesymbol}{middle\xspace}

The $n$-gram Closed Position Set (\ngramcps) decider which we introduce here, Algorithm~\ref{alg:NGramCPS}, is a simplification of an earlier technique, Closed Position Set (CPS), itself introduced in \texttt{bbfind} \cite{Skelet_bbfind}, see Section~\ref{sec:intro:mainresults}. Surprisingly, \ngramcps is a relatively simple technique which makes a potent decider as it decides 99.89\% of all nonhalting enumerated 5-state machines excluding loops, see Table~\ref{tab:pipelineBB5}.

The method is especially potent when augmenting the binary alphabet of Turing machines to record extra information on the tape, such as a fixed-length history of previously seen (state,symbol) pairs, see Section~\ref{sec:n-gramCPS:augmentations}. \ngramcps was first developed without augmentations \cite{ngramcps_fenner} which were later introduced with \CoqBB.

\subsubsection{Algorithm}\label{sec:n-gramCPS:algo}

\begin{algorithm}
    \caption{{\sc decider-NGramCPS}}\label{alg:NGramCPS}

    \begin{algorithmic}[1]
        \State{\textbf{Input:} A Turing machine $\mathcal{M}$, the zero symbol of the alphabet $\alphabet_0$, the size of the n-grams $n > 0$.}
        \State{\textbf{Output:} \NONHALT if the decider detects that the machine does not halt and \UNKNOWN otherwise.}

        \State $g_0 = (\alphabet_0)^n$ \Comment{The zero n-gram consists of $n$ zero symbols}
        \State $L = \{ g_0 \}$ \Comment{The seen left n-grams}
        \State $R = \{ g_0 \}$ \Comment{The seen right n-grams}
        \State $C =$ \{\{.\leftngram $=$ $g_0$, .\rightngram $=$ $g_0$, .state $=$ \stateA, .\middlesymbol $=$ $\alphabet_0$ \}\} \Comment{The seen local configurations}
        \While{true}\label{alg:NGramCPS:line:whileTrue}
        \State $V = C$
        \State any\_updates $=$ false
        \While{$|V| \neq 0$}
        \State $c = V.\textbf{pop}()$ \Comment{Remove an arbitrary element $c$ from $V$}
        \State $c' = c$
        \State $\{w,d,s\}$ $=$ $\mathcal{M}$(c.state, c.\middlesymbol) \Comment{Transition's write symbol, move direction, and next state}
        \State \If{$s$ is undefined} \Comment{Undefined transition is met, we cannot conclude}
        \State \textbf{return} UNKNOWN\label{alg:NGramCPS:line:unknown}
        \EndIf
        \State \If{$d$ is Right}
        \State Insert $c.\text{\leftngram}$ in $L$ \label{alg:NGramCPS:line:insertL}
        \State Set $c'.\text{\leftngram}$ to the last $n-1$ symbols of $c.\text{\leftngram}$ followed by $w$
        \State Set $c'.\text{\middlesymbol}$ to the first symbol of $c.\text{\rightngram}$
        \For{each ngram $r\in R$ starting with the last $n-1$ symbols of $c.\text{\rightngram}$}
        \State Set $c'.\text{\rightngram}$ to $r$
        \If{$c'$ is not in $C$}
        \State \tabi Insert $c'$ in $C$ \label{alg:NGramCPS:line:insertInConfSet}
        \State \tabi Insert $c'$ in $V$ \label{alg:NGramCPS:line:insertInConfSetToVisit}
        \State \tabi any\_updates $=$ true
        \EndIf
        \EndFor
        \EndIf
        \State \If{$d$ is Left} \label{alg:NGramCPS:line:moveLeft}
        \State Insert $c.\text{\rightngram}$ in $R$
        \State Set $c'.\text{\rightngram}$ to the first $n-1$ symbols of $c.\text{\rightngram}$ preceded by $w$
        \State Set $c'.\text{\middlesymbol}$ to the last symbol of $c.\text{\leftngram}$
        \For{each ngram $l \in L$ ending with the first $n-1$ symbols of $c.\text{\leftngram}$}
        \State Set $c'.\text{\leftngram}$ to $l$
        \If{$c'$ is not in $C$}
        \State \tabi Insert $c'$ in $C$
        \State \tabi Insert $c'$ in $V$
        \State \tabi any\_updates $=$ true
        \EndIf
        \EndFor
        \EndIf
        \EndWhile
        \State \If{\textbf{not} any\_updates}
        \State \textbf{return} \NONHALT\label{alg:NGramCPS:line:nonhalt} \Comment{Set $C$ is closed, and does not include undefined transitions:\\ \tabi \tabi \tabi \tabi \tabi \tabi \tabi \tabi \tabi \tabi \tabi \space \space \space \space the machine does not halt}
        \EndIf

        \EndWhile
    \end{algorithmic}
\end{algorithm}

Algorithm~\ref{alg:NGramCPS} gives a pseudo-code of the \ngramcps decider. The decider considers finite, \textit{local configurations} of a Turing machine consisting of: (i) the \textit{$n$-grams} (see after) respectively to the left and to the right of the head; (ii) the state the machine is in; (iii) the symbol currently read by the head, referred to as \textit{middle} symbol (as opposed to the left and right part of the tape, modelled by the n-grams). By \textit{$n$-gram}, we mean a sequence of $n > 0$ symbols from the tape alphabet (for instance, the binary alphabet $\alphabet=\{\szero,\sone\}$).

The algorithm builds a set of local configurations \textit{potentially} reachable by the machine until either an undefined transition is met (Algorithm~\ref{alg:NGramCPS}, l.\ref{alg:NGramCPS:line:unknown}) or no new configurations are added to the set, i.e. the set is closed under Turing machine operations (Algorithm~\ref{alg:NGramCPS}, l.\ref{alg:NGramCPS:line:nonhalt}). In the first case, the decider cannot conclude and the machine is left undecided. In the second case, the decider concludes that the machine does not halt as no undefined transition (\ie where the machine could be asked to halt) can be reached, Theorem~\ref{th:ngramcps}, this is a CTL argument, see Section~\ref{sec:deciders-overview}.

The central idea of this decider and the reason behind using the ``$n$-gram'' terminology (originating from \textit{$n$-gram models} in language analysis)  is better illustrated by the following example. Let $n=3$ and consider local configuration \texttt{011 [B0] 100}, meaning that the left $n$-gram is \texttt{011}, right $n$-gram is \texttt{100}, the machine is in state \stateB and reading symbol \texttt{0}. Assume that the  machine's transition for reading a \texttt{0} in state \stateB is \texttt{1RC}, meaning that the machine writes \texttt{1}, moves right and transitions to state \stateC. The local configuration becomes \texttt{011 1 [C1] 00?}, where \texttt{?} means that we do not know which symbol to use. Then:

\begin{enumerate}
    \item \textbf{Left $n$-gram update.} We record the left $n$-gram \texttt{011} as seen (it is inserted in set $L$, Algorithm~\ref{alg:NGramCPS}, l.\ref{alg:NGramCPS:line:insertL}) and we discard its first bit, updating the left $n$-gram to \texttt{111}. The local configuration becomes \texttt{111 [C1] 00?}.
    \item \textbf{Right $n$-gram update.} In order to deal with the unknown symbol \texttt{?}, we look among the previously seen right $n$-grams (contained in set $R$ in Algorithm~\ref{alg:NGramCPS}) the ones that start by \texttt{00}. For instance, let us assume it is \texttt{000} and \texttt{001}. Then we add both local contexts \texttt{111 [C1] 000} and \texttt{111 [C1] 001}, if not already in: (a) to our set of local configurations (Algorithm~\ref{alg:NGramCPS}, l.\ref{alg:NGramCPS:line:insertInConfSet}), and (b) to our set of configurations to visit (Algorithm~\ref{alg:NGramCPS}, l.\ref{alg:NGramCPS:line:insertInConfSetToVisit}) in order to repeat this procedure (or symmetrical when the machine moves left) on them.
\end{enumerate}

The algorithm systematically revisits all previously added local configurations, in case they contain a right/left $n$-gram that was newly met (Algorithm~\ref{alg:NGramCPS}, l.\ref{alg:NGramCPS:line:whileTrue}). Assuming a finite tape alphabet (which we always do in this work), the algorithm will eventually terminate since the number of possible local configurations is finite. In practice, one may add a limit on the number of iterations to avoid long computations.

\begin{theorem}[\CoqBB: \texttt{Lemma NGramCPS\_decider\_spec}]\label{th:ngramcps}
    Let $\mathcal{M}$ be a Turing machine using tape alphabet $\alphabet$ containing zero symbol $\alphabet_0$ and let $n \in \N^+$ be the $n$-gram length parameter. \textsc{decider-NGramCPS}($\mathcal{M}$, $\alphabet_0$, $n$) terminates and its result is correct -- see Algorithm~\ref{alg:NGramCPS}: if it returns \NONHALT then $\mathcal{M}$ does not halt from the all-$\alphabet_0$ tape.
\end{theorem}
\begin{proof}
    Algorithm~\ref{alg:NGramCPS} is guaranteed to terminate because either an undefined transition is eventually met (Algorithm~\ref{alg:NGramCPS}, l.\ref{alg:NGramCPS:line:unknown}) or because the set of local configuration -- which is bounded by the finite set of all possible local configurations -- is saturated (Algorithm~\ref{alg:NGramCPS}, l.\ref{alg:NGramCPS:line:nonhalt}).

    By construction, Algorithm~\ref{alg:NGramCPS} overestimates the set of all local configurations reached by the machine from the all-$\alphabet_0$ tape, \ie it contains at least all the reached local configurations and potentially more. If this set contains no local configuration leading to an undefined transition, we are assured that the machine does not halt, Algorithm~\ref{alg:NGramCPS}, l.\ref{alg:NGramCPS:line:nonhalt}. This is a CTL argument, see Section~\ref{sec:deciders-overview}.
\end{proof}

\subsubsection{Tape alphabet augmentations}\label{sec:n-gramCPS:augmentations}

The \ngramcps decider becomes particularly powerful for deciding 5-state 2-symbol Turing machines when augmenting the 2-symbol alphabet to store more information on the tape. Two augmentations are used in \CoqBB:

\begin{enumerate}
    \item \textbf{Fixed-length history.} In this variant, tape symbols encode the current binary symbol on a cell as well as a fixed-length list of (state,binary symbol) pairs previously seen on the cell. For instance, if the non-augmented machine currently reads binary symbol \texttt{1} and the machine has previously visited the cell in state \stateA reading symbol \texttt{0} and before that in state \stateB reading symbol \texttt{1}, in the augmented machine, the cell will contain the augmented symbol ``\texttt{1, [(A,0),(B,1)]}''. If the history length is set to 2 and the machine was in state \stateC when reading ``\texttt{1, [(A,0),(B,1)]}'' the cell will be updated to ``\texttt{0, [(C,1),(A,0)]}'', assuming the transition of the machine for reading a \texttt{1} in state \stateC requires to write symbol \texttt{0}. The zero-symbol for this augmentation $\alphabet_0$ is ``\texttt{0, []}''. Furthermore, it is easy to verify that if the decider returns \NONHALT for a fixed-length augmented machine, then the non-augmented machine does not halt.
    \item \textbf{Least Recent Usage history (LRU).} In this variant, tape symbols encode the set of state-symbol pairs seen at that cell,
          in order of when it was seen last, the most recent first.  For instance, if the non-augmented machine currently reads binary symbol \texttt{1} and the machine has previously visited the cell in state \stateD reading symbol \texttt{1} and before that in state \stateC reading symbol \sone, and before that in state \stateD reading symbol \texttt{0}, in the augmented machine, the cell will contain the augmented symbol ``\texttt{1, [(D,1),(C,1),(D,0)]}''. Assume the augmented machine was in state \stateC when reading ``\texttt{1, [(D,1),(C,1),(D,0)]}'' the cell will be updated to ``\texttt{0, [(C,1),(D,1),(D,0)]}'', (assuming the transition of the machine for reading a \texttt{1} in state \stateC writes symbol \texttt{0}) with pair \texttt{(C,1)} bubbling up to the beginning of the LRU history. The zero-symbol for this augmentation $\alphabet_0$ is also ``\texttt{0, []}''. Similarly to above, one can verify that if the decider returns \NONHALT for an LRU augmented machine, then the non-augmented machine does not halt. One fundamental difference with the fixed-length history augmentation is that here, the history is not of fixed length but is bounded by number of states times number of symbols, \ie 10 in the case of $S(5)$.
\end{enumerate}

\subsubsection{Implementations and results}\label{sec:n-gramCPS:results}
\begin{table}[h!]
    \centering
    \begin{tabular}{|l|r|}
        \hline
        Variant                                              & Nonhalt   \\  \hline
        NGram-CPS without augmentation                       & 5,117,863 \\
        NGram-CPS augmented using fixed-length history       & 887,093   \\
        NGram-CPS augmented using Least Recent Usage history & 182       \\ \hline
        Total decided                                        & 6,005,138 \\
        \hline
    \end{tabular}
    \caption{\ngramcps results in the $S(5)$ pipeline (see Table~\ref{tab:pipelineBB5}) per variant (see Section~\ref{sec:n-gramCPS:augmentations}).}\label{tab:ngramcps:results}
\end{table}

\CoqBB implements \ngramcps (Algorithm~\ref{alg:NGramCPS}) in the three variants discussed here, (i) without augmentation (function \texttt{NGramCPS\_decider\_impl2}) -- \ie using standard binary alphabet $\alphabet=\{\szero,\sone\}$; (ii) fixed-length history (function \texttt{NGramCPS\_decider\_impl1}); (iii) Least Recent Usage history (function \texttt{NGramCPS\_LRU\_decider}). Compared to Algorithm~\ref{alg:NGramCPS}, \CoqBB implementations integrate an additional parameter allowing them to terminate early for the sake of performance. The implementations for (ii) and (iii) use the same core implementation as for (i) just accordingly augmenting the tape-alphabet of the machine and its read/write behaviour (see definitions \texttt{TM\_history} and \texttt{TM\_history\_LRU}).

Altogether, \ngramcps decides 99.89\% of all nonhalting enumerated 5-state machines excluding loops, see Table~\ref{tab:pipelineBB5}. The number of machines decided by each \ngramcps variant in the $S(5)$ pipeline (Table~\ref{tab:pipelineBB5}) are given in Table~\ref{tab:ngramcps:results}. Augmentations allowed to decide machines that resisted all other methods, without having to resort to individual proofs of nonhalting -- see details in the full $S(5)$ pipeline, Appendix~\ref{app:pipelines}.

Figure~\ref{fig:ngram-cps-more} gives an example of a ``fractal-looking'' 5-state Turing machine that is solved by the LRU augmentation but has no known solution with standard \ngramcps or the fixed-length history augmentation.

\begin{figure}
    \centering
    \includegraphics[width=0.35\textwidth]{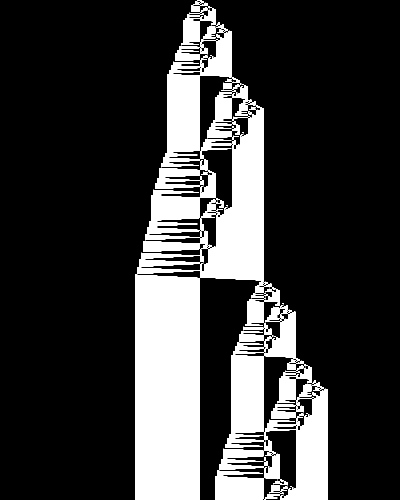}

    \caption{10,000-step space-time diagram of a ``fractal-looking'' 5-state Turing machine that is solved by the LRU augmentation but has no known solution with standard \ngramcps or the fixed-length history augmentation, see Section~\ref{sec:n-gramCPS:augmentations}. \tm{1RB0RA_1LC---_1RC1LD_0LE1RA_0LC0LE}}\label{fig:ngram-cps-more}
\end{figure}

\newpage

\newpage

\begin{figure}[h!]
    \centering
    \includegraphics[scale=0.8]{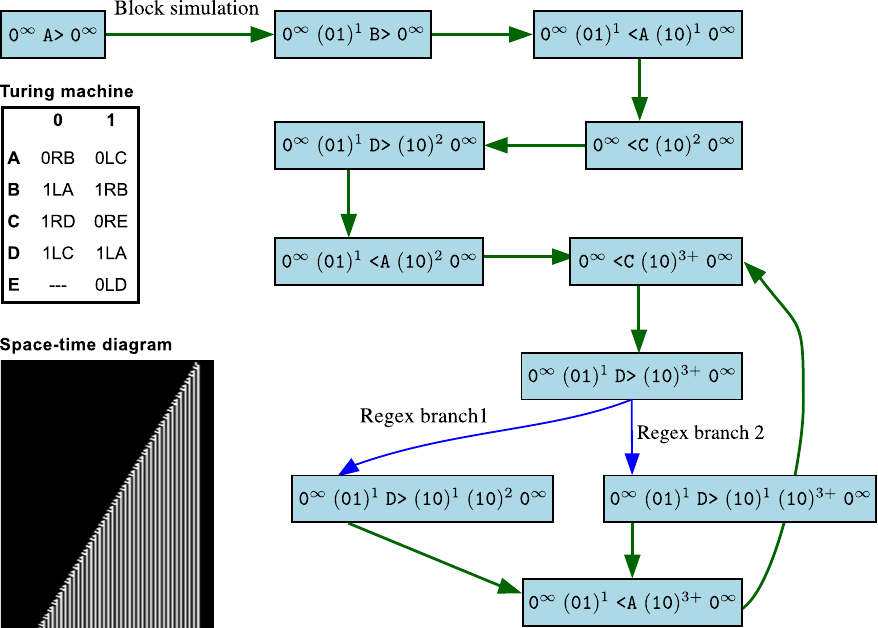}
    \caption{{\small \textbf{RepWL graph.} Closed graph of regex configurations constructed by the Repeated Word List (RepWL) method (Section~\ref{sec:RepWL}) for machine \tm{0RB0LC_1LA1RB_1RD0RE_1LC1LA_---0LD} with block length $l=2$ and repeat threshold $T=3$. Block simulation and regex branching steps (see Section~\ref{sec:RepWL}) are illustrated using respectively green and blue arrows. As illustrated by its 300-step space-time diagram, the machine is a simple Translated Cycler which can be easily handled by Algorithm~\ref{alg:loops}, but, because of its very small graph it is convenient to use this machine for illustrative purposes. Because the graph is closed and contains no halting configuration, the machine does not halt, Theorem~\ref{th:repwl}.}}\label{fig:repWL}
\end{figure}

\subsection{Repeated Word List (RepWL)}\label{sec:RepWL}

\begin{figure}[h!]
    \centering
    \includegraphics[scale=0.48]{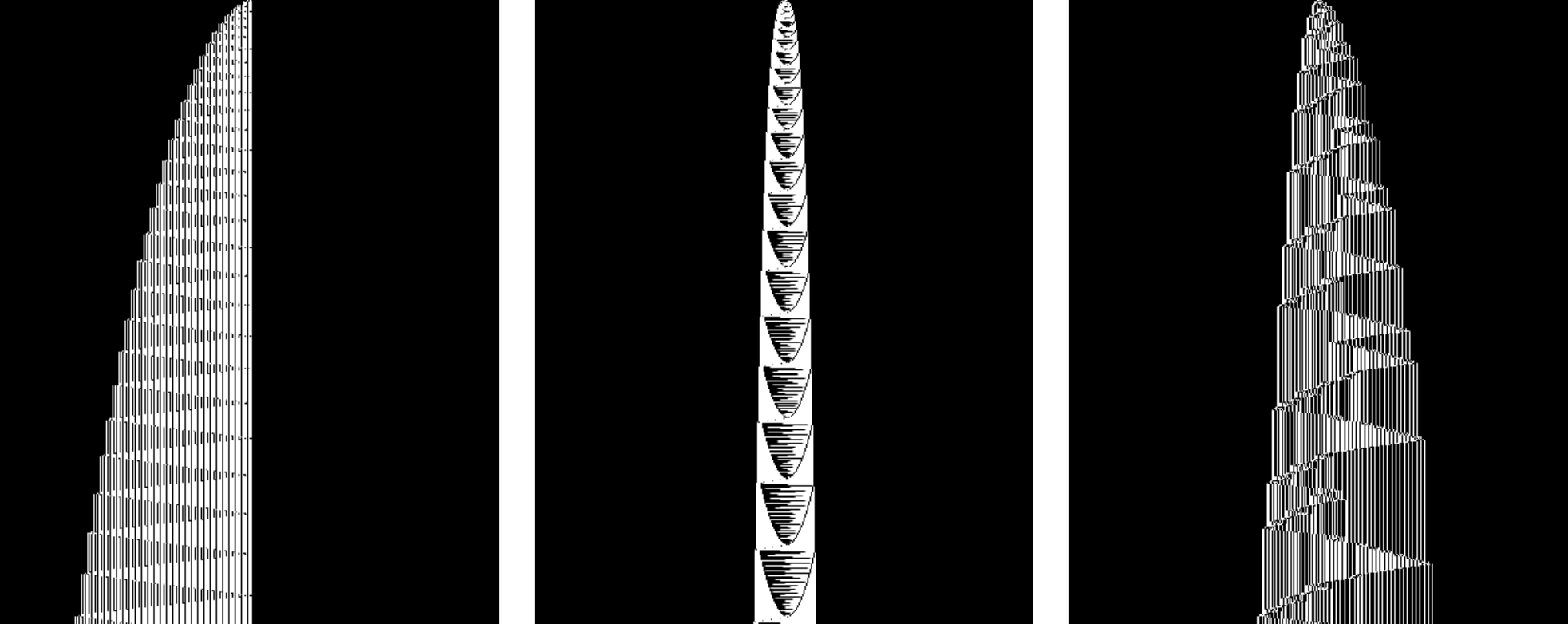}
    \caption{10,000-step space-time diagrams of three 5-state machines decided by the Repeated Word List (RepWL) decider, Algorithm~\ref{alg:RepWL}. Left: {\small \tm{1RB0RD_0LC0LA_0LD1LC_1RA0LE_0RC---}}. Center: {\small \tm{1RB---_1LB1RC_1RA1RD_1LE0RD_0LB0LC}}. Right: {\small \tm{1RB---_0RC1RD_0LD1RC_1LE0RA_1RA0LE}}. The RepWL graphs of these machines respectively have $42$, $845$, and $143{,}181$ nodes, ranging the entire distribution of RepWL node counts for 5-state machines, see Section~\ref{sec:RepWL:results}. RepWL parameters $(l,T)$ for these machines are respectively: $(5,2)$, $(2,3)$, and, $(20,2)$.
    }\label{fig:repWLThree}
\end{figure}

\subsubsection{Algorithm}

The Repeated Word List (RepWL) technique, introduced in \CoqBB, is based on the following simple idea: if a word (or \textit{block}) of length $l > 0$ appears consecutively on the tape more than $T > 0$ times (with $l, T \in \mathbb{N}$ fixed) then, we assume it may repeat an unbounded number of times in the future. In practice, it means we represent configurations as regular expressions and call them \textit{regex configurations}. For instance, consider the following configuration:
$$ \texttt{0}^\infty \; \texttt{11100 A> 111101010111111111} \; \texttt{0}^\infty$$
Using block length $l=2$, by grouping symbols from the head outwards and symbols $\szero$ drawn from $\texttt{0}^\infty$ if needed, we get:
$$ \texttt{0}^\infty \; (\texttt{01}) \; (\texttt{11}) \; (\texttt{00}) \; \texttt{A>} \; (\texttt{11})^2 \; (\texttt{01})^{3} \; (\texttt{11})^{4} \; \texttt{0}^\infty $$
And, using \textit{repeat threshold} $T=3$ we get the following regex configuration:
$$ \texttt{0}^\infty \; (\texttt{01}) \; (\texttt{11}) \; (\texttt{00}) \; \texttt{A>} \; (\texttt{11})^2 \; (\texttt{01})^{3+} \; (\texttt{11})^{3+} \; \texttt{0}^\infty $$

Any repetition of more than $T$ times the same word $w \in \{\szero,\sone\}^l$ is replaced by the regular expression $(w)^{T+}$ meaning that word $w$ is repeated at least $T$ times, hence the only exponents to ever be used in this representation are $\{1,2,\dots,T-1\}$ and $T+$. We call $T$ the repeat threshold. Note that here, we use \textit{directional head notation} for Turing machines, where the head lives in between cells and points either right or left. This framework is equivalent to the Turing machines setup used elsewhere in this work, see Section~\ref{sec:TMs}.

\paragraph{RepWL graph.} Using the rules explained below (\textit{block simulation} and \textit{regex branching}), {\sc decider-RepWL} (Algorithm~\ref{alg:RepWL}) simulates Turing machines directly on these regex configurations starting from the initial configuration (\ie $\texttt{0}^\infty \; \texttt{A>} \; \texttt{0}^\infty$), as to create a graph of such regex configurations to explore. If this graph is eventually closed (Algorithm~\ref{alg:RepWL}~l.\ref{alg:RepWL:closed}) and contains no halting configuration then we know that the machine will never halt, Theorem~\ref{th:repwl}, this is a CTL argument, see Section~\ref{sec:deciders-overview}. Because there is no guarantee the graph is finite, in order to force termination, we also need an additional parameter, named $N$ in Algorithm~\ref{alg:RepWL}, indicating how many distinct nodes we are willing to visit at most. Figure~\ref{fig:repWL} gives the RepWL graph of a simple machine.

For simulating Turing machines on regex configurations we need to deal with two cases: (i) \textit{block simulation} when the head is facing a constant block (\ie block without a $+$), such as $\texttt{A>} \; (\texttt{11})^2$ and (ii) \textit{regex branching} when the head is facing a block with a $+$, \eg $\texttt{D>} \; (\texttt{01})^{3+}$.

\paragraph{Block simulation.} When the head is facing a constant block, such as in the above example $\texttt{A>} \; (\texttt{11})^2$ (or if the head is facing $\texttt{0}^\infty$, we add constant block $(\szero^l)^1$), we can proceed to \textit{block simulation}.
Block simulation consists of simulating the Turing machine until the head eventually leaves the block or until a maximum step limit is reached (parameter named $B$ in Algorithm~\ref{alg:RepWL}) or until the machine halts. Note that the TM may never leave the block if it enters an infinite cycle which is why we need the step limit -- one could alternatively implement cycle detection (Section~\ref{sec:loops}) in block simulation but it is not the route taken in \CoqBB. Performing block simulation from $\texttt{A>} \; (\texttt{11})^2$ could produce, for instance, $\texttt{00}\; \texttt{00} \; \texttt{B>}$ or $\texttt{<C} \; \texttt{10} \; \texttt{11} $ or enter a cycle and never leave the block, depending on the Turing machine being simulated. After block simulation, identical contiguous blocks are regrouped into powers, \eg $\texttt{00} \; \texttt{00} \; \texttt{B>}$ becomes $(\texttt{00})^2 \; \texttt{B>}$ and, assuming $T=3$, the tape $(\texttt{10})^2\; (\texttt{10})^1\; \texttt{B>}$ would become $(\texttt{10})^{3+}\; \texttt{B>}$. In Figure~\ref{fig:repWL}, block simulation transforms $\texttt{0}^\infty \; \texttt{A>} \; \texttt{0}^\infty$ into $\texttt{0}^\infty \; (\texttt{01})^1 \; \texttt{B>} \; \texttt{0}^\infty$, see Example~\ref{ex:repWL} for details.

\paragraph{Regex branching.} When the head is facing a block with a $^+$, for instance in Figure~\ref{fig:repWL} we have $\texttt{0}^\infty \; \texttt{01}^1 \; \texttt{D>} \; (\texttt{01})^{3+} \; \texttt{0}^\infty$, from which we add two configurations to the set of configurations to visit next:
\begin{enumerate}
    \item \textbf{Regex branch 1.} We visit $\texttt{0}^\infty \; \texttt{01}^1 \; \texttt{D>} \; (\texttt{01})^1 \; (\texttt{01})^{2} \; \texttt{0}^\infty$.
    \item \textbf{Regex branch 2.} We visit $\texttt{0}^\infty \; \texttt{01}^1 \; \texttt{D>} \; (\texttt{01})^1 \; (\texttt{01})^{3+} \; \texttt{0}^\infty$.
\end{enumerate}
In both cases, we have reduced to block simulation.

\begin{example}\label{ex:repWL}
    Figure~\ref{fig:repWL} gives the RepWL graph for machine \tm{0RB0LC_1LA1RB_1RD0RE_1LC1LA_---0LD}\footnote{This machine is a simple Translated Cycler (Section~\ref{sec:loops}), but it was chosen because its RepWL graph is small.} for block length $l=2$ and repeat threshold $T=3$. The first edge of the graph reads $\texttt{0}^\infty\; \texttt{A>}\;\texttt{0}^\infty$ goes to $\texttt{0}^\infty\; (\texttt{01})^1\; \texttt{B>}\; \texttt{0}^\infty$ using block simulation, this is because, using our parameters, $\texttt{0}^\infty\; \texttt{A>}\;\texttt{0}^\infty$ rewrites as $\texttt{0}^\infty\; \texttt{A>}\; (\texttt{00})\;\texttt{0}^\infty$ and we have $\texttt{0}^\infty\; \texttt{A>}\; (\texttt{00})\;\texttt{0}^\infty \TMstep \texttt{0}^\infty\;  \texttt{0}\; \texttt{B>}\;\texttt{0}\;\texttt{0}^\infty \TMstep \texttt{0}^\infty\;  \texttt{A>}\texttt{0}\;\texttt{1}\;\texttt{0}^\infty \TMstep \texttt{0}^\infty\;  \texttt{0}\;\texttt{B>}\texttt{1}\;\texttt{0}^\infty \TMstep  \texttt{0}^\infty\; (\texttt{01})^1\; \texttt{B>}\; \texttt{0}^\infty$, as needed. The graph contains only one case of regex branching, on regex configuration $\texttt{0}^\infty\; (\texttt{01})^1\; \texttt{D>}\; (\texttt{10})^{3+}\; \texttt{0}^\infty$, all the other edges are instances of block simulation. The graph is closed and contains no halting configuration, the machine does not halt, Theorem~\ref{th:repwl}.
\end{example}

\begin{theorem}[\CoqBB: \texttt{Lemma RepWL\_ES\_decider\_spec}]\label{th:repwl}
    Let $\mathcal{M}$ be a Turing machine, $l\in \N^+$ the block-length parameter, $T \in \N^+$ the repeat threshold, $B \in \N$ the maximum number of steps allowed in block simulation and $N \in \N$ the maximum number of nodes we are willing to visit. Then, \textsc{decider-RepWL}($\mathcal{M}$, $l$, $T$, $B$, $N$) terminates and its result is correct -- see Algorithm~\ref{alg:RepWL}: if it returns \NONHALT then $\mathcal{M}$ does not halt from the all-$0$ tape.
\end{theorem}
\begin{proof}
    The algorithm terminates thanks to parameters $B$ and $N$. For a machine $\mathcal{M}$, the algorithm returns $\NONHALT$ (Algorithm~l.\ref{alg:RepWL:closed}) iff the RepWL graph of $\mathcal{M}$ contains less than $N$ nodes (\ie is closed), and contains no halting configuration (Algorithm~l.\ref{alg:RepWL:fail}). Since the set of configurations reached by the machine is a subset of the regular language consisting of the union of each node's regex configuration, which includes no halting configuration, we get that the machine cannot halt from the all-0 tape. This is a CTL argument, see Section~\ref{sec:deciders-overview}.
\end{proof}

\begin{algorithm}
    \caption{{\sc decider-RepWL}}\label{alg:RepWL}

    \begin{algorithmic}[1]
        \State{\textbf{Input:} A Turing machine $\mathcal{M}$, block length parameter $l>0$, repeat threshold $T>0$, maximum number of steps allowed in block simulation $B \in \mathbb{N}$, maximum number of distinct nodes we are willing to visit $N\in \mathbb{N}$.}
        \State{\textbf{Output:} \NONHALT if the decider detects that the machine does not halt and \UNKNOWN otherwise.}
        \State
        \State $\texttt{to\_visit} = [\texttt{A>}]$
        \State $V = \{\}$ \Comment{Visited regex configurations}
        \State \While{$|V| < N$ \textbf{and} $|\texttt{to\_visit}| \neq 0$}
        \State $\texttt{regex\_config} = \texttt{to\_visit}.\textbf{pop}()$
        \State \If{$\texttt{regex\_config}$ is in $V$}
        \State \textbf{continue}
        \EndIf
        \State
        \State Insert $\texttt{regex\_config}$ in $V$
        \State \If{head is facing a constant block}
        \State $\texttt{new\_regex\_config} = \texttt{regex\_config}.\textbf{block\_simulation}(B)$
        \If{\texttt{new\_regex\_config} has halted (\ie undefined transition was met) \textbf{or} \\ $\quad \quad \quad \;\,\,\,$ limit $B$ was exceeded during block simulation}
        \State \Return \UNKNOWN \label{alg:RepWL:fail}
        \EndIf
        \State $\texttt{to\_visit}.\textbf{append}(\texttt{new\_regex\_config})$
        \Else \Comment{Head is facing a block with a $+$}
        \State $\texttt{regex\_config\_1}, \; \texttt{regex\_config\_2} = \texttt{regex\_config}.\textbf{regex\_branching}(M)$
        \State $\texttt{to\_visit}.\textbf{append}(\texttt{regex\_config\_1})$
        \State $\texttt{to\_visit}.\textbf{append}(\texttt{regex\_config\_2})$
        \EndIf
        \EndWhile

        \State \If{$|V| < N$}
        \State \Return{\NONHALT}\label{alg:RepWL:closed}
        \Else
        \State \Return{\UNKNOWN}
        \EndIf
    \end{algorithmic}
\end{algorithm}

\subsubsection{Implementations and results}\label{sec:RepWL:results}

\CoqBB implements RepWL (Algorithm~\ref{alg:RepWL}), see function \texttt{RepWL\_ES\_decider}. In the $S(5)$ pipeline (Table~\ref{tab:pipelineBB5}), contrarily to previously presented deciders, RepWL is not applied in bulk using generic parameters. Instead, the ${6,577}$ machines it decides are hardcoded in the proof together with the specific $l$ and $T$ parameters that decide them, see file \texttt{Decider\_RepWL\_Hardcoded\_Parameters.v}. These parameters were found using a grid search in C++. Block length varies between 1 and 38, and repeat threshold between 2 and 4. For all these machines, maximum block simulation parameters and maximum number of graph nodes parameters are set to 320 and ${150{,}001}$, respectively.

Figure~\ref{fig:repWLThree} gives space-time diagrams for machines with RepWL graphs with 42 (left), ${845}$ (center) and $143{,}181$ nodes which are respectively the minimum, average and maximum sizes of RepWL graphs constructed by \CoqBB for solving $S(5)$.

The machines on the left and on the right fit in the zoological category of \textit{Bouncers}, see Section~\ref{sec:zoo}. We developed a dedicated decider for solving Bouncers, but it was not used in \CoqBB \cite{bbchallenge_part1}.

In the $S(4)$ pipeline (Table~\ref{tab:pipelineBB4}), RepWL is only used to decide two machines using parameters $l=4$ and $T=3$. These 4-state machines are given in Figure~\ref{fig:RepWLBB4}, and they respectively have $3{,}130$ and $3{,}076$ nodes in their RepWL graph.

\begin{figure}
    \centering
    \includegraphics[scale=0.48]{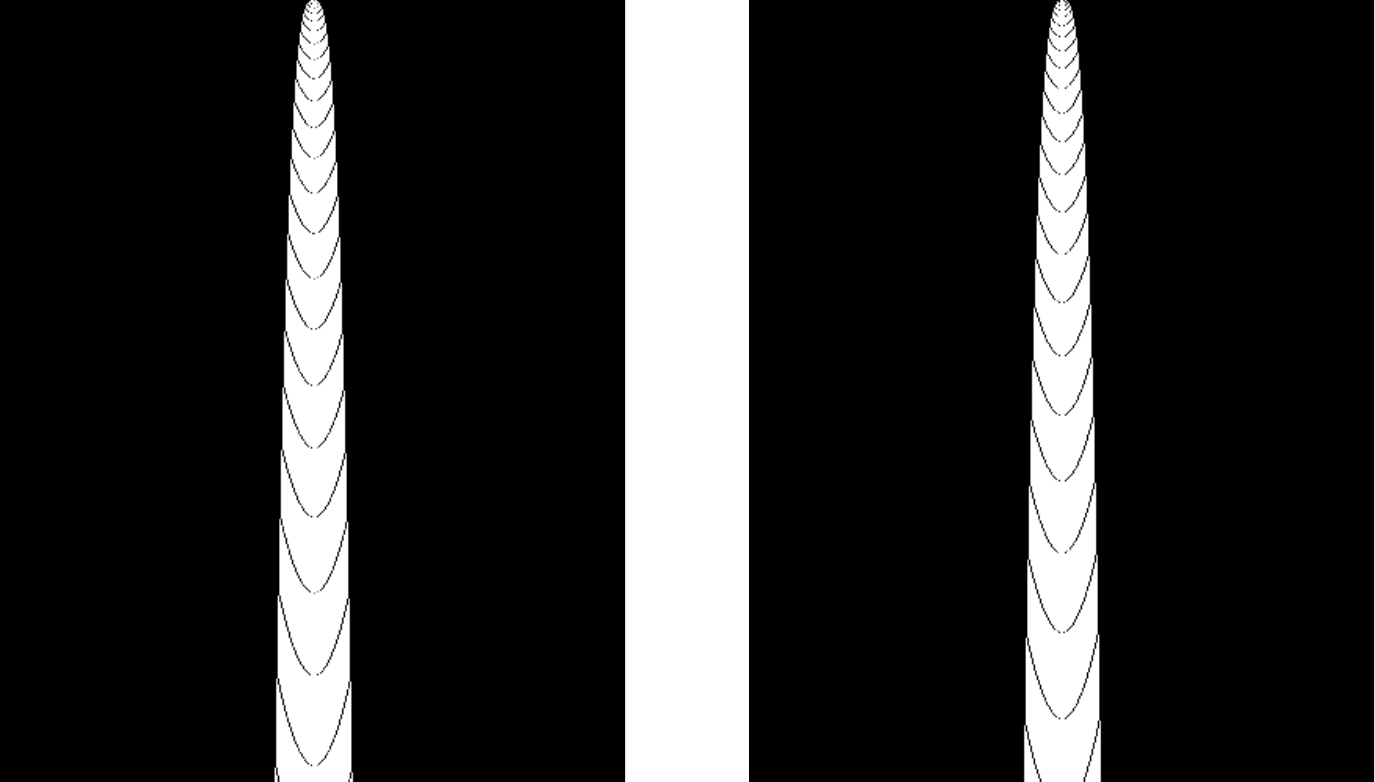}
    \caption{10,000-step space-time diagrams of the two 4-state machines decided by the Repeated Word List (RepWL) decider in the $S(4)$ pipeline, Table~\ref{tab:pipelineBB4}. Left: {\small \tm{1RB1LA_1LA0RC_1LD1RC_---0LA}}. Right: {\small \tm{1RB0RB_1LC1RB_---0LD_1RA1LD}}. The RepWL graphs of these machines have $3{,}130$ and $3{,}076$ nodes, respectively. Both machines are decided using $l=4$ and $T=3$.
    }\label{fig:RepWLBB4}
\end{figure}

\paragraph{Other implementations.} At the time of this writing, RepWL also has a Haskell and a Python implementation \cite{RepWL_haskell,RepWL_python}.

\newpage
\subsection{Finite Automata Reduction (FAR)}\label{sec:FAR}

\usetikzlibrary {automata, positioning}
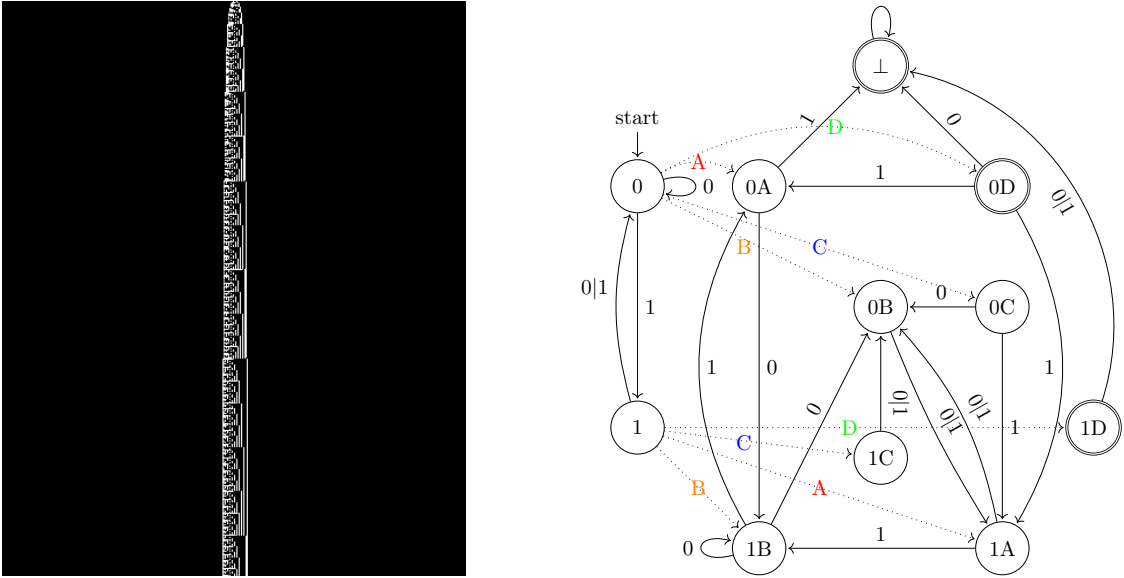
\begin{figure}[h]
    \centering
    \begin{subfigure}[t]{0.45\textwidth}
        \centering
        \includegraphics[width=0.85\textwidth]{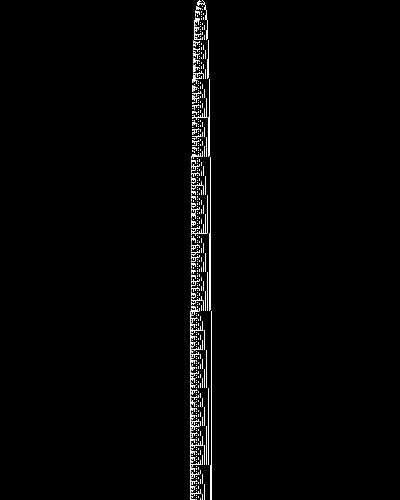}
        \caption*{}
    \end{subfigure}%
    \hfill
    \begin{subfigure}[t]{0.5\textwidth}
        \scalebox{0.8}{
            \begin{tikzpicture}[shorten >=1pt]
                \node[state,initial above] (0) at (-2, 2) {0};
                \node[state]           (1) at (-2, -2) {1};
                \node[state,accepting] (H) at (2, 4) {$\bot$};
                \node[state]           (0A) at (0, 2) {0A};
                \node[state,accepting] (0D) at (4, 2) {0D};
                \node[state]           (0B) at (2, 0) {0B};
                \node[state]           (0C) at (4, 0) {0C};
                \node[state]           (1C) at (2, -2.5) {1C};
                \node[state]           (1B) at (0, -4) {1B};
                \node[state]           (1A) at (4, -4) {1A};
                \node[state,accepting] (1D) at (5.5, -2) {1D};

                \path[->]  (0)  edge [loop right]       node {$0$} (0)
                edge                    node [right] {$1$} (1)
                (1)  edge [bend left=15]     node [above left] {$0|1$} (0)
                (0A) edge                    node [right] {$0$} (1B)
                edge                    node [above,rotate=60] {$1$} (H)
                (0C) edge                    node [above] {$0$} (0B)
                edge                    node [right] {$1$} (1A)
                (0D) edge                    node [above,rotate=300] {$0$} (H)
                edge                    node [above] {$1$} (0A)
                edge [bend left]        node [left] {$1$} (1A)
                (1A) edge                    node [above] {$1$} (1B)
                (1A) edge [bend right=15]    node [above,rotate=300] {$0|1$} (0B)
                (0B) edge [bend right=5]     node [above,rotate=300] {$0|1$} (1A)
                (1B) edge                    node [above right,rotate=60] {$0$} (0B)
                edge [loop left]        node {$0$} (1B)
                edge [bend left]        node [right] {$1$} (0A)
                (1C) edge                    node [above right,rotate=270] {$0|1$} (0B)
                (1D) edge [bend right=48]    node [below,rotate=300] {$0|1$} (H)
                (H)  edge [loop above]       node {$0|1$} (H)
                (0)  edge [dotted,bend left] node {\stateA} (0A)
                edge [dotted]           node [left] {\textcolor{colorB}{B}} (0B)
                edge [dotted]           node {\textcolor{colorC}{C}} (0C)
                edge [dotted,bend left] node [right] {\textcolor{colorD}{D}} (0D)
                (1)  edge [dotted]           node {\textcolor{colorB}{B}} (1B)
                edge [dotted]           node {\textcolor{colorA}{A}} (1A)
                edge [dotted]           node [left] {\textcolor{colorC}{C}} (1C)
                edge [dotted]           node [left] {\textcolor{colorD}{D}} (1D);
            \end{tikzpicture}
        }
        \label{fig:far_nfa}
    \end{subfigure}
    \caption[Short caption]{{\small Left: 20,000-step space-time diagram of 4-state machine \tm{1RB0LD_1LC1RA_0RB0LC_---1LA} -- we use a 4-state machine to have a small FAR Nondeterministic Finite Automaton (NFA). Right: NFA that satisfies the FAR conditions (Theorem~\ref{far-main-theorem}, with accepted steady state-set $\{\bot\}$) for this machine and hence is a certificate that the machine does not halt. This NFA accepts at least all eventually-halting configurations\footnotemark ~of the machine (configurations are represented as words, see Section~\ref{sec:FAR:theorem}); because it rejects the initial all-0 configuration (e.g. word-encoded as \texttt{A0}, or just \texttt{A}, not leading to an accept state), we know the machine does not halt.}}
    \label{fig:finite-automata-reduction}
\end{figure}

\footnotetext[\value{footnote}]{\label{fn:halting-caption}With finitely many \texttt{1}s, see Section~\ref{sec:FAR:theorem}.}
\vspace{-1em}
\subsubsection{Overview}

\newcommand{\M}{M}
\newcommand{\T}{^{T}}

Finite Automata Reduction (FAR) is a \textit{co-CTL} technique, \ie it is dual to the Closed Tape Language (CTL) framework given in Section~\ref{sec:deciders-overview}: for a given Turing machine, we are looking for a regular language that contains the set of the machine's eventually-halting configurations and, provided that the all-0 configuration is not in the regular language, we know that the machine does not halt.

The specificity of FAR is to restrict regular languages to a class of Nondeterministic Finite Automata (NFA) -- those satisfying Theorem~\ref{far-main-theorem} -- for which it is computationally simple to verify that they have the co-CTL properties: (i) reject the all-0 initial configuration, (ii) closed under Turing machine transitions, (iii) accept all eventually-halting configurations.

In \CoqBB, FAR is only used as a \textbf{verifier} meaning that specific Turing machines together with their FAR NFAs are directly hardcoded in the proof (in file \texttt{Verifier\_FAR\_Hardcoded\_Certificates.v}) and then verified using Theorem~\ref{far-main-theorem} -- see Section~\ref{sec:FAR:results} for results. FAR was originally developed as a fully-fledged decider -- \ie the verifier together with search algorithms for NFAs \cite{FAR,bbchallenge_part1}.

Here, we only present the verifier part of FAR (Theorem~\ref{far-main-theorem}) while we present the decider and its variations in \cite{bbchallenge_part1}. Figure~\ref{fig:finite-automata-reduction} (right) gives a FAR NFA (\ie satisfying\footnote{Using accepted steady state-set $\{\bot\}$, see Section~\ref{sec:FAR:theorem}.} Theorem~\ref{far-main-theorem}) for machine \tm{1RB0LD_1LC1RA_0RB0LC_---1LA}: the NFA accepts at least all the eventually-halting configurations\footref{fn:halting-caption} of the machine and rejects the initial all-0 configuration (\ie \texttt{A0} does not lead to an accept state), giving a certificate that the machine does not halt.

\subsubsection{FAR theorem}\label{sec:FAR:theorem}

In the following, we limit ourselves to Turing machines configurations with finite support, \ie configurations with finitely many \texttt{1}s (or, more generally, finitely many non-\texttt{0} symbols) and, when we write \textit{configuration}, we mean, \textit{configuration with finite support}.

A Turing machine configuration $c$ is represented as a finite word, called a \textit{word-representation} of $c$, by concatenating the tape content (from left to right, making sure to include all the \texttt{1}s) and adding the state (in our case, a letter from A to E) just before the position of the head, which is the same directional head notation used in Section~\ref{sec:RepWL}. For instance, two word-representations of the configuration $\texttt{0}^\infty \; \texttt{A> 0011} \; \texttt{0}^\infty$, are $\hat{c} = \texttt{A0011}$ and $\hat{c}' = \texttt{000A00110000}$. Similarly, the initial all-0 configuration can be encoded as \texttt{A0} or even just \texttt{A}. Word-representations of the same configuration will only differ in the number of leading and trailing \texttt{0}s that they have.

Then, a co-CTL regular language of word-represented configurations $\mathcal{L}$ for a Turing machine $\mathcal{M}$ satisfies:
\begin{align}
    u \in \mathcal{L}                               & \iff 0u \in \mathcal{L}           &  & \text{(leading zeros ignored)}
    \label{eq:lzignore}
    \\
    u \in \mathcal{L}                               & \iff u0 \in \mathcal{L}           &  & \text{(trailing zeros ignored)}
    \label{eq:tzignore}
    \\
    c\TMstep\bot                                    & \implies \hat{c} \in \mathcal{L}  &  & \text{(recognising halt, base case)} \nonumber \\
    (c_1\TMstep c_2)\land \hat{c}_2 \in \mathcal{L} & \implies\hat{c}_1 \in \mathcal{L} &  & \text{(recognising halt, induction)} \nonumber
\end{align}

With $c, c_1, c_2$ configurations of $\mathcal{M}$ (with finite support) and $\hat{c}, \hat{c}_1, \hat{c}_2$ any of their word-representations.

Given how word-representations are defined, the last two above conditions become:
\begin{align}
    \forall u,z\in\balphabet^*\!: \; ufrz \in \mathcal{L},\;                                                             & \text{if $\delta(f,r)$ is undefined (\ie halting)} \label{eq:h0}
    \\
    \forall u,z\in\balphabet^*,\,\forall b \in \balphabet\!: utbwz \in \mathcal{L} \implies ubfrz \in \mathcal{L},\;     & \text{if $\delta(f,r) = (w,\text{L},t)$} \label{eq:hnl}
    \\
    \forall u,z\in\balphabet^*,\,\forall b \in \balphabet\!: u w t z \in \mathcal{L} \implies u f r z \in \mathcal{L},\; & \text{if $\delta(f,r) = (w,\text{R},t)$} \label{eq:hnr}
\end{align}

With $f,t \in \{\stateA,\stateB,\stateC,\stateD,\stateE\}$ the ``from'' and ``to'' states in a transition, $r,w,b \in \balphabet$ the bit ``read'', the bit ``written'', and just a bit, and $\delta$ the transition table (see Section~\ref{sec:TMs}) of $\mathcal{M}$.

We now transform Conditions~\eqref{eq:lzignore}–\eqref{eq:hnr} into, sometimes stronger, conditions on the structure of NFAs -- using the usual linear-algebraic description of NFAs, which we first recall. Let $\mathbf{2}$ denote the Boolean semiring $\{0,1\}$ with operations $+$ and $\cdot$ respectively implemented by $\operatorname{OR}$ and $\operatorname{AND}$ \cite{CUNINGHAMEGREEN1991251}. Let $\M_{m,n}$ be the set of matrices with $m$ rows and $n$ columns over $\mathbf{2}$. We may define a Nondeterministic Finite Automaton (NFA) with $n$ states and alphabet $\mathcal{A}$ as a tuple $(q_0, \{T_\gamma\}_{\gamma \in \mathcal{A}}, a)$ where $q_0 \in \M_{1,n}$ and $a \in \M_{1,n}$ respectively represent the initial states and accepting states of the NFA. (i.e. if the $i^\text{th}$ state of the NFA is an initial state then the $i^\text{th}$ entry of $q_0$ is set to 1 and the rest are 0, and the $i^\text{th}$ entry of $a$ is set to 1 if and only if the $i^\text{th}$ state of the NFA is accepting), and where transitions are matrices $T_\gamma\in \M_{n,n}$ for each $\gamma\in\mathcal{A}$ (i.e. the entry $(i,j)$ of matrix $T_\gamma$ is set to 1 iff the NFA transitions from state $i$ to state $j$ when reading $\gamma$). Furthermore, for any word $u=\gamma_1\dots\gamma_\ell \in \mathcal{A}^*$, let $T_u = T_{\gamma_1} T_{\gamma_2} \dots T_{\gamma_\ell}$ be the state transformation resulting from reading word $u$ (note, $T_\epsilon = I$). A word $u=\gamma_1\dots\gamma_\ell \in \mathcal{A}^*$ is accepted by the NFA iff there exists a path from an initial state to an accepting state that is labelled by the symbols of $u$, which algebraically translates to $q_0 T_u a\T = 1$ with $a\T \in \M_{n,1}$ the transposition of $a$.

Using this algebraic framework\footnote{In the following, we limit ourselves to the binary tape alphabet $\balphabet$, but the results generalise transparently to arbitrary alphabets $\alphabet$.}, Conditions~\eqref{eq:lzignore}~and~\eqref{eq:tzignore} are implied by the following stronger conditions on transition matrix $T_0 \in \M_{n,n}$:
\begin{align}
    q_0 T_0 & = q_0
    \label{far-cond-leading-0}
    \\
    T_0 a\T & = a\T
    \label{far-cond-trailing-0}
\end{align}

Indeed, Condition~\eqref{far-cond-leading-0} transparently ignores leading zeros, Condition~\eqref{far-cond-trailing-0} means that for all accepting states of the NFA, reading a \szero is possible and leads to an accepting state since $T_0 a\T$ describes the set of NFA states that reach the set of accepting states $a$ after reading a \szero.

Then, Conditions~\eqref{eq:h0}--\eqref{eq:hnr} algebraically translate to:
{\small
\begin{align*}
    \forall u,z\in\balphabet^*\!: \; q_0 T_u T_f T_r T_z a\T = 1, \;                                                                             & \text{if $\delta(f,r)$ is undefined (\ie halting)}
    \\
    \forall u,z\in\balphabet^*,\,\forall b \in \balphabet\!: q_0 T_{u} T_t T_b T_w T_{z} a\T = 1 \implies q_0 T_{u} T_b T_f T_r T_{z} a\T = 1,\; & \text{if $\delta(f,r) = (w,\text{L},t)$}
    \\
    \forall u,z\in\balphabet^*,\,\forall b \in \balphabet\!: q_0 T_{u} T_w T_t T_{z} a\T = 1 \implies q_0 T_{u} T_f T_r T_{z} a\T = 1,\;         & \text{if $\delta(f,r) = (w,\text{R},t)$}
\end{align*}
}

These conditions are unwieldy. We seek stronger (thus still sufficient) conditions which are simpler:

\begin{itemize}

    \item For machine transitions going left or right, simply require $T_t T_b T_w\preceq T_b T_f T_r$ and $T_w T_t\preceq T_f T_r$, respectively, where $\preceq$ is the following relation on same-size matrices: $M\preceq M'$ if $M_{ij}\le M'_{ij}$ element-wise, that is, if the second matrix has at least the same 1-entries as the first matrix.

    \item To simplify the condition for halting machine transitions: define an \emph{accepted steady state-set} $s$ to be a row vector such that $sa\T = 1$, $s T_0\succeq s$, and $s T_1\succeq s$. Given such $s$, we have that: $\forall q\in\M_{1,n}\; q \succeq s\implies \forall z\in\balphabet^*\!: qT_{z}a\T = 1$. Assuming that such $s$ exists we can simply require: $\forall u\in\balphabet^*\!: q_0T_u T_f T_r \succeq s$ which is stronger than $\forall u,z\in\balphabet^*\!: \; q_0 T_u T_f T_r T_z a\T = 1$ where $\delta(f,r)$ is undefined.

\end{itemize}

Combining the above, we get FAR:

\begin{theorem}[\CoqBB: \texttt{Lemma dfa\_nfa\_verifier\_spec}\footnote{\CoqBB's lemma is slightly different, as it builds the NFA satisfying this theorem using a given ``precursor'' Deterministic Finite Automaton (DFA) -- as initially developed in \cite{FAR} -- see Section~\ref{sec:FAR:results}.}]
    \label{far-main-theorem}
    Machine $\mathcal{M}$, with transition table $\delta$ (see Section~\ref{sec:TMs}), does not halt from the initial all-0 configuration if there is a Nondeterministic Finite Automaton $(q_0, \{T_\gamma\}, a)$ and row vector $s$ satisfying the below:
    \begin{align}
        \label{far-cond-first}
        q_0 T_0                                    & = q_0
                                                   &                     & \text{(leading zeros ignored)}
        \tag{\ref{far-cond-leading-0}}
        \\
        T_0a\T                                     & = a\T
                                                   &                     & \text{(trailing zeros ignored)}
        \tag{\ref{far-cond-trailing-0}}
        \\
        sa\T                                       & = 1
                                                   &                     & \text{($s$ is accepted)}
        \label{far-cond-ass-accepted}
        \\
        sT_0,sT_1                                  & \succeq s
                                                   &                     & \text{($s$ is a steady state)}
        \label{far-cond-ass-steady}
        \\
        \forall u\in\balphabet^*\!: q_0T_u T_f T_r & \succeq s
                                                   &                     & \text{if $\delta(f,r)$ is undefined (\ie halting)}
        \label{far-cond-halt}
        \\
        \forall b\in\balphabet\!: T_b T_f T_r      & \succeq T_t T_b T_w
                                                   &                     & \text{if $\delta(f,r) = (w,\text{L},t)$}
        \label{far-cond-left}
        \\
        T_f T_r                                    & \succeq T_w T_t
                                                   &                     & \text{if $\delta(f,r) = (w,\text{R},t)$}
        \label{far-cond-last}
        \\
        q_0 T_A a\T                                & = 0
                                                   &                     & \text{(initial configuration rejected)}
        \label{far-cond-reject-start}
    \end{align}
\end{theorem}
\begin{proof}
    Conditions \eqref{far-cond-leading-0}--\eqref{far-cond-last} ensure that the NFA's language includes at least all eventually halting configurations of $\mathcal{M}$, see above. Condition~\eqref{far-cond-reject-start} ensures that the initial all-0 configuration of the machine is rejected, hence not eventually halting. Hence, if conditions \eqref{far-cond-leading-0}--\eqref{far-cond-reject-start} are satisfied, we can conclude that $\mathcal{M}$ does not halt from the initial all-0 configuration.
\end{proof}

\paragraph{Verifier.} Theorem~\ref{far-main-theorem} has the nice property of being easy to verify: given a Turing machine, an NFA and a vector $s$, the task of verifying that equations \eqref{far-cond-first}--\eqref{far-cond-reject-start} hold and thus that the machine does not halt, is computationally simple.\footnote{Note that although equation~\eqref{far-cond-halt} has a $\forall u\in\balphabet^*$ quantification, the set of NFA states reachable after reading an arbitrary $u \in \balphabet^*$ is computable, and we just have to consider one instance of equation~\eqref{far-cond-halt} replacing $q_0 T_u$ per such state.} For instance, it is easy to check that the NFA given in Figure~\ref{fig:finite-automata-reduction} satisfies Theorem~\ref{far-main-theorem}, using accepted steady state-set $\{\bot\}$, for machine \tm{1RB0LD_1LC1RA_0RB0LC_---1LA} and hence, the NFA provides a certificate that the machine does not halt.

\subsubsection{Implementations and results}\label{sec:FAR:results}

\CoqBB implements Theorem~\ref{far-main-theorem} in the special case where the FAR NFA is computed from a \textit{precursor} Deterministic Finite State Automaton, as described in \cite{bbchallenge_part1} (``direct FAR algorithm''). Certificates, consisting of such DFAs are hardcoded in the proof for 23 machines (in file \texttt{Verifier\_FAR\_Hardcoded\_Certificates.v}) and then verified using \texttt{dfa\_nfa\_verifier} (see file \texttt{Verifier\_FAR.v}).

These certificates were either found using extensive compute (\eg several weeks of searching DFAs essentially by brute force) or translated from other, undocumented, regular co-CTL methods (see Section~\ref{sec:deciders-overview}); indeed, in \cite{bbchallenge_part1} we show that FAR is an \textit{universal} regular co-CTL method: any regular co-CTL proof can be shoehorned into the framework of Theorem~\ref{far-main-theorem}.

\paragraph{Other implementations.} FAR has several other implementations: in Rust, C++ and Python~\cite{FAR,FAR_Tony,FAR_cosmo}.

\newpage
\subsection{Weighted FAR (WFAR)}\label{sec:WFAR}

\usetikzlibrary{automata, positioning, arrows.meta}
\begin{figure}[h!]
    \centering
    \begin{minipage}[t]{0.23\textwidth}
        \raggedright
        (a) Turing machine \\
        \centering
        \vspace{0.6em}
        \begin{tabular}{ccc}
            \toprule
                    & \textbf{0} & \textbf{1} \\
            \midrule
            \stateA & 1R\stateB  & ---        \\
            \stateB & 0R\stateC  & 1L\stateC  \\
            \stateC & 1R\stateD  & 1R\stateC  \\
            \stateD & 1L\stateE  & 1L\stateD  \\
            \stateE & 0R\stateA  & 0L\stateE  \\
            \bottomrule
        \end{tabular}

        \vspace{0.8em}
        \raggedright
        (a') Space-time diagram \\
        \vspace{0.3em}
        \centering
        \includegraphics[width=1.1\linewidth]{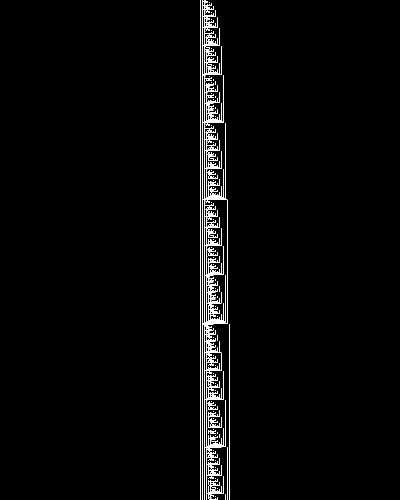}

    \end{minipage}
    \hfill
    \vrule
    \hfill
    \begin{minipage}[t]{0.71\textwidth}

        \begin{minipage}[t]{1\linewidth}

            \begin{minipage}[t]{0.49\textwidth}
                \raggedright
                (b) Left Weighted Automaton \\
                \vspace{0.5em}
                \centering
                \begin{tikzpicture}[->, >=Stealth, auto, node distance=1.8cm, every node/.style={scale=1}]
                    \tikzset{
                        state/.style={
                                circle, draw, minimum size=0.9cm, inner sep=1pt
                            }
                    }

                    \node[state] (0) {$p_0$};
                    \node[state] (2) [right of=0] {$p_2$};
                    \node[state] (3) [right of=2] {$p_3$};
                    \node[state] (1) [above of=3] {$p_1$};

                    \draw[->] (-1,0) -- (0);

                    \draw[->, black] (0) edge[loop above] node{\szero} (0);
                    \draw[->, black] (0) edge node{\sone} (2);

                    \draw[->, black] (2) edge[loop above] node{\sone} (2);
                    \draw[->, black, bend left=20] (3) to node{\sone} (2);
                    \draw[->, green!70!black, thick, bend left=20] (2) to node{\szero} (3);

                    \draw[->, black] (3) edge node{\szero} (1);
                    \draw[->, black] (1) edge[loop left] node{\szero, \sone} (1);

                \end{tikzpicture}
            \end{minipage}
            \hfill
            \begin{minipage}[t]{0.49\textwidth}
                \raggedright
                (c) Right Weighted Automaton \\
                \vspace{0.3em}
                \centering
                \begin{tikzpicture}[->, >=Stealth, auto, node distance=2cm, every node/.style={scale=1}]
                    \tikzset{
                        state/.style={
                                circle, draw, minimum size=1cm, inner sep=1pt
                            }
                    }

                    \node[state] (0) at (0,0) {$q_0$};
                    \node[state] (2) at (3,0) {
                        $q_1$};

                    \draw[->] (0) edge[loop above] node{\szero} (0);
                    \draw[->] (2) edge[loop above] node{\sone} (2);

                    \draw[->] (0) edge[bend left=10] node{\sone} (2);
                    \draw[->, red, bend left=20, thick] (2) to node[below]{\szero} (0);

                    \draw[->] (-1,0) -- (0);

                    \node[draw, below=0.5cm of 2, inner sep=3pt, rounded corners] (legend) {
                        \scalebox{0.7}{
                            \begin{tabular}{@{}rl@{}}
                                \tikz[baseline=-0.5ex]\draw[->, black, thick] (0,0) -- +(0.6,0);          & \;Weight 0    \\
                                \tikz[baseline=-0.5ex]\draw[->, green!70!black, thick] (0,0) -- +(0.6,0); & \;Weight 1    \\
                                \tikz[baseline=-0.5ex]\draw[->, red, thick] (0,0) -- +(0.6,0);            & \;Weight $-1$
                            \end{tabular}
                        }
                    };
                \end{tikzpicture}
            \end{minipage}

        \end{minipage}

        \vspace{1em}

        \raggedright
        (d) Example: configuration is accepted, hence nonhalting \\
        \vspace{0.3em}
        \centering

        \newcommand{\underarrowleft}[1]{%
            \tikz[baseline=(X.base)]{
                \node (X) {$#1$};
                \draw[->, thick] ([yshift=-1.3ex]X.east) -- ([yshift=-1.3ex]X.west);
            }
        }

        \newcommand{\underarrowright}[1]{%
            \tikz[baseline=(X.base)]{
                \node (X) {$#1$};
                \draw[->, thick] ([yshift=-1.3ex]X.west) -- ([yshift=-1.3ex]X.east);
            }
        }

        \scalebox{0.8}{

            \begin{minipage}{\textwidth}
                \centering
                {\LARGE
                $\underarrowright{\texttt{10101}}\; \texttt{\stateC}\sone\; \underarrowleft{\texttt{01}}$ \\
                \vspace{-0.3em}
                {\small $\quad\quad$Left WA reads$\quad\quad\quad\quad\quad\quad\quad\quad$Right WA reads} \\
                $\quad  \; \;$\fcolorbox{red}{yellow!30}{$[p_2] \; \texttt{\stateC}\sone\; [q_0]$} \\
                \vspace{-2em}
                \[
                    \begin{tikzpicture}[baseline=(base)]
                        \node (base) {$\quad \quad W = W_l + W_r = 1$};

                        \coordinate (Wl) at ([xshift=5.33em]base.base west);
                        \coordinate (Wr) at ([xshift=7.85em]base.base west);

                        \node[below=1.8ex of Wl] (lval) {\normalsize 2};
                        \node[below=1.8ex of Wr] (rval) {\normalsize $-1$};

                        \draw[->, thick] ([yshift=-0.5ex]Wl) -- (lval.north);
                        \draw[->, thick] ([yshift=-0.5ex]Wr) -- (rval.north);
                    \end{tikzpicture}
                \] \\
                \vspace{-0.5em}
                {\large Configuration accepted, see (e), hence machine does not halt starting from $\texttt{10101}\; \texttt{\stateC}\sone\; \texttt{01}$, see Theorem~\ref{th:WFAR}. }
                }

            \end{minipage}
        }

        \vspace{0.5em}
        \raggedright
        (e) Accepted weighted configurations \\
        \centering

        \usetikzlibrary{positioning, shapes.multipart, fit}
        \scalebox{0.8}{
            \begin{tikzpicture}[node distance=0.2cm and 0.3cm, every node/.style={font=\normalsize}, anchor=north]

                \coordinate (topref) at (0,0);

                \node[draw, rectangle, rounded corners, inner sep=3pt] (wmin1box)
                {\begin{tabular}{l}
                        $[p_2]\; \texttt{\stateE}\szero\; [q_0]$
                    \end{tabular}};
                \node[above=0cm of wmin1box] {$W \geq$ -1};

                \node[draw, rectangle, rounded corners, inner sep=3pt, right=of wmin1box] (w0box)
                {\begin{tabular}{l}
                        $[p_0]\; \texttt{\stateA}\szero\; [q_0]$ \\
                        $[p_0]\; \texttt{\stateE}\szero\; [q_0]$ \\
                        $[p_0]\; \texttt{\stateE}\sone\; [q_0]$
                    \end{tabular}};
                \node[above=0cm of w0box] {$W = 0$};

                \node[draw, rectangle, rounded corners, inner sep=3pt, right=of w0box] (w0plusbox)
                {\begin{tabular}{l}
                        $[p_2]\; \texttt{\stateB}\szero\; [q_0]$ \\
                        $[p_2]\; \texttt{\stateE}\sone\; [q_0]$  \\
                        $[p_3]\; \texttt{\stateE}\sone\; [q_0]$
                    \end{tabular}};
                \node[above=0cm of w0plusbox] {$W\geq0$};
                \node[draw, rectangle, rounded corners, inner sep=3pt, below=0.6cm of w0box, xshift=1.4cm] (w1box)
                {\begin{tabular}{ll}
                        $[p_3]\; \texttt{\stateA}\szero\; [q_1]$                                             & $[p_2]\; \texttt{\stateC}\sone\; [q_1]$  \\
                        $[p_2]\; \texttt{\stateB}\szero\; [q_1]$                                             & $[p_3]\; \texttt{\stateC}\sone\; [q_1]$  \\
                        $[p_2]\; \texttt{\stateB}\sone\; [q_0]$                                              & $[p_2]\; \texttt{\stateD}\szero\; [q_0]$ \\
                        $[p_2]\; \texttt{\stateB}\sone\; [q_1]$                                              & $[p_2]\; \texttt{\stateD}\szero\; [q_1]$ \\
                        $[p_2]\; \texttt{\stateC}\szero\; [q_0]$                                             & $[p_2]\; \texttt{\stateE}\sone\; [q_1]$  \\
                        $[p_3]\; \texttt{\stateC}\szero\; [q_0]$                                             & $[p_3]\; \texttt{\stateE}\sone\; [q_1]$  \\
                        \rule{0pt}{3.3ex}\fcolorbox{red}{yellow!30}{$[p_2]\; \texttt{\stateC}\sone\; [q_0]$} &                                          \\
                    \end{tabular}};
                \node[above=0cm of w1box] {$W \geq 1$};

                \node[draw, rectangle, rounded corners, inner sep=6pt, below=0.7cm of wmin1box] (w2box)
                {\begin{tabular}{l}
                        $[p_3]\; \texttt{\stateA}\szero\; [q_0]$ \\
                        $[p_2]\; \texttt{\stateC}\szero\; [q_1]$ \\
                        $[p_3]\; \texttt{\stateC}\szero\; [q_1]$ \\
                        $[p_3]\; \texttt{\stateC}\sone\; [q_0]$  \\
                        $[p_2]\; \texttt{\stateD}\sone\; [q_0]$  \\
                        $[p_2]\; \texttt{\stateD}\sone\; [q_1]$  \\
                        $[p_3]\; \texttt{\stateD}\sone\; [q_1]$  \\
                    \end{tabular}};
                \node[above=0cm of w2box] {$W \geq 2$};

            \end{tikzpicture}
        }
    \end{minipage}

    \caption{{\small WFAR certificate of nonhalting for machine \tm{1RB---_0RC1LC_1RD1RC_1LE1LD_0RA0LE}: (a) transition table and 20,000-step space-time diagram, (b) left weighted automaton: processes symbols to the left of the head in the left-to-right direction, which results in a left end-state -- \eg state $p_2$ when processing \texttt{10101} -- and a left weight obtained by summing the weights of each transition -- \eg $W_l = 2$ when processing \texttt{10101} (c) right weighted automaton: processes symbols to the right of the head (excluding the symbol read by the head) in the right-to-left direction, indicated with arrow, which results in a right end-state -- \eg state $q_0$ when processing \texttt{01} right-to-left -- and a right weight -- \eg $W_r = -1$ when processing \texttt{01} right-to-left (d) example, the total weight of configuration $\texttt{10101}\; \texttt{\stateC}\sone\; \texttt{01}$ is $W=W_l+W_r = 1$, using same word-encoding of configurations as in Section~\ref{sec:FAR}, and the right and left end-states are $p_2$ and $q_0$. Weighted automaton configuration $[p_2]\; \texttt{\stateC}\texttt{1}\; [q_0]$ with $W = 1$ is in the set of accepted weighted configurations (under more general bound $W \geq 1$), see (e). Therefore we know that the machine does not halt from configuration $\texttt{10101}\; \texttt{\stateC}\sone\; \texttt{01}$, Theorem~\ref{th:WFAR}. Similarly, Turing machine configuration $\texttt{\stateA}\texttt{0}$, which results in weighted configuration $[p_0]\;\texttt{\stateA}\texttt{0}\; [q_0]$ with $W=0$ is accepted, ensuring that the machine does not halt from the all-0 initial tape, Theorem~\ref{th:WFAR}.}}\label{fig:WFAR}
\end{figure}

\subsubsection{Overview}\label{sec:WFAR:overview}

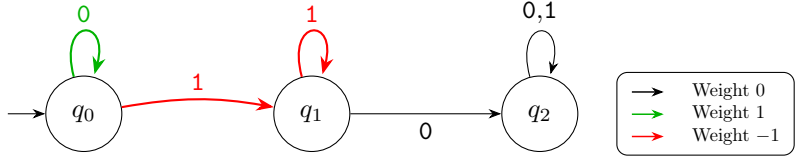
\begin{figure}
    \centering
    \begin{tikzpicture}[->, >=Stealth, auto, node distance=2cm, every node/.style={scale=1}]
        \tikzset{
            state/.style={
                    circle, draw, minimum size=1cm, inner sep=1pt
                }
        }

        \node[state] (0) at (0,0) {$q_0$};
        \node[state] (2) at (3,0) {$q_1$};
        \node[state] (3) at (6,0) {$q_2$};

        \draw[->, green!70!black, thick] (0) edge[loop above] node{\szero} (0);
        \draw[->, red, thick] (2) edge[loop above] node{\sone} (2);
        \draw[->] (3) edge[loop above] node{\szero,\sone} (3);

        \draw[->, red, thick] (0) edge[bend left=10] node{\sone} (2);
        \draw[->] (2) to node[below]{\szero} (3);

        \draw[->] (-1,0) -- (0);

        \node[draw, right=0.5cm of 3, inner sep=3pt, rounded corners] (legend) {
            \scalebox{0.7}{
                \begin{tabular}{@{}rl@{}}
                    \tikz[baseline=-0.5ex]\draw[->, black, thick] (0,0) -- +(0.6,0);          & \;Weight 0    \\
                    \tikz[baseline=-0.5ex]\draw[->, green!70!black, thick] (0,0) -- +(0.6,0); & \;Weight 1    \\
                    \tikz[baseline=-0.5ex]\draw[->, red, thick] (0,0) -- +(0.6,0);            & \;Weight $-1$
                \end{tabular}
            }
        };
    \end{tikzpicture}
    \caption{Weighted automaton recognising nonregular language $\texttt{0}^n \texttt{1}^n$, using accept set $\{(q_1,W=0)\}$ or $\{(q_1,W=0), (q_0,W=0)\}$ if we include the empty word.}\label{fig:ex_wa}
\end{figure}

Weighted automata are an extension of usual finite state automata where each transition is given a weight in $\Z$: when a word is processed, total weight $W \in \Z$ is computed by summing the weights of all the encountered transitions. Accepted words are described by a set of pairs of final-state and weight lower and upper bounds (potentially infinite) to satisfy: for instance, the archetypal nonregular language $\texttt{0}^n \texttt{1}^n$ is recognised by the weighted automaton of Figure~\ref{fig:ex_wa} using accept set $\{(q_1,0 \leq W \leq 0)\}$ which we can simplify as $\{(q_1,W=0)\}$ and we may add $(q_0,W=0)$ to the set if we want to include the empty word.

Weighted Finite Automata Reduction (WFAR) is an extension of FAR (Section~\ref{sec:FAR}) using deterministic weighted finite automata. Figure~\ref{fig:WFAR} gives a \textit{WFAR automaton}, which is a certificate of nonhalting the machine given in Figure~\ref{fig:WFAR}~(a). A WFAR automaton consists of (i) a \textit{left deterministic weighted automaton} (ii) a \textit{right deterministic weighted automaton} and (iii) a set of accepted \textit{weighted configurations}, see Figure~\ref{fig:WFAR}~(b), (c), and (e).
A WFAR automaton processes word-representations (as defined in Section~\ref{sec:FAR}) of Turing machine configurations\footnote{With finitely many $\sone$s, which we always assume from now on.} in the way described below, and, if a configuration is accepted by the WFAR automaton, we know that the associated Turing machine does not halt from that configuration, Theorem~\ref{th:WFAR}. That way, WFAR is a CTL method (instead of co-CTL for FAR), see Section~\ref{sec:deciders-overview}. The WFAR automaton of Figure~\ref{fig:WFAR} accepts (see below for what it means) the initial configuration $\texttt{\stateA}\szero$, giving a certificate of nonhalting for the machine of Figure~\ref{fig:WFAR}~(a) from the all-0 tape.

The method was initially developed as a decider \cite{iijil1_2025_14914502} and integrated to \CoqBB as a verifier: similarly to FAR (Section~\ref{sec:FAR}), 17 WFAR certificates are directly hardcoded in the \Coq proof, see file \texttt{Verifier\_WFAR\_Hardcoded\_Certificates.v}, see Section~\ref{sec:WFAR:results} for results.

\paragraph{WFAR processing.} Let us describe how a WFAR automaton processes a word-represented Turing machine configuration in order to decide whether it is accepted or not, as illustrated in Figure~\ref{fig:WFAR}. WFAR is an extension of the ``Meet-in-the-middle''\footnote{See Section 6.6 in \cite{bbchallenge_part1}.} instance of FAR \cite{bbchallenge_part1}: word-representations of configurations are split into three segments, (i) word to the left of the head, (ii) head state and symbol, (iii) word to the right of the head; \eg $\texttt{10101}\; \texttt{\stateC}\sone\; \texttt{01}$, Figure~\ref{fig:WFAR}~(d). The left word -- here $\texttt{10101}$ -- is processed left-to-right by the left weighted automaton, Figure~\ref{fig:WFAR}~(b), and the right word -- here $\texttt{01}$ -- is processed right-to-left, by the right weighted automaton, Figure~\ref{fig:WFAR}~(c). In this case, this results in final left state $p_2$, final right state $q_0$, final left weight $W_l = 2$ and final right weight $W_r = -1$; the final total weight is $W = W_l + W_r = 1$, Figure~\ref{fig:WFAR}~(d). We denote this final \textit{weighted configuration} as $[p_2]\; \texttt{\stateC}\sone\; [q_0]$ with $W=1$. This final weighted configuration belongs to the set of accepted weighted configurations, Figure~\ref{fig:WFAR}~(e), which means that configuration $\texttt{10101}\; \texttt{\stateC}\sone\; \texttt{01}$ is \textit{accepted} by this WFAR automaton.

\subsubsection{WFAR theorem}

WFAR is a CTL technique -- see Section~\ref{sec:deciders-overview}: a WFAR automaton for a given Turing machine is meant to recognise a language of configurations $\mathcal{L}$ that includes the initial all-0 configuration, closed under Turing machine steps and that does not contain any halting configuration. Hence, we get the following CTL formalism, plus leading/trailing zeros conditions similarly to FAR:
\begin{align}
    u \in \mathcal{L}                               & \iff 0u \in \mathcal{L}               &  & \text{(leading zeros ignored)}
    \tag{\ref{eq:lzignore}}
    \\
    u \in \mathcal{L}                               & \iff u0 \in \mathcal{L}               &  & \text{(trailing zeros ignored)}
    \tag{\ref{eq:tzignore}}
    \\
    c\TMstep\bot                                    & \implies \hat{c} \not \in \mathcal{L} &  & \text{(reject halt)} \label{eq:rejecthalt}     \\
    (c_1\TMstep c_2)\land \hat{c}_1 \in \mathcal{L} & \implies\hat{c}_2 \in \mathcal{L}     &  & \text{(forward closure)} \label{eq:forwardclo}
\end{align}

Let us now show how to verify that a given WFAR automaton for a given Turing machine $\mathcal{M}$ accepts such $\mathcal{L}$, hence providing a certificate that $\mathcal{M}$ does not halt from the all-0 tape.

In the following, $\delta_L: Q_L \times \balphabet \to Q_L$ and $\delta_R: Q_R \times \balphabet \to Q_R$ respectively refer to the transition functions of the deterministic left and right weighted automaton of a WFAR automaton, \eg Figure~\ref{fig:WFAR}~(b) and (c), with $Q_L = \{p_0, \, \dots, \, p_{n_L-1}\}$ and $Q_R = \{q_0, \, \dots, \, q_{n_R-1}\}$ their respective set of states with $n_L$ and $n_R$ the number of left/right states and $p_0$ and $q_0$ are the respective initial states of the left and right weighted automaton. Weights are given by $w_L: Q_L \times \balphabet \to \Z$ and $w_R: Q_R \times \balphabet \to \Z$. We write $\delta_\mathcal{M}:\states \times \balphabet \partialto \balphabet \times \{\text{L},\text{R}\} \times \states$ for the transition function of $\mathcal{M}$.\footnote{In the following, we limit ourselves to the binary tape alphabet $\balphabet$, but the results generalise transparently to arbitrary alphabets $\alphabet$.}

\paragraph{Leading/trailing zeros.} Checking Conditions~\eqref{eq:lzignore} and \eqref{eq:tzignore} for a WFAR automaton is simple: thanks to the left-to-right and right-to-left respective read directions for the left and right weighted automaton, we simply have to check that $\delta_L(p_0,\szero) = p_0$ and $\delta_R(q_0,\szero) = q_0$ as well as $w_L(p_0,\szero) = w_R(q_0,\szero) = 0$ to ensure the convention that the weight of all word-representations of the initial all-0 configuration is 0.

\paragraph{Forward closure, without weights: back to FAR.} First, let us reformulate forward closure for a WFAR automaton, ignoring weight computations. Forward closure concerns the WFAR automaton's accept state, let us consider an example first. The WFAR automaton of Figure~\ref{fig:WFAR} accepts the initial Turing machine configuration $\texttt{\stateA}\szero$: the WFAR configuration $[p_0] \; \texttt{\stateA}\szero\; [q_0]$ (ignoring $W=0$) is in the accept set given in Figure~\ref{fig:WFAR}~(e). To ensure forward closure, Condition~\ref{eq:forwardclo}, let us consider how $\delta_\mathcal{M}(\texttt{\stateA},\szero) = \texttt{1R\stateB}$ affects $[p_0] \; \texttt{\stateA}\szero\; [q_0]$; we get $[p_0] \; \texttt{1}\; \texttt{\stateB}\texttt{?}\; [\texttt{?}]$, which is $[p_2] \; \texttt{\stateB}\texttt{?}\; [\texttt{?}]$ given that $\delta_L(p_0,\sone) = p_2$, see Figure~\ref{fig:WFAR}~(b). In order to resolve $\texttt{?}$, we look at all the transitions in the right weighted automaton that lead to $q_0$, see Figure~\ref{fig:WFAR}~(c): there are two, both reading a \texttt{0}, giving $[p_2] \; \texttt{\stateB}\texttt{0}\; [q_0]$ and $[p_2] \; \texttt{\stateB}\texttt{0}\; [q_1]$. Ignoring weights, we want both in the accept set\footnote{Which is the case here, with $W\geq 0$ and $W \geq 1$ in Figure~\ref{fig:WFAR}~(e).}: that ensures that for any Turing machine configuration $c_1$ yielding WFAR configuration $[p_0] \; \texttt{\stateA}\szero\; [q_0]$, then $c_2$ is also accepted by the WFAR automaton with $c_1 \TMstep c_2$. Note that $c_1$ and $c_2$ are not necessarily reachable from the initial all-0 tape: CTL methods provide languages that overestimate the language generated by Turing machines from the all-0 tape.

In general, ignoring weights, forward closure means the following for WFAR automaton accept set $\mathfrak{A}$:
\begin{align}
     & \forall q', r' \in Q_R \times \balphabet \text{ s.t. } \delta_R(q', r') = q,
     &                                                                              & [p]\, fr\, [q] \in \mathfrak{A} \Rightarrow [\delta_L(p,b)]\, tr'\, [q'] \in \mathfrak{A}
     &                                                                              & \text{if } \delta_{\mathcal{M}}(f,r) = (b, \text{R}, t) \label{eq:forwardcloR}
    \\[0.5em]
     & \forall p', r' \in Q_L\times \balphabet \text{ s.t. } \delta_L(p', r') = p,
     &                                                                              & [p]\, fr\, [q] \in \mathfrak{A} \Rightarrow [p']\, tr'\, [\delta_R(q,b)] \in \mathfrak{A}
     &                                                                              & \text{if } \delta_{\mathcal{M}}(f,r) = (b, \text{L}, t) \label{eq:forwardcloL}
\end{align}

For all left/right weighted automata states $p,q \in Q_L \times Q_R$ and notations $f,t \in \{\stateA,\stateB,\stateC,\stateD,\stateE\}$ the ``from'' and ``to'' states in a transition, $r,b \in \balphabet$ respectively the bit read and the bit written in a transition. If, ignoring weights, a WFAR accept set is forward-closed in the above sense, contains no halting configuration, and contains $[p_0] \; \texttt{\stateA}\szero\; [q_0]$, then we are in a particular case of FAR, as shown in \cite{bbchallenge_part1} \ie Theorem~\ref{far-main-theorem} can be applied: the Turing machine does not halt from the initial all-0 configuration and is regular in the sense of Section~\ref{sec:deciders-overview}.

\paragraph{Forward closure, with weights: beyond FAR.} Weights allow to further restrict the accept set in cases where the above, \textit{weightless}, forward closure does include halting configurations. For instance, in the case of Figure~\ref{fig:WFAR}, we have $[p_2] \; \texttt{\stateD}\sone \; [q_1]$ (with $W \geq 2$) in the accept set $\mathfrak{A}$, given in Figure~\ref{fig:WFAR}~(e), and, computing weightless forward closure from this WFAR configuration, ignoring weights, yields, using \eqref{eq:forwardcloR} and \eqref{eq:forwardcloL}, $[p_0] \; \texttt{\stateD}\sone \; [q_1]$, then $[p_0] \; \texttt{\stateD}\szero \; [q_1]$, then $[p_0] \; \texttt{\stateE}\szero \; [q_1]$ and finally, $[p_0] \; \texttt{\stateA}\sone \; [q_0] $, which is a halting configuration, meaning that we cannot conclude that $\mathcal{M}$ does not halt from the initial all-0 configuration. However, looking at Figure~\ref{fig:WFAR}~(e) we see that $[p_0] \; \texttt{\stateD}\sone \; [q_1]$ is not in $\mathfrak{A}$ and hence none of the successors either. This \textit{refinement} of $\mathfrak{A}$ is due to discarding \textit{impossible} weighted configurations, which we explain now.

With weights, WFAR configurations are expressed as follows: $[p] \; fr \; [q];\; W \geq m; \; W \leq M;$ with $m \in \Z \cup\{-\infty\}$ and $M \in \Z \cup \{+\infty\}$ weight bounds. For instance, considering the accept state $\mathfrak{A}$ of Figure~\ref{fig:WFAR}~(e) with weights, we have that the initial weighted configuration, $c'_1 = [p_0] \; \texttt{\stateA}\szero \; [q_0];\; W \geq 0;\; W \leq 0;$ is in $\mathfrak{A}$. When computing closure, bounds are updated by the \textit{total weight change} incurred when processing weighted transitions: consider $c'_2 = [p_2] \; \texttt{\stateB}\szero \; [q_1];\; W \geq \texttt{?};\; W \leq \texttt{?}$ obtained from the initial weighted configuration by \eqref{eq:forwardcloR}; we have $W(c_2) = W(c_1) + w_L(p_0,\sone) - w_R(q_1, \szero)$ with $c_1 \TMstep c_2$ Turing machine configurations such that $c_1$ yields WFAR configuration $c'_1$ and $c_2$ yields $c'_2$. Hence, for $c'_2$ to be accepted, we must update its weight bounds by weight change $w_L(p_0,\sone) - w_R(q_1, \szero) = 0 - (- 1) = +1$, giving $c'_2 = [p_2] \; \texttt{\stateB}\szero \; [q_1];\; W \geq 1;\; W \leq 1$, which is implied by more general $c'_2 = [p_2] \; \texttt{\stateB}\szero \; [q_1];\; W \geq 1;\; W < +\infty$ in $\mathfrak{A}$ of Figure~\ref{fig:WFAR}~(e). In general, in the case of \eqref{eq:forwardcloR}, using same notations, weight bounds $m$ and $M$ are added to weight change $w_L(p,b) - w_R(q', r')$ and weight change $w_R(q,b) - w_L(p',r')$ in the case of \eqref{eq:forwardcloL}; infinite bounds remain the same under any weight change.

Coming back to $[p_2] \; \texttt{\stateD}\sone \; [q_1];\; W \geq 2;\; W \leq + \infty;$ which we have shown above to lead to a halting configuration, we can now compute the bounds of the weighted configuration we obtained using \eqref{eq:forwardcloL}: $[p_0] \; \texttt{\stateD}\sone \; [q_1]\; W \geq 2; \; W < +\infty;$ as there is no weight changes. However, note that any left word reaching $q_0$ has left weight $W_l = 0$ and any right word reaching $q_1$ has right weight $W_r \leq 0$, hence total weight $W \leq 0$, which is incompatible with the constraint $W \geq 2$; hence we can discard $[p_0] \; \texttt{\stateD}\sone \; [q_1]\; W \geq 2; \; W < +\infty$ from the accept set $\mathfrak{A}$ as it is an impossible weighted configuration. Doing this also discards from $\mathfrak{A}$ all the weighted configurations we computed from $[p_0] \; \texttt{\stateD}\sone \; [q_1]\; W \geq 2; \; W < +\infty$, including the halting one, $[p_0] \; \texttt{\stateA}\sone \; [q_0] $.

In this case, in order to conclude, we needed to know that, in the left weighted automaton of Figure~\ref{fig:WFAR}~(b), terminating in state $q_0$ implies $W_l = 0$. In general, the exact feasible weight bounds of any state in a weighted automaton can be computed using the Bellman-Ford algorithm\footnote{Both the original and the \CoqBB implementations do not need the Bellman-Ford algorithm as they use restricted weighted automata on which it is easy to check whether the feasible weights for each state are nonpositive or nonnegative \cite{iijil1_2025_14914502}.}: the Bellman-Ford algorithm is able to compute the minimum weighted path from initial state to any state, thereby giving the lower weight bound for each state; the algorithm is also able to detect negative cycles, leading to $-\infty$ lower bounds; similarly, upper bounds are computed using Bellman-Ford on the automaton with negated weights.

For instance, in the left weighted automaton of Figure~\ref{fig:WFAR}~(b), at state $p_2$, we have $0 \leq W_l < +\infty$. Hence we can use these feasible weight bounds to automatically discard impossible weighted configurations from $\mathfrak{A}$ and hopefully, end up with no halting configuration in $\mathfrak{A}$.

We say that $\mathfrak{A}$ is \textit{weighted forward closed} for machine $\mathcal{M}$ if for all weighted configurations $c$ in $\mathfrak{A}$ and for any weighted configuration $c'$ obtained by closure from $c$ using \eqref{eq:forwardcloL} or \eqref{eq:forwardcloR} together with the weight bounds update rules stated above, either (i) there is $c'' \in \mathfrak{A}$ with bounds $m''$ and $M''$ such that $m'' \leq m'$ and $M'' \geq M'$ with $m'$ and $M'$ the bounds of $c'$, or (ii) $c'$ is an impossible weighted configuration as defined above, \ie incompatible with the feasible weight bounds computed from the left and right weighted automata.

We finally get the WFAR theorem:

\begin{theorem}[\CoqBB: \texttt{Lemma MITM\_WDFA\_verifier\_spec}]\label{th:WFAR}
    Let $\mathcal{M}$ be a Turing machine and $\mathcal{W}$ be a WFAR automaton with accept set $\mathfrak{A}$ such that:
    \begin{enumerate}
        \item Leading and trailing zeros are ignored: $\delta_L(p_0,\szero) = p_0$ and $\delta_R(q_0,\szero) = q_0$ with $w_L(p_0,\szero) = w_R(q_0,\szero) = 0$ with $\delta_L$ and $\delta_R$ the transition functions of the left and right weighted automata of $\mathcal{W}$ and $w_L$ and $w_R$ their weight functions.\label{pt:wfar:top}
        \item The initial configuration is accepted: \ie $[p_0]\; \texttt{\stateA}\szero\; [q_0];\; W=0;$ is in $\mathfrak{A}$.\label{pt:wfar:ttop}
        \item $\mathfrak{A}$ is weighted forward closed for $\mathcal{M}$.\label{pt:wfar:tttop}
        \item $\mathfrak{A}$ contains no halting configurations.\label{pt:wfar:ttttop}
    \end{enumerate}

    Then $\mathcal{M}$ does not halt for any configuration accepted by $\mathcal{W}$, which includes the initial all-0 configuration.
\end{theorem}
\begin{proof}
    Point~\ref{pt:wfar:top} guarantees that all the word-representations (see Section~\ref{sec:FAR}) of the same Turing machine configuration result in the same weighted WFAR configuration when processed by $\mathcal{W}$ (see Section~\ref{sec:WFAR:overview}). Points \ref{pt:wfar:ttop}-\ref{pt:wfar:ttttop} are the WFAR reformulations of the CTL argument (Section~\ref{sec:deciders-overview}). Hence, using the CTL argument, any Turing machine configuration accepted by $\mathcal{W}$ is nonhalting, and, in particular, the initial all-0 configuration.
\end{proof}

\subsubsection{Implementations and results}\label{sec:WFAR:results}

\CoqBB implements Theorem~\ref{th:WFAR}, see file \texttt{Verifier\_WFAR.v}. Certificates consist of left and right weighted automaton: accept sets are constructed by the \Coq verifier which computes the closure from $[p_0]\; \texttt{\stateA} \szero\; [q_0]; W=0$ and uses an integer parameter $P$ given in the certificate such that a bound $W \geq P$ is replaced by $W \leq +\infty$. See the 17 certificates in \texttt{Verifier\_WFAR\_Hardcoded\_Certificates.v}.

These certificates were mainly found by the original WFAR decider implementation \cite{iijil1_2025_14914502} which searches the space of WFAs using brute force. Certificates for ``Helices'' (see Section~\ref{sec:WFAR}) were handcrafted and are significantly bigger than the other certificates: about 50 states in the left and right weighted automata each where other certificates have less than 10 in each.

\newpage
\section{5-state Sporadic Machines}\label{sec:sporadic}

\begin{figure}[h!]
    \centering

    \begin{minipage}{\textwidth}
        \centering
        \begin{subfigure}{0.3\textwidth}
            \centering
            \includegraphics[width=\linewidth]{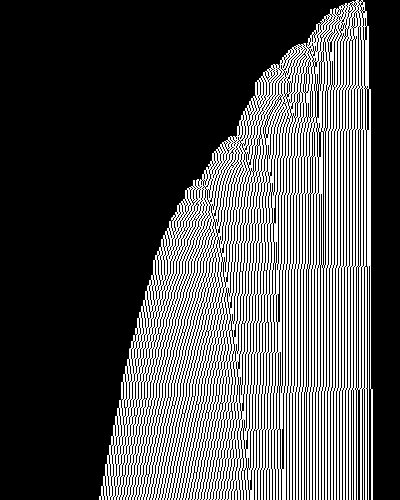}
            \caption*{\href{https://bbchallenge.org/1RB1RD_1LC0RC_1RA1LD_0RE0LB_---1RC}{Skelet \#1}}
        \end{subfigure}
        \hfill
        \raisebox{8.5em}[0pt][0pt]{%
            \begin{minipage}{0.3\textwidth}
                \centering
                \includegraphics[width=\linewidth]{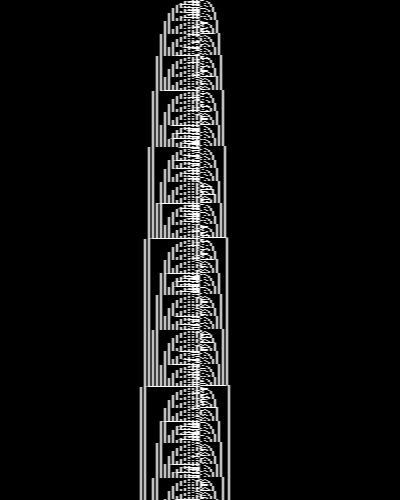}

                \caption*{\href{https://bbchallenge.org/1RB0RA_0LC1RA_1RE1LD_1LC0LD_---0RB}{Skelet \#10}}
                {\small\emph{Double Fibonacci Counter}}
            \end{minipage}
        }
        \hfill
        \begin{subfigure}{0.3\textwidth}
            \centering
            \includegraphics[width=\linewidth]{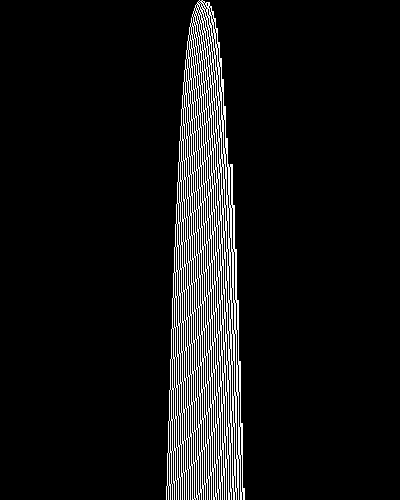}
            \caption*{\href{https://bbchallenge.org/1RB---_0LC1RE_0LD1LC_1RA1LB_0RB0RA}{Skelet \#17}}
        \end{subfigure}
    \end{minipage}

    \vspace{2.5em}

    \begin{tikzpicture}
        \node[draw=magenta, thick, rounded corners, inner sep=8pt] (box1) {
            \begin{minipage}{0.95\textwidth}
                \centering
                \textbf{\textcolor{magenta}{Shift Overflow Counters}}\\[0.8em]
                \begin{subfigure}{0.17\textwidth}
                    \centering
                    \includegraphics[width=\linewidth]{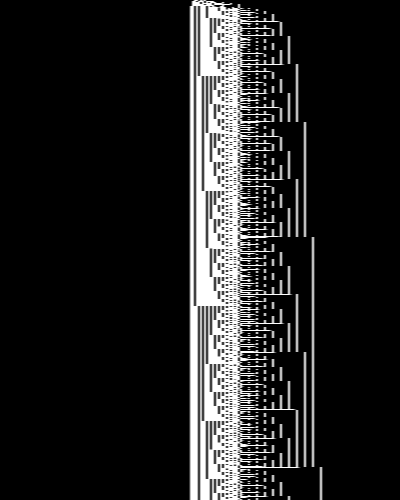}
                    \caption*{\href{https://bbchallenge.org/1RB---_1RC1LB_1LD1RE_1LB0LD_1RA0RC}{Skelet \#15}}
                \end{subfigure}
                \hfill
                \begin{subfigure}{0.17\textwidth}
                    \centering
                    \includegraphics[width=\linewidth]{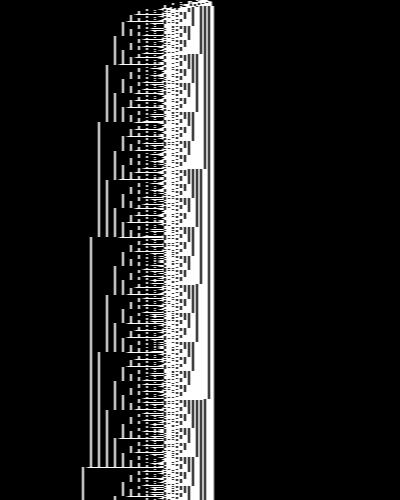}
                    \caption*{\href{https://bbchallenge.org/1RB1LD_1RC0RB_1LA1RC_1LE0LA_1LC---}{Skelet \#26}}
                \end{subfigure}
                \hfill
                \begin{subfigure}{0.17\textwidth}
                    \centering
                    \includegraphics[width=\linewidth]{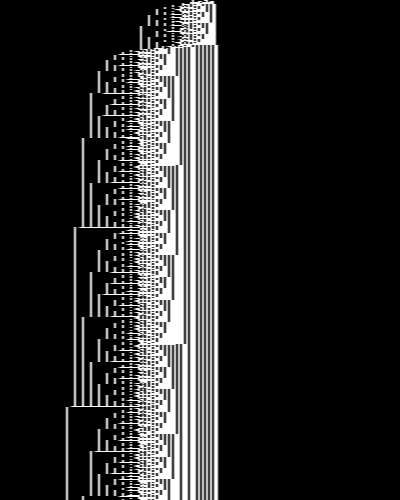}
                    \caption*{\href{https://bbchallenge.org/1RB1LC_0RC0RB_1LD0LA_1LE---_1LA1RE}{Skelet \#33}}
                \end{subfigure}
                \hfill
                \begin{subfigure}{0.17\textwidth}
                    \centering
                    \includegraphics[width=\linewidth]{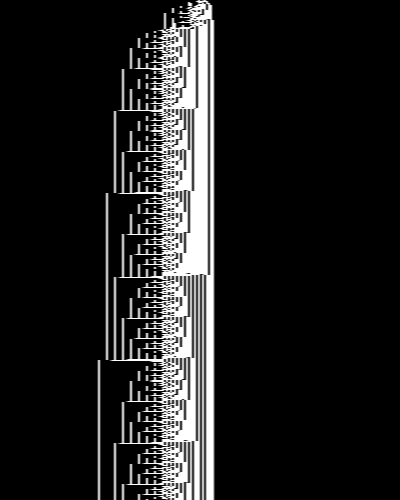}
                    \caption*{\href{https://bbchallenge.org/1RB1LC_0RC0RB_1LD0LA_1LE---_1LA1RA}{Skelet \#34}}
                \end{subfigure}
                \hfill
                \begin{subfigure}{0.17\textwidth}
                    \centering
                    \includegraphics[width=\linewidth]{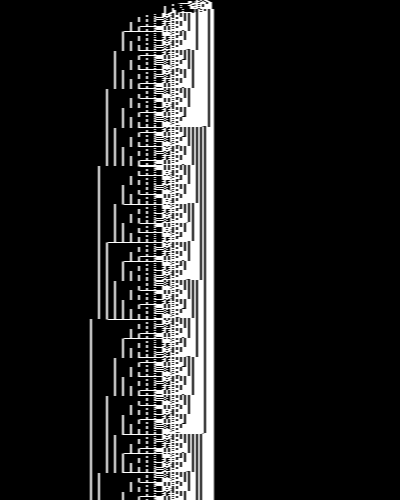}
                    \caption*{\href{https://bbchallenge.org/1RB1LC_0RC0RB_1LD0LA_1LE---_1LA0LA}{Skelet \#35}}
                \end{subfigure}
            \end{minipage}
        };
    \end{tikzpicture}

    \vspace{1.5em}

    \begin{tikzpicture}
        \node[draw=magenta, thick, rounded corners, inner sep=8pt] (box2) {
            \begin{minipage}{0.95\textwidth}
                \centering
                \textbf{\textcolor{magenta}{Finned Machines}}\\[0.8em]
                \begin{subfigure}{0.17\textwidth}
                    \centering
                    \includegraphics[width=\linewidth]{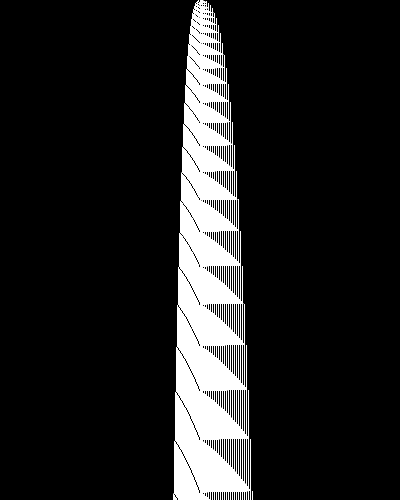}
                    \caption*{\href{https://bbchallenge.org/1RB0LE_1RC1RB_1RD1LC_0LE0RB_---1LA}{Finned \#1}}
                \end{subfigure}
                \hfill
                \begin{subfigure}{0.17\textwidth}
                    \centering
                    \includegraphics[width=\linewidth]{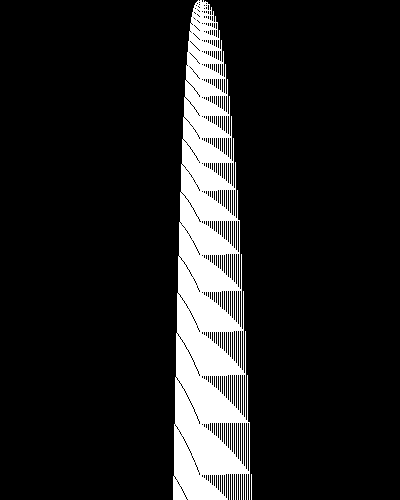}
                    \caption*{\href{https://bbchallenge.org/1RB1RA_1RC1LB_0LD0RA_1RA1LE_---0LD}{Finned \#2}}
                \end{subfigure}
                \hfill
                \begin{subfigure}{0.17\textwidth}
                    \centering
                    \includegraphics[width=\linewidth]{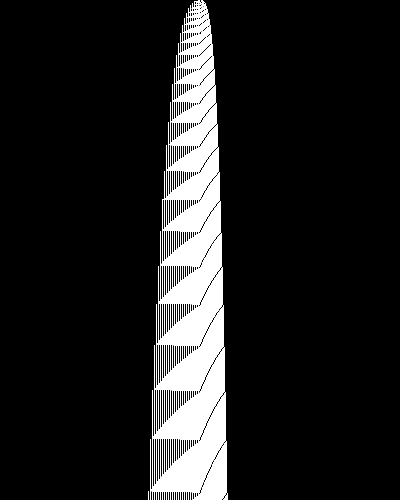}
                    \caption*{\href{https://bbchallenge.org/1RB1RE_1LC1RB_0RA0LD_1LB1LD_---0RA}{Finned \#3}}
                \end{subfigure}
                \hfill
                \begin{subfigure}{0.17\textwidth}
                    \centering
                    \includegraphics[width=\linewidth]{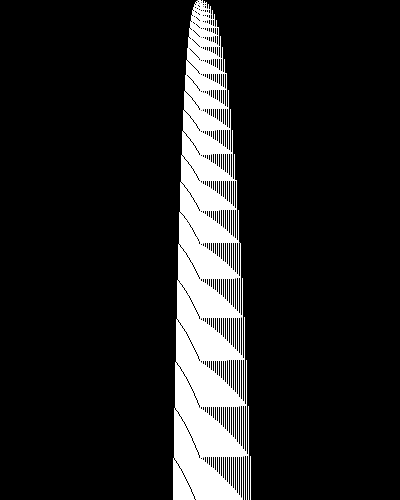}
                    \caption*{\href{https://bbchallenge.org/1RB1LA_0LC0RE_---1LD_1RA0LC_1RA1RE}{Finned \#4}}
                \end{subfigure}
                \hfill
                \begin{subfigure}{0.17\textwidth}
                    \centering
                    \includegraphics[width=\linewidth]{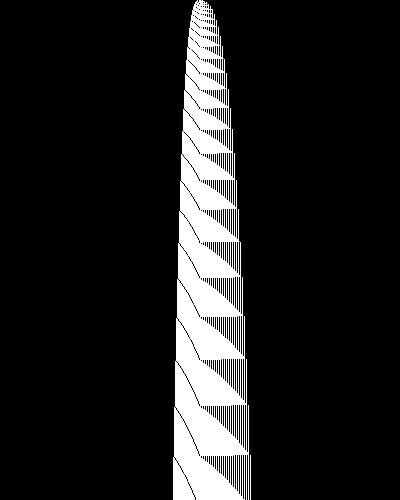}
                    \caption*{\href{https://bbchallenge.org/1RB1LA_0LC0RE_---1LD_1LA0LC_1RA1RE}{Finned \#5}}
                \end{subfigure}
            \end{minipage}
        };
    \end{tikzpicture}

    \caption{{\small Family picture of the 5-state Sporadic Machines (20,000-step space-time diagrams) which required individual \Coq nonhalting proofs; machine names in the Figure are clickable URLs giving the TNF-normalised transition table of each machine (see Section~\ref{sec:enum}). All Sporadic Machines were also identified by Skelet \cite{Skelet_bbfind} as holdouts of his \texttt{bbfind} program. For better visibility, diagrams of counters (Skelet \#10 and Shift Overflow Counters) have been represented using a tape of length 200 instead of 400, giving a \textit{zoomed-in} effect.}}
    \label{fig:sporadic}
\end{figure}

Sporadic Machines are 13 nonhalting 5-state Turing machines that were not captured by deciders (Section~\ref{sec:deciders}) and required individual \Coq proofs of nonhalting; their space-time diagrams are given in Figure~\ref{fig:sporadic} where each name is a clickable URL leading to the machine's transition table and space-time diagram. Twelve of these machines, \ie all but ``Skelet \#17'' (see below), were proved nonhalting in \texttt{busycoq} \cite{busycoq}, and then integrated\footnote{For convenience, the relevant parts of \texttt{busycoq} have been added \href{https://github.com/ccz181078/Coq-BB5/tree/main/BusyCoq}{to the root} of \CoqBB. \CoqBB translates \texttt{busycoq} proofs using \href{https://github.com/ccz181078/Coq-BB5/blob/main/CoqBB5/BB5/BusyCoq_Translation.v}{\texttt{BusyCoq\_Translation.v}}.} into \CoqBB. Machine ``Skelet \#17'' was the last 5-state machine to be formally proven nonhalting in Coq, as part of \CoqBB, achieving the proof of $S(5) = \BBtheFifth$ -- a different proof also had been released as a standalone paper, \cite{xu2024skelet17fifthbusy}.

Interestingly all Sporadic Machines had been identified by Georgi Georgiev (also known as ``Skelet'', see Section~\ref{sec:intro:mainresults}) in 2003: either as part of his 43 unsolved machines\footnote{Apart from \cite{Skelet_bbfind_list}, these 43 machines are also listed here: \url{https://bbchallenge.org/skelet}.} which are named after him, \eg ``Skelet \#1'' (see Figure~\ref{fig:sporadic}), or, in the case of what we call ``Finned Machines'', marked by him as ``easily provable by hand'' \cite{Skelet_bbfind_list}. Sporadic Machines can be arranged in three buckets:
\begin{itemize}
    \item \textbf{Finned Machines.} These are five similar machines that hold 3 unary numbers on the tape (and can merge the middle one into its neighbour), vary them while maintaining a linear relation, and in the process ensure any deviation from this linear relation would be detected and cause a halt. These machines were solved by handcrafting nonhalting certificates similar in flavor to WFAR certificates (Section~\ref{sec:WFAR}). The certificates were crafted by Blanchard, translated in \Coq by Kądziołka, see \texttt{busycoq} files \texttt{Finned\{1-5\}.v}. An argument of irregularity (Section~\ref{sec:deciders-overview}) was given for machine ``Finned \#3'' \cite{irregularFinned3}. A later-developed irregular extension of RepWL (Section~\ref{sec:RepWL}) has been reported to solve these machines.\footnote{\url{https://discuss.bbchallenge.org/t/bb5s-finned-machines-summary/234}}
    \item \textbf{Shift Overflow Counters.} This family concerns Skelet's machines \href{https://bbchallenge.org/1RB---_1RC1LB_1LD1RE_1LB0LD_1RA0RC}{15}, \href{https://bbchallenge.org/1RB1LD_1RC0RB_1LA1RC_1LE0LA_1LC---}{26}, \href{https://bbchallenge.org/1RB1LC_0RC0RB_1LD0LA_1LE---_1LA1RE}{33}, \href{https://bbchallenge.org/1RB1LC_0RC0RB_1LD0LA_1LE---_1LA1RA}{34} and \href{https://bbchallenge.org/1RB1LC_0RC0RB_1LD0LA_1LE---_1LA0LA}{35}; Figure~\ref{fig:sporadic}. These machines are similar: they implement two independent binary counters, one to the left of the tape and the other to the right. Their behaviour can be described by two distinct phases: an orderly ``Counter Phase'' where each counter is simply incremented and a more complex ``Reset Phase'' triggered by one of the counters overflowing. If the counter were to overflow again during a ``Reset Phase'', the machines would halt. Therefore, the proof of nonhalting depends upon demonstrating that the machine maintains a ``Reset Invariant'' throughout the ``Reset Phase'' which does not allow another overflow. These machines were first analysed and described by Ligocki, who provided an informal proof of ``Skelet \#34'' as well as conjectures about the others \cite{ShawnSOC}. They were then proved in \Coq by Yuen and Kądziołka as part of \texttt{busycoq}, see files \texttt{Skelet\{15,26,33,34,35\}.v} \cite{busycoq}.
    \item \textbf{\href{https://bbchallenge.org/1RB1RD_1LC0RC_1RA1LD_0RE0LB_---1RC}{Skelet \#1}, \href{https://bbchallenge.org/1RB0RA_0LC1RA_1RE1LD_1LC0LD_---0RB}{Skelet \#10}, and, \href{https://bbchallenge.org/1RB---_0LC1RE_0LD1LC_1RA1LB_0RB0RA}{Skelet \#17}.} These three machines each have unique behaviours which we detail below.
\end{itemize}
\vspace{-1.5em}
\paragraph{\href{https://bbchallenge.org/1RB1RD_1LC0RC_1RA1LD_0RE0LB_---1RC}{Skelet \#1}.} This machine is a Translated Cycler, \ie a machine that eventually repeats the same pattern translated in space (see Section~\ref{sec:loops}), but with enormous parameters: its pre-period (number of steps to wait before the pattern first appears) is about $5.42 \times 10^{51}$ and its period (number of steps taken by the repeated pattern) is $8,468,569,863$. This was discovered by means of accelerated simulation by Kropitz and Ligocki \cite{uniSk1} and thorough analysis by Ligocki \cite{ShawnSkelet1Before, ShawnSkelet1}. The result was confirmed correct after Kądziołka formalised it in \Coq as part of \texttt{busycoq}, see file \texttt{Skelet1.v} \cite{busycoq}. The $10^{51}$ pre-period was computed later by Huang \cite{hipparcosSk1}.
\vspace{-0.5em}
\paragraph{\href{https://bbchallenge.org/1RB0RA_0LC1RA_1RE1LD_1LC0LD_---0RB}{Skelet \#10 (Double Fibonacci Counter)}.} This machine implements two independent \textit{base Fibonacci} counters, one to the left of the tape and the other to the right. Counting in base Fibonacci means exploiting Zeckendorf's theorem \cite{wiki:Zeckendorf's_theorem}: any natural number can be expressed as a sum of Fibonacci numbers in exactly one way, excluding using numbers immediately adjacent in the Fibonacci sequence, where the Fibonacci sequence is $F = 1,2,3,5,8,13,21,34\dots$ -- each number in the sequence is the sum of the two previous ones. For instance, $17 = 1 + 3 + 13$ and this decomposition would be represented as \texttt{100101} in big-endian binary: the $i^\text{th}$ bit from the right is $1$ if we use $F_i$ in the sum. Each of the two counters of Skelet \#10 enumerate natural numbers in base Fibonacci, using slightly different encodings and the machine halts iff the counters ever get out of sync -- which, does not happen. The machine was analysed independently by Briggs and Ligocki \cite{DanBriggs,ShawnSkelet10} and Ligocki's proof \cite{ShawnSkelet10} was formalised in \Coq by Kądziołka as part of \texttt{busycoq}, see file \texttt{Skelet10.v} \cite{busycoq}. Skelet \#10 is the only known 5-state \textit{double} Fibonacci counter, but there are several known \textit{single} Fibonacci counters, such as \tm{1RB0RA_0LC1RA_1LD0LC_1RE1LC_---0RB}, solved by \CoqBB's NGramCPS (Section~\ref{sec:n-gramCPS}).

\vspace{-0.5em}
\paragraph{\href{https://bbchallenge.org/1RB---_0LC1RE_0LD1LC_1RA1LB_0RB0RA}{Skelet \#17}.} \textit{The final boss.} This machine manages a list of integers $n_1, \dots, n_k \in \mathbb{N}$ represented in unary on the tape using encoding: $(\texttt{10})^{n_1} \texttt{1} (\texttt{10})^{n_2} \texttt{1} \dots \texttt{1} (\texttt{10})^{n_k}$. The list can only increase and undergoes a set of complex transformations related to Gray code, and, the machine halts iff $n_1 = n_2 = 0$ and $n_3,\, \dots,\, n_k$ are all even, which, never happens. This description was first drafted by savask \cite{savaskSk17}, proven in a standalone paper by Xu \cite{xu2024skelet17fifthbusy} and, finally, formalised (using a different proof) in \Coq by mxdys as part of \CoqBB. Skelet \#17 was the last 5-state machine to be solved.

\newpage
\section{Results}\label{sec:results}

\CoqBB is available at \url{https://github.com/ccz181078/Coq-BB5} \cite{mxdys_2025_17061968} and contains extensive instructions for how to compile the proof. The proof compiles in about 45 minutes using 13 cores on a standard laptop. The proof only relies on \Coq's standard library axiom \texttt{functional\_extensionality\_dep}\footnote{See \url{https://rocq-prover.org/doc/v8.9/stdlib/Coq.Logic.FunctionalExtensionality.html}.}, which claims that two functions are equal if they are equal at all points.

\thBBTheFifth*
\begin{proof}
    The \Coq proof enumerates 5-state Turing machines in Tree Normal Form, Section~\ref{sec:enum} and Table~\ref{tab:TNF-numbers}. Each enumerated machine goes through the $S(5)$ pipeline, Table~\ref{tab:pipelineBB5}, where the halting problem from all-zero tape of each machine is solved using a decider (or verifier), Section~\ref{sec:deciders}, unless the machine is one of the \numSporadic Sporadic Machines, Section~\ref{sec:sporadic}, for which individual \Coq proofs of nonhalting are provided.
    When encountering a halting machine, the proof checks that it halts before $47{,}176{,}870$ steps, giving $S(5) \leq 47{,}176{,}870$, see \texttt{Lemma BB5\_upperbound}, and, using the 5-state champion (Figure~\ref{fig:bb5win}), see \texttt{Lemma BB5\_lowerbound}, the proof concludes $S(5) = 47{,}176{,}870$. Thanks to this proof, the 5-state champion becomes the winner of the 5th Busy Beaver competition, Figure~\ref{fig:bb5win}.
\end{proof}

\paragraph{Extracting machines.} The essential output of The Busy Beaver Challenge is the list of all $\BBtheFifthTNF$ TNF-enumerated 5-state machines together with their halting status and method used to determine it. Find the list at \url{https://docs.bbchallenge.org/CoqBB5_release_v1.0.0/}. This list was computed by \textit{extracting} the \Coq proof to OCaml, which means that all the \Coq-implemented algorithms were automatically transcribed  in a trusted way to OCaml by the \Coq engine. From there, print statements were added to the OCaml code which, once ran, produced the list. If given an arbitrary 5-state Turing machine, computing TNF normalisation (see Section~\ref{sec:enum}) and then performing lookup in the list allows to determine the halting status of the machine from all-zero tape.

Using this Coq-verified list, any ``observable'' metric on 5-state Turing machines such as \rado's $\Sigma$ (see Section~\ref{sec:intro}) can be computed:

\begin{theorem}\label{th:Sigma5}
    $\Sigma(5) = 4{,}098$.
\end{theorem}
\begin{proof}
    The winning machine for $\Sigma(5)$ is the same as for $S(5)$, Figure~\ref{fig:bb5win}.
    This metric was computed from the Coq-extracted list of all TNF-enumerated 5-state Turing machines, see above. Three agreeing independent reproductions of the computation were asked to ensure there was no mistake made.
\end{proof}

Another similar observable, called $\text{space}(n)$ \cite{Ben-Amram1996} (also called $\text{BB}_{\text{SPACE}}(n)$ \cite{sterin_2022_14955828}) is the maximum number of tape cells that an $n$-state Turing machine may scan before it halts. This observable is similar to $S$ in the sense that, if a machine visits $\text{space}(n)+1$ tape cells, we know it will never halt (from the all-zero tape). We get:

\begin{theorem}
    $\text{space}(5) = 12{,}289$.
\end{theorem}
\begin{proof}
    The winning machine for $\text{space}(5)$ is the same as for $S(5)$, Figure~\ref{fig:bb5win}.
    Same method as for $\Sigma$, Theorem~\ref{th:Sigma5}.
\end{proof}

\CoqBB also computes $S$ for 2-state 4-symbol Turing machines:

\thBBTxF*

\begin{proof}
    Similarly to Theorem~\ref{th:BB5}, the \Coq proof enumerates 2-state 4-symbol machines in (almost) Tree Normal Form, see Section~\ref{sec:enum} and Table~\ref{tab:TNF-numbers}. Then the $S(2,4)$ pipeline, Table~\ref{tab:pipelineBB2x4}, which consists only of deciders (Section~\ref{sec:deciders}) is applied to solve the halting problem from all-zero tape of all the enumerated machines. The proof keeps track of the maximum number of steps reached by halting machines and eventually concludes $S(2,4) = \BBTxF$. Thanks to this proof, the 2-state 4-symbol champion (Figure~\ref{fig:bb2x4}) becomes the winner among all 2-state 4-symbol machines.
\end{proof}

Additionally, as illustrated in Table~\ref{table:landscape}, \CoqBB also provides \Coq proofs for previously known values of $S$, including $S(4)$ of which original proof \cite{Brady83} had slight uncertainties -- see Section~\ref{sec:intro}:

\begin{theorem}[Confirmation of Brady's result \cite{Brady83}]\label{th:BB4}
    $S(4) = \BBtheFourth$.
\end{theorem}
\begin{proof}
    Same as Theorem~\ref{th:BB5}, using the $S(4)$ pipeline, Table~\ref{tab:pipelineBB4}, which consists only of deciders (Section~\ref{sec:deciders}), \ie no individual proofs of nonhalting.
\end{proof}

\newpage
\section{Zoology}\label{sec:zoo}
\vspace{-1em}
\begin{figure}[htbp]
    \centering

    \begin{subfigure}{0.3\textwidth}
        \centering
        \includegraphics[width=\linewidth]{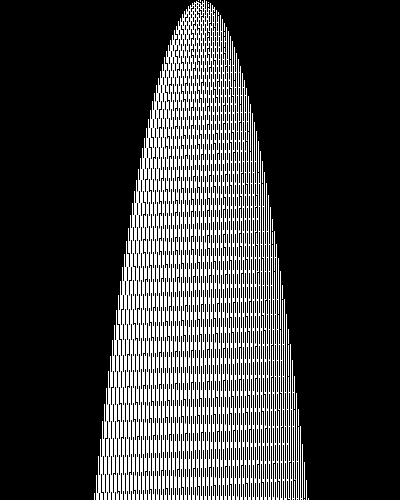}
        \caption*{Bouncers}
        {\scriptsize \tm{1RB0LE_0RC1RE_1LD1RA_0LA---_0RC0LB}}
    \end{subfigure}
    \hfill
    \begin{subfigure}{0.3\textwidth}
        \centering
        \includegraphics[width=\linewidth]{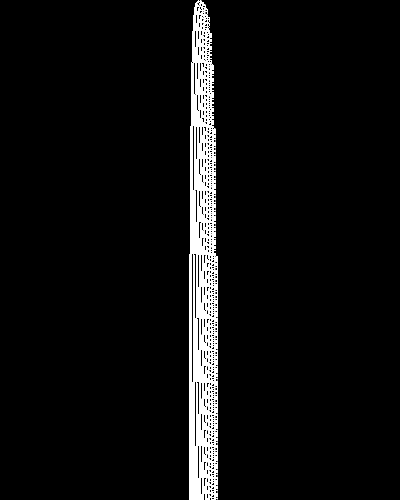}
        \caption*{Counters}
        {\scriptsize \tm{1RB1LA_0LA0RC_0LC1RD_1RE0LA_0RB---}}
    \end{subfigure}
    \hfill
    \begin{subfigure}{0.3\textwidth}
        \centering
        \includegraphics[width=\linewidth]{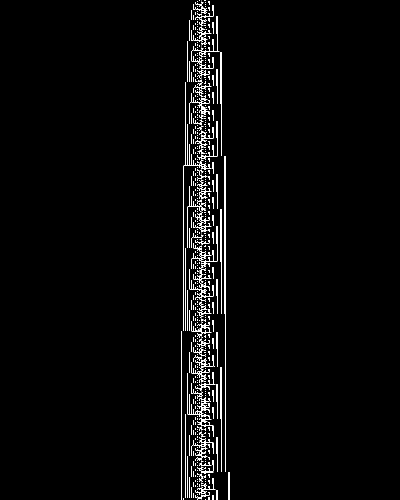}
        \caption*{Double Counters}
        {\scriptsize \tm{1RB1LE_1LC0RB_1RD0LC_1LA1RB_1LB---}}
    \end{subfigure}

    \vspace{0.5cm}

    \begin{subfigure}{0.3\textwidth}
        \centering
        \includegraphics[width=\linewidth]{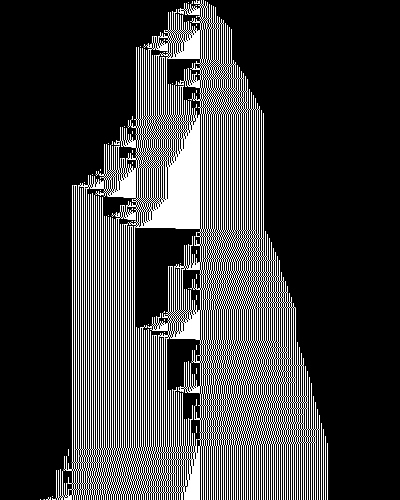}
        \caption*{Fractals}
        {\scriptsize \tm{1RB0RB_1LC0RA_1LD0LB_0LE---_1RE1LB}}
    \end{subfigure}
    \hfill
    \begin{subfigure}{0.3\textwidth}
        \centering
        \includegraphics[width=\linewidth]{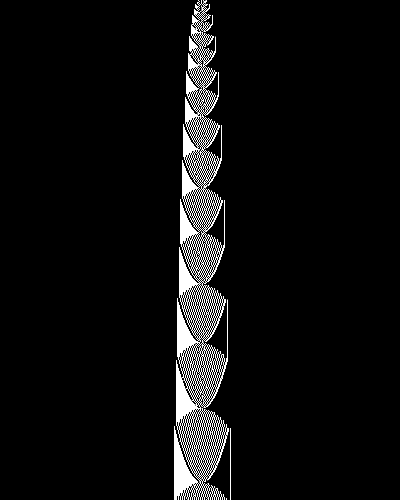}
        \caption*{Helices}
        {\scriptsize \tm{1RB0RE_0RC0RA_1LD---_1LA0LB_1RA0LC}}
    \end{subfigure}
    \hfill
    \begin{subfigure}{0.3\textwidth}
        \centering
        \includegraphics[width=\linewidth]{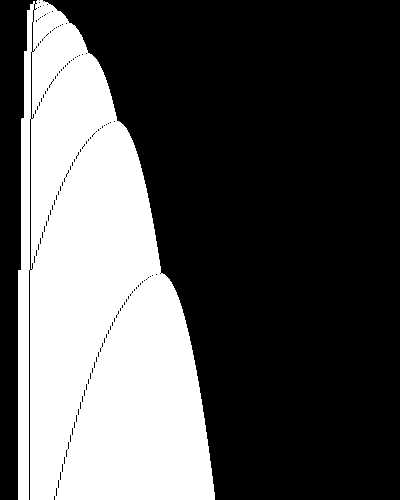}
        \caption*{Bells}
        {\scriptsize \href{https://bbchallenge.org/antihydra}{Antihydra, 6-state Cryptid}}
    \end{subfigure}

    \caption{Main zoological families that were identified among 5-state Turing machines, together with Cyclers and Translated Cyclers which are not illustrated here (see Section~\ref{sec:loops}). }
    \label{fig:zoology}
\end{figure}

As we were \textit{computing in the wild}, we identified several \textit{zoological families} among 5-state Turing machines:

\begin{enumerate}
    \item \textbf{Cyclers and Translated Cyclers}; Section~\ref{sec:loops}. Translated Cyclers are the most common \textit{species} among 5-state machines -- estimated to be about 80\% of TNF-enumerated machines.
    \item \textbf{Bouncers}; Figure~\ref{fig:zoology}. This family consists of machines that populate the tape with linearly-expanding patterns while bouncing back and forth from tape extremities. Bouncers have been formally defined and efficient deciders (not used in \CoqBB) have been crafted to detect them \cite{bbchallenge_part1}.
    \item \textbf{Counters}\footnote{Also referred to as \textit{exponential counters}. Since they count in a base > 1, it takes exponential time to add further digits to the tape.}; Figure~\ref{fig:zoology}. Counters are machines that enumerate numbers in a basis bigger than $1$. There is a rich variety of 5-state counters, counting in all sorts of bases: base 2 (Figure~\ref{fig:finite-automata-reduction}), \href{https://bbchallenge.org/1RB1RB_1RC0LD_1LD1RA_---1LE_0RA0LE }{base 3}, \href{https://bbchallenge.org/1LB1RC_0LE0RA_1LD1RA_0RA1LB_0RD0LB}{base 3/2}, base Fibonacci (Section~\ref{sec:sporadic}) and \href{https://bbchallenge.org/1RB1RA_0LC1LC_0RD1LD_0RA0LB}{Fibonacci variants} $A(n) = A(n-1) + A(n-3)$ etc.\footnote{Classifying 5-state counters would be a beautiful project.}
    \item \textbf{Double Counters}; Figure~\ref{fig:zoology}. Double Counters implement two independent counters, Skelet \#10 is an example (Section~\ref{sec:sporadic}). The example of Figure~\ref{fig:zoology} is a base-2 counter on the left-side of the tape and base-3 counter on the right-side.
    \item \textbf{Fractals}; Figure~\ref{fig:zoology}. Fractal machines are loosely defined as machines whose space-time diagrams draw a fractal, self-similar, pattern when ``zooming-out'' (\ie simulating more steps). There exist some hybrids, such as this \href{https://bbchallenge.org/1RB1RC_1RC1RB_1LD0RA_---1LE_0LD0LA&s=20000}{Sierpiński triangle} growing off the side of a bouncer.
    \item \textbf{Helices}; Figure~\ref{fig:zoology}. Helices are loosely defined as Turing machines whose space-time diagrams resemble a double helix. Helices require big nonregular certificates of nonhalting, see Section~\ref{sec:WFAR}.
    \item \textbf{Bells}; Figure~\ref{fig:zoology}. Bells are loosely defined as Turing machines whose space-time diagrams resemble a succession of bells. The 5-state winner (Figure~\ref{fig:bb5win}) and the 2-state 4-symbol winner (Figure~\ref{fig:bb2x4}) and the 6-state Cryptid Antihydra (Section~\ref{sec:intro:discuss} and Appendix~\ref{app:cryptids}) fit this category.
\end{enumerate}

Although useful for talking about Turing machines, these families, especially when they are loosely defined, are not to be taken too seriously: (i) widely different behaviours can be implemented while maintaining a similar space-time diagram \textit{silhouette} and (ii) the zoological effort can quickly become vain given all the possible hybrids (\eg \href{https://bbchallenge.org/1RB---_1LC1RE_0RD0LC_1LB1LA_0RA1RE}{bouncer-counter}) and variations within the same family, especially as the number of states increases. Attesting to the limits of a zoological effort, here are eccentric 5-state machines: \href{https://bbchallenge.org/1RB0RA_1LC0LD_1RE1RD_1LA1LB_---1RC}{translated counter}, \href{https://bbchallenge.org/1RB0RD_1LC1LB_1RA0LB_0RE1RD_---1RA}{``fountain''}, and, \href{https://bbchallenge.org/1RB0RC_0LC---_1RD1RC_0LE1RA_1RD1LE}{``toboggan''}.

\addcontentsline{toc}{section}{References}
\bibliographystyle{abbrv}
\bibliography{bbchallenge-paper}
\begin{figure}[h!]
    \centering
    \includegraphics[scale=0.27]{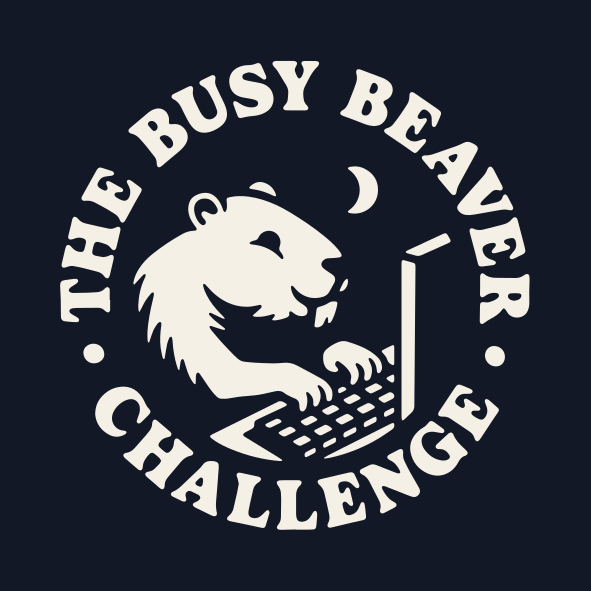}
    \caption{\texttt{bbchallenge} logo; credits: Léo Ramaën (\url{https://leoramaen.com/}).}
\end{figure}

\appendix
\addtocontents{toc}{\protect\setcounter{tocdepth}{1}}
\newpage

\section{Author Contributions}\label{app:contribs}

\paragraph{The bbchallenge Collaboration: $S(5)$ credits.} The following contributions resulted in the determination of the fifth Busy Beaver value and in the better understanding of the landscape of small Busy Beaver values (Table~\ref{table:landscape}): mxdys (\CoqBB, Loops, RepWL); Nathan Fenner, Georgi Georgiev a.k.a~Skelet, savask, mxdys (NGramCPS); Justin Blanchard, Mateusz Naściszewski, Konrad Deka (FAR); Iijil (WFAR); Maja Kądziołka (\texttt{busycoq}); Shawn Ligocki, Jason Yuen, Maja Kądziołka (Sporadic Machines ``Shift Overflow Counters''); Shawn Ligocki, Pavel Kropitz, Maja Kądziołka (Sporadic Machine ``Skelet \#1''); savask, Chris Xu, mxdys (Sporadic Machine ``Skelet \#17''); Shawn Ligocki, Dan Briggs, Maja Kądziołka (Sporadic Machine ``Skelet \#10''); Justin Blanchard, Maja Kądziołka (Sporadic Machines ``Finned Machines''); Shawn Ligocki, Daniel Yuan, mxdys, Matthew L. House, Rachel Hunter, Jason Yuen (``Cryptids''); Yannick Forster, Théo Zimmermann (Coq review); Yannick Forster (Coq optimisation); Tristan Stérin (bbchallenge.org); Tristan Stérin, Justin Blanchard (paper writing).
\vspace{-1ex}
{\small\raggedright
    \begin{multicols}{2}
        \begin{itemize}\setlength{\itemsep}{1pt}\setlength{\parskip}{0pt}
            \item The bbchallenge Collaboration,\\ \url{bbchallenge.org},\\ \texttt{bbchallenge@bbchallenge.org}
            \item Justin Blanchard,\\ \texttt{UncombedCoconut@gmail.com}
            \item Daniel Briggs, \texttt{dbriggs@alum.mit.edu}
            \item Konrad Deka, \texttt{deka.konrad@gmail.com}
            \item Nathan Fenner,\\ \texttt{nfenneremail@gmail.com}
            \item Yannick Forster, INRIA Paris,\\ \texttt{yannick.forster@inria.fr}
            \item Georgi Georgiev (Skelet),\\ Sofia University, Faculty of Mathematics and Informatics,\\ \texttt{skeleta@gmail.com}
            \item Rachel Hunter,\\ \texttt{racheline@bbchallenge.org}
            \item Matthew L. House,\\ \texttt{mattlloydhouse@gmail.com}
            \item Iijil, \texttt{hheussen@web.de}
            \item Maja Kądziołka, University of Warsaw,\\ \texttt{bb@compilercrim.es}
            \item Pavel Kropitz, \texttt{uni@bbchallenge.org}
            \item Shawn Ligocki, \texttt{sligocki@gmail.com}
            \item mxdys, \texttt{mxdys@bbchallenge.org}
            \item Mateusz Na\'{s}ciszewski,\\ \texttt{mateusz.nasciszewski@gmail.com}
            \item savask
            \item Tristan Stérin$^{\ast}$, PRGM DEV,\\ \texttt{tristan@prgm.dev}
            \item Chris Xu, UC San Diego,\\ \texttt{chx007@ucsd.edu}
            \item Jason Yuen,\\ \texttt{jason\_yuen2007@hotmail.com}
            \item Théo Zimmermann, LTCI, Télécom Paris, Institut Polytechnique de Paris, France,\\ \texttt{theo.zimmermann@telecom-paris.fr}
        \end{itemize}
    \end{multicols}
}
\vspace{-4ex}
\paragraph{Acknowledgement.} The bbchallenge Collaboration is not limited to the above contributors but regroups all those who participated in the bbchallenge's discussions through all our channels (Discord chat, forum, wiki, GitHub, emails) and in particular we thank:  Nick Drozd, Andrew Ducharme, Frans Faase, Tony Guilfoyle, Johannes Hostert, Nick Howell, Jeffrey Huang, Alexandre Jouandin, Carl Kadie, Frank S. Lin, Dawid Loranc, Terry J. Ligocki, Heiner Marxen, Pascal Michel, Milo Mighdoll (milomg), Seraphina Nix, Sébastien Ohleyer, Peacemaker II, Andrés Sancho, tomtom2357, Valentin, Daniel Yuan, Polygon.

We'd like to thank \texttt{bbchallenge.org} sponsor \texttt{prgm.dev}. We'd also like to thank those who gave feedback on early versions of this paper, including the following GitHub and Discord users: desseim, frank-s-lin, Lysxia, Nitrome, RobinCodes, XnoobSpeakable, ZhiqiuCao.

We'd also like to thank the following people who helped disseminate the bbchallenge project: Damien Woods, Dave Doty, Eric E. Severson, Scott Aaronson, Timothy Gowers, Jean-Claude Bermond, Luc Albert, Ben Brubaker, and Léo Ramaën.

\begingroup
\renewcommand\thefootnote{$\ast$}
\footnotetext{Funding: \texttt{prgm.dev} through the French state (\textit{Crédit Impôt Recherche} and \textit{Jeune Entreprise Innovante} programs).}
\endgroup

\newpage

\vspace*{-5em}
\section{Busy Beaver winners and champions}\label{app:lowerbounds}

\paragraph{Winners.} Below we give the Busy Beaver winners for known values of $S$ and $\Sigma$. We give all the machines that are \textit{ex aequo} for these metrics -- only considering machines in TNF; see Section~\ref{sec:enum}. Note that the use of \texttt{1RZ} in the Turing machines below instead of \verb|---| means that \textit{we know} that machines halt; see Section~\ref{sec:TMs}.
\begin{itemize}
    \item $S(2) = 6$ is achieved by 5 winners: \tm{1RB1LB_1LA1RZ}, \tm{1RB0LB_1LA1RZ}, \tm{1RB1RZ_1LB1LA}, \tm{1RB1RZ_0LB1LA}, and \tm{0RB1RZ_1LA1RB}; $\Sigma(2) = 4$ is only achieved by the first $S$ winner of the list.
    \item $S(3) = 21$ is achieved by 1 winner: \tm{1RB1RZ_1LB0RC_1LC1LA}; $\Sigma(3) = 6$ is achieved by 5 winners: \tm{1RB1RZ_0RC1RB_1LC1LA}, \tm{1RB1RC_1LC1RZ_1RA0LB}, \tm{1RB1LC_1LA1RB_1LB1RZ}, \tm{1RB1RA_1LC1RZ_1RA1LB}, and \tm{1RB1LC_1RC1RZ_1LA0LB}. Note that this is the only known case where no $S$ winner is a $\Sigma$ winner. Originally proved in \cite{Lin1963}.
    \item $S(4) = 107$ is achieved by 1 winner: \tm{1RB1LB_1LA0LC_1RZ1LD_1RD0RA}; $\Sigma(4) = 13$ is achieved by 2 winners: the winner for $S$ and \tm{1RB0RC_1LA1RA_1RZ1RD_1LD0LB}. Originally proved in \cite{Brady83}.
    \item $S(5) = \BBtheFifth$ is achieved by 1 winner: \tm{1RB1LC_1RC1RB_1RD0LE_1LA1LD_1RZ0LA}; $\Sigma(5) = 4{,}098$ is achieved by 2 winners: the winner for $S$ and \tm{1RB1RA_1LC1LB_1RA1LD_1RA1LE_1RZ0LC}.
    \item $S(2,3) = 38$ is achieved by 1 winner, \tm{1RB2LB1RZ_2LA2RB1LB}, which is also the only winner for $\Sigma(2,3) = 9$. Originally proved in \cite{LafittePapazian2007}.
    \item $S(2,4) = \BBTxF$ is achieved by 1 winner, \tm{1RB2LA1RA1RA_1LB1LA3RB1RZ}, which is also the only winner for $\Sigma(2,4)=2{,}050$.
\end{itemize}

All these values have been verified in \CoqBB, see Section~\ref{sec:intro:discuss} and Section~\ref{sec:results}.

\paragraph{Champions.} Below are the current Busy Beaver champions for some of the next Busy Beaver values that are still unknown, \ie machines achieving higher $S$ scores could be found in the future. Note that, for most of these champions, determining the exact number of steps is harder than proving that the machine halts; consequently, most of these bounds are not \textit{exact} step counts but strict lower bounds of these counts.\footnote{Also, at these scales, $S$ and $\Sigma$ are believed to be roughly the same number since, experimentally, $\Sigma$ is always ``only'' quadratically smaller than $S$.}
\begin{itemize}
    \item $S(6) > 2 \uparrow \uparrow \uparrow 5$; this bound comes from \tm{1RB1RA_1RC1RZ_1LD0RF_1RA0LE_0LD1RC_1RA0RE} which has been proved to halt in Coq\footnote{\url{https://github.com/ccz181078/busycoq/blob/3f302b87f5fb933c46e97672ffbb6907f373fb6e/verify/SOBCv5.v\#L10210-L11283}} and whose number of steps was then manually estimated using the machine's \Coq-verified abstracted rules of evolution.\footnote{\url{https://wiki.bbchallenge.org/wiki/1RB1RA_1RC1RZ_1LD0RF_1RA0LE_0LD1RC_1RA0RE}}
    \item $S(7) > 2 \uparrow^{11} 2 \uparrow^{11} 3$; this bound comes from \tm{1RB0RA_1LC1LF_1RD0LB_1RA1LE_1RZ0LC_1RG1LD_0RG0RF}, see analysis.\footnote{\url{https://wiki.bbchallenge.org/wiki/1RB0RA_1LC1LF_1RD0LB_1RA1LE_1RZ0LC_1RG1LD_0RG0RF}}
    \item $S(3,3) \geq 119{,}112{,}334{,}170{,}342{,}541$; this bound is the exact number of steps performed by \tm{0RB2LA1RA_1LA2RB1RC_1RZ1LB1LC}, yielding the simplified bound $S(3,3) > 10^{17}$. Only 6 unsolved 3-state 3-symbol Turing machines currently remain, including \tm{1RB2LC1RC_2LC---2RB_2LA0LB0RA} which is believed to halt and which number of steps would significantly surpass the current bound, see Section~\ref{sec:fw}.\footnote{\url{https://wiki.bbchallenge.org/wiki/BB(3,3)}}
    \item $S(4,3) > 10 \uparrow^4 4$; this bound comes from \tm{1RB1RD1LC_2LB1RB1LC_1RZ1LA1LD_0RB2RA2RD}, see analysis.\footnote{\url{https://wiki.bbchallenge.org/wiki/1RB1RD1LC_2LB1RB1LC_1RZ1LA1LD_0RB2RA2RD}}
    \item $S(3,4) > 2 \uparrow^{15} 5$; this bound comes from \tm{1RB3LB1RZ2RA_2LC3RB1LC2RA_3RB1LB3LC2RC}. For this machine it is possible to prove that the number of \sone symbols on the final tape is \textit{exactly} $(2 \uparrow^{15} 5) + 14$, see \cite{ligocki2024bb34ack14}.
    \item $S(2,5) > 10 \uparrow \uparrow 4$; this bound comes from \tm{1RB3LA4RB0RB2LA_1LB2LA3LA1RA1RZ}, see analysis.\footnote{\url{https://wiki.bbchallenge.org/wiki/1RB3LA4RB0RB2LA_1LB2LA3LA1RA1RZ}}
\end{itemize}

These bounds are subject to improvement as new champions are discovered. For the most up-to-date information, please refer to our wiki\footnote{\url{https://wiki.bbchallenge.org/wiki/Main_Page}} or Michel’s website \cite{PMichel_website}.

\section{Cryptids}\label{app:cryptids}

Cryptids are Turing machines whose halting problem from all-zero tape is believed to be mathematically hard; see Section~\ref{sec:intro:discuss}. Here we give Cryptids that were found for $S(6)$, $S(3,3)$, and $S(2,5)$; see Table~\ref{table:landscape} (bright orange cells).

\begin{itemize}
    \item $S(6)$: \tm{1RB1RA_0LC1LE_1LD1LC_1LA0LB_1LF1RE_---0RA}; \textbf{Antihydra}.\footnote{\label{note:antihydra}\url{https://wiki.bbchallenge.org/wiki/Antihydra}} Take the following parametrised configuration:
          $$ A(a,b) = \szero^\infty \; \sone^a \; \szero \; \sone^b \; \texttt{E>} \; \szero^\infty$$

          The following Collatz-like rules can be proved:
          \[
              \begin{array}{llcll}
                  A(a, & 2k)   & \longrightarrow & A(a+2, 3k+2)                         \\
                  A(0, & 2k+1) & \longrightarrow & \texttt{Halt}                        \\
                  A(a, & 2k+1) & \longrightarrow & A(a-1, 3k+3)  & \text{ if } a \geq 1 \\
              \end{array}
          \]

          The machine reaches $A(0,4)$ after 11 steps and halts if and only if it reaches a config $A(0,2k+1)$. This can be reformulated as whether repeatedly applying the Collatz-like \textit{Hydra map}\footnote{\url{https://wiki.bbchallenge.org/wiki/Hydra_function}} $H(n) = \lfloor \frac{3}{2}n\rfloor$ starting from $n=8$ will ever have reached more odd values than twice the number of reached even values \cite{ligocki2024bb6antis}. A probabilistic biased random walk model\footnotemark[\getrefnumber{note:antihydra}] suggests that the chance of Antihydra ever halting are minuscule (less than $10^{-200\,000\,000}$). The rules of a 6-state machine very similar to Antihydra have been proven in \Coq.\footnote{\url{https://github.com/ccz181078/busycoq/blob/BB6/verify/AntiHydra2.v}} A dozen other 6-state Cryptids have been identified to date.\footnote{\url{https://wiki.bbchallenge.org/wiki/BB(6)\#Cryptids}}

    \item $S(3,3)$: \tm{1RB2RA1LC_2LC1RB2RB_---2LA1LA}; \textbf{Bigfoot}.\footnote{\label{note:bigfoot}\url{https://wiki.bbchallenge.org/wiki/Bigfoot}} Take the following parametrised configuration:
          $$B(a,b,c) = \szero^\infty\; (\sone \texttt{2})^a \; (\sone \sone)^b \; \texttt{<A} \; (\sone \sone)^c \; \szero^\infty$$

          The following Collatz-like rules can be proved:
          \[
              \begin{array}{lllcllll}
                  B(a, & 6k,   & c) & \to & B(a,   & 8k + c - 1, & 2)                        \\
                  B(a, & 6k+1, & c) & \to & B(a+1, & 8k + c - 1, & 3)                        \\
                  B(a, & 6k+2, & c) & \to & B(a-1, & 8k + c + 3, & 2) & \text{ if } a \geq 1 \\
                  B(a, & 6k+3, & c) & \to & B(a,   & 8k + c + 1, & 5)                        \\
                  B(a, & 6k+4, & c) & \to & B(a+1, & 8k + c + 3, & 2)                        \\
                  B(a, & 6k+5, & c) & \to & B(a,   & 8k + c + 5, & 3)                        \\
              \end{array}
          \]

          The machine enters configuration $B(2,1,2)$ at step 69 and halts if and only if it reaches a config $B(0,6k+2,c)$ \cite{ligocki2023bb33}. A probabilistic model\footnotemark[\getrefnumber{note:bigfoot}] suggests that the chance of Bigfoot ever halting are minuscule (less than $10^{-1\,000\,000}$).

    \item $S(2,5)$: \tm{1RB3RB---3LA1RA_2LA3RA4LB0LB0LA}; \textbf{Hydra}.\footnote{\label{note:hydra}\url{https://wiki.bbchallenge.org/wiki/Hydra}} Take the following parametrised configuration:
          $$C(a,b) = \szero^\infty \; \texttt{<B} \; \szero^a \; \texttt{3}^b \; \texttt{2} \; \szero^\infty$$

          The following Collatz-like rules can be proved:
          \[
              \begin{array}{llcll}
                  C(2k,   & 0) & \longrightarrow & \texttt{Halt}                        \\
                  C(2k,   & b) & \longrightarrow & C(3k+3, b-1)  & \text{ if } b \geq 1 \\
                  C(2k+1, & b) & \longrightarrow & C(3k+3, b+2)                         \\
              \end{array}
          \]

          The machine reaches $C(3,0)$ after 19 steps and halts if and only if it reaches a config $C(2k,0)$. This can be reformulated as whether repeatedly applying the \textit{Hydra map} $H(n) = \lfloor \frac{3}{2}n\rfloor$ from $n=3$ will ever have reached more even values than twice the number of reached odd values \cite{ligocki2024bb25}. This is the same Hydra map used for Antihydra above, but with the "opposite" condition (switching the role of even and odd values). A probabilistic model\footnotemark[\getrefnumber{note:hydra}] suggests that the chance of Hydra ever halting are minuscule (less than $10^{-400\,000}$).

\end{itemize}

\newpage
\section{Exact \CoqBB pipelines}\label{app:pipelines}

Tables~\ref{tab:pipelineBB5}, \ref{tab:pipelineBB2x4}, and~\ref{tab:pipelineBB4} give simplified descriptions of the proof technique (mainly, deciders) pipelines (see Section \ref{sec:pipelines}) implemented in \CoqBB.\footnote{\url{https://github.com/ccz181078/Coq-BB5/} \cite{mxdys_2025_17061968}} Here we give the exact pipelines that were used. The identifiers we use for identifying deciders are the same as used in the released\footnote{\url{https://docs.bbchallenge.org/CoqBB5_release_v1.0.0/}} \Coq-extracted lists of enumerated machines (in TNF, see Section~\ref{sec:enum}).

In the following tables, formal decider IDs are built as follows: \texttt{LOOP1\_params\_107} corresponds to Loops (Section~\ref{sec:loops}) with step-limit parameter $L=107$; \texttt{NGRAM\_CPS\_IMPL2\_params\_1\_1\_100} corresponds to NGramCPS (Section~\ref{sec:n-gramCPS}) with no augmentation, left n-gram size $1$, right n-gram size $1$, and an additional parameter, here set to $100$, limiting the size of the set of local configurations; \texttt{NGRAM\_CPS\_IMPL1\_params\_4\_2\_2\_600} corresponds to NGramCPS with fixed-length history, the first parameter is history length (set to $4$ in this example), then same as \texttt{IMPL2}; \texttt{NGRAM\_CPS\_LRU\_params\_2\_2\_10000} corresponds to NGramCPS with Least Recent Usage history, parameters are the same as from \texttt{IMPL2}; \texttt{REPWL\_params\_4\_3\_320\_10000} corresponds to RepWL (Section~\ref{sec:RepWL}) with block length $l=4$, repeat threshold $T=3$, block simulation parameter set to $320$ and maximum number of graph nodes set to $10{,}000$; in the case of $S(5)$ \textit{table based} machines (see Section~\ref{app:pipeS5exact}), the use of $\texttt{params\_custom}$ indicates hardcoded parameters that vary for each decided machine.

\subsection{Exact $S(2,4)$ and $S(4)$ pipelines}

Pipelines for $S(2,4)$ and $S(4)$ only use deciders, \ie there are no hardcoded verifier certificates, individual proofs of nonhalting or reduction argument used; see Section~\ref{sec:deciders-overview}. Both $S(4)$ and $S(2,4)$ only use the following deciders: Loops (Section~\ref{sec:loops}), NGramCPS (Section~\ref{sec:n-gramCPS}), and RepWL (Section~\ref{sec:RepWL}); which are all regular deciders (see Section~\ref{sec:deciders-overview}). The pipeline of $S(2,4)$ simulates 24 machines (including the $S(2,4)$ champion) up to $3{,}932{,}964$ steps in order to prove that they halt.

\begin{table}[h!]
    \centering
    \scriptsize
    \begin{tabular}{|lrrr|}
        \hline
        Exact $S(4)$ pipeline                               & Nonhalt   & Halt      & Total decided \\
        \hline
        \texttt{LOOP1\_params\_107}                         & 588{,}373 & 249{,}693 & 838{,}066     \\
        \texttt{NGRAM\_CPS\_IMPL2\_params\_1\_1\_100}       & 11{,}644  & 0         & 11{,}644      \\
        \texttt{NGRAM\_CPS\_IMPL2\_params\_2\_2\_200}       & 4{,}759   & 0         & 4{,}759       \\
        \texttt{NGRAM\_CPS\_IMPL2\_params\_3\_3\_400}       & 1{,}731   & 0         & 1{,}731       \\
        \texttt{NGRAM\_CPS\_IMPL1\_params\_2\_2\_2\_1600}   & 2{,}296   & 0         & 2{,}296       \\
        \texttt{NGRAM\_CPS\_IMPL1\_params\_2\_3\_3\_1600}   & 161       & 0         & 161           \\
        \texttt{NGRAM\_CPS\_IMPL1\_params\_4\_2\_2\_600}    & 174       & 0         & 174           \\
        \texttt{NGRAM\_CPS\_IMPL1\_params\_4\_3\_3\_1600}   & 29        & 0         & 29            \\
        \texttt{NGRAM\_CPS\_IMPL1\_params\_6\_2\_2\_3200}   & 14        & 0         & 14            \\
        \texttt{NGRAM\_CPS\_IMPL1\_params\_6\_3\_3\_3200}   & 10        & 0         & 10            \\
        \texttt{NGRAM\_CPS\_IMPL1\_params\_8\_2\_2\_1600}   & 8         & 0         & 8             \\
        \texttt{NGRAM\_CPS\_IMPL1\_params\_8\_3\_3\_1600}   & 3         & 0         & 3             \\
        \texttt{NGRAM\_CPS\_LRU\_params\_2\_2\_10000}       & 8         & 0         & 8             \\
        \texttt{NGRAM\_CPS\_IMPL1\_params\_10\_4\_4\_10000} & 4         & 0         & 4             \\
        \texttt{REPWL\_params\_4\_3\_320\_10000}            & 2         & 0         & 2             \\
        \hline
        Total                                               & 609{,}216 & 249{,}693 & 858{,}909     \\
        \hline
    \end{tabular}
\end{table}

\begin{table}[h!]
    \centering
    \scriptsize
    \begin{tabular}{|lrrr|}
        \hline
        Exact $S(2,4)$ pipeline                           & Nonhalt       & Halt      & Total decided \\
        \hline
        \texttt{LOOP1\_params\_107}                       & 1{,}262{,}432 & 720{,}959 & 1{,}983{,}391 \\
        \texttt{NGRAM\_CPS\_IMPL2\_params\_1\_1\_400}     & 102{,}018     & 0         & 102{,}018     \\
        \texttt{NGRAM\_CPS\_IMPL2\_params\_2\_2\_800}     & 49{,}224      & 0         & 49{,}224      \\
        \texttt{NGRAM\_CPS\_IMPL2\_params\_3\_3\_400}     & 7{,}518       & 0         & 7{,}518       \\
        \texttt{NGRAM\_CPS\_IMPL2\_params\_4\_4\_800}     & 2{,}286       & 0         & 2{,}286       \\
        \texttt{LOOP1\_params\_4100}                      & 870           & 354       & 1{,}224       \\
        \texttt{REPWL\_params\_2\_3\_320\_400}            & 6{,}012       & 0         & 6{,}012       \\
        \texttt{NGRAM\_CPS\_LRU\_params\_2\_2\_1000}      & 1{,}206       & 0         & 1{,}206       \\
        \texttt{NGRAM\_CPS\_IMPL1\_params\_2\_2\_2\_3000} & 894           & 0         & 894           \\
        \texttt{NGRAM\_CPS\_IMPL1\_params\_2\_3\_3\_1600} & 120           & 0         & 120           \\
        \texttt{NGRAM\_CPS\_IMPL1\_params\_4\_2\_2\_600}  & 12            & 0         & 12            \\
        \texttt{NGRAM\_CPS\_IMPL1\_params\_4\_3\_3\_1600} & 90            & 0         & 90            \\
        \texttt{NGRAM\_CPS\_IMPL1\_params\_6\_2\_2\_3200} & 48            & 0         & 48            \\
        \texttt{NGRAM\_CPS\_IMPL1\_params\_6\_3\_3\_3200} & 36            & 0         & 36            \\
        \texttt{NGRAM\_CPS\_IMPL1\_params\_8\_3\_3\_1600} & 6             & 0         & 6             \\
        \texttt{NGRAM\_CPS\_LRU\_params\_3\_3\_20000}     & 24            & 0         & 24            \\
        \texttt{REPWL\_params\_4\_2\_320\_2000}           & 54            & 0         & 54            \\
        \texttt{REPWL\_params\_6\_2\_320\_2000}           & 12            & 0         & 12            \\
        \texttt{NGRAM\_CPS\_IMPL2\_params\_4\_4\_20000}   & 18            & 0         & 18            \\
        \texttt{HALT\_MAX\_params\_3932964}               & 0             & 24        & 24            \\
        \hline
        Total                                             & 1{,}432{,}880 & 721{,}337 & 2{,}154{,}217 \\
        \hline
    \end{tabular}
\end{table}

\subsection{Exact $S(5)$ pipeline}\label{app:pipeS5exact}

\begin{table}[H]
    \centering

    \begin{tabular}{|lrrr|}
        \hline
        Exact $S(5)$ pipeline                             & Nonhalt         & Halt           & Total decided   \\
        \hline
        \texttt{LOOP1\_params\_130}                       & 126{,}950{,}828 & 48{,}367{,}435 & 175{,}318{,}263 \\
        \texttt{NGRAM\_CPS\_IMPL2\_params\_1\_1\_100}     & 3{,}291{,}498   & 0              & 3{,}291{,}498   \\
        \texttt{NGRAM\_CPS\_IMPL2\_params\_2\_2\_200}     & 1{,}328{,}432   & 0              & 1{,}328{,}432   \\
        \texttt{NGRAM\_CPS\_IMPL2\_params\_3\_3\_400}     & 497{,}142       & 0              & 497{,}142       \\
        \texttt{NGRAM\_CPS\_IMPL1\_params\_2\_2\_2\_1600} & 681{,}789       & 0              & 681{,}789       \\
        \texttt{NGRAM\_CPS\_IMPL1\_params\_2\_3\_3\_1600} & 91{,}101        & 0              & 91{,}101        \\
        \texttt{LOOP1\_params\_4100}                      & 43{,}269        & 12{,}276       & 55{,}545        \\
        \texttt{NGRAM\_CPS\_IMPL1\_params\_4\_2\_2\_600}  & 60{,}468        & 0              & 60{,}468        \\
        \texttt{NGRAM\_CPS\_IMPL1\_params\_4\_3\_3\_1600} & 28{,}868        & 0              & 28{,}868        \\
        \texttt{NGRAM\_CPS\_IMPL1\_params\_6\_2\_2\_3200} & 16{,}084        & 0              & 16{,}084        \\
        \texttt{NGRAM\_CPS\_IMPL1\_params\_6\_3\_3\_3200} & 5{,}213         & 0              & 5{,}213         \\
        \texttt{NGRAM\_CPS\_IMPL1\_params\_8\_2\_2\_1600} & 2{,}279         & 0              & 2{,}279         \\
        \texttt{NGRAM\_CPS\_IMPL1\_params\_8\_3\_3\_1600} & 855             & 0              & 855             \\
        \texttt{TABLE\_BASED}                             & 8{,}045         & 183            & 8{,}228         \\
        \texttt{NORMAL\_FORM\_TABLE\_BASED}               & 24              & 0              & 24              \\
        \hline
        Total                                             & 133{,}005{,}895 & 48{,}379{,}894 & 181{,}385{,}789 \\
        \hline
    \end{tabular}
\end{table}

The $S(5)$ pipeline differs from $S(2,4)$ and $S(4)$ in the following ways: for $8{,}045$ machines (marked as $\texttt{TABLE\_BASED}$), \CoqBB hardcodes either the parameters with which deciders should be called to solve them or their verifier certificates (see Section~\ref{sec:deciders-overview}); also, the 13 Sporadic Machines are proven using individual \Coq proofs, see Section~\ref{sec:sporadic}. Finally, machines marked \texttt{NORMAL\_FORM\_TABLE\_BASED} above are 24 machines which are proved using \texttt{1RB}-reduction (see Section~\ref{sec:deciders-overview}); we give the methods used to solve the \texttt{1RB} machines they reduce to.
\begin{table}[H]
    \centering
    \begin{minipage}{0.55\linewidth}
        \centering
        \begin{tabular}{|lrr|}
            \hline
            \texttt{TABLE\_BASED} machines             & Nonhalt & Halt \\
            \hline
            \texttt{REP\_WL\_params\_custom}           & 6{,}576 & 0    \\
            \texttt{NGRAM\_CPS\_IMPL2\_params\_custom} & 795     & 0    \\
            \texttt{NGRAM\_CPS\_IMPL1\_params\_custom} & 436     & 0    \\
            \texttt{HALT\_DECIDER\_47176870}           & 0       & 183  \\
            \texttt{LOOP1\_params\_1050000}            & 2       & 0    \\
            \texttt{NGRAM\_CPS\_LRU\_params\_custom}   & 182     & 0    \\
            \texttt{REPWL\_params\_20\_2}              & 1       & 0    \\
            \texttt{FAR\_certificates}                 & 23      & 0    \\
            \texttt{WFAR\_certificates}                & 17      & 0    \\
            \texttt{SPORADIC\_MACHINES}                & 13      & 0    \\
            \hline
            Total                                      & 8{,}045 & 183  \\
            \hline
        \end{tabular}
    \end{minipage}%
    \hfill
    \begin{minipage}{0.35\linewidth}
        \centering
        \begin{tabular}{|lr|}
            \hline
            \texttt{NORMAL\_FORM\_TABLE\_BASED} & Count \\
            \hline
            FAR                                 & 9     \\
            WFAR                                & 14    \\
            Sporadic Machine (Finned \#3)       & 1     \\
            \hline
            Total                               & 24    \\
            \hline
        \end{tabular}
    \end{minipage}
\end{table}

\vspace{-1em}
\section{Use of AI}

The $S(5)$ proof (2022--2024) is roughly concomitant with the striking progress of Large Language Models in AI. We disclose their use -- or lack thereof -- in our project:

\begin{itemize}
    \item \textbf{Code co-pilot.} AI-based code completion \textbf{was not} used in \CoqBB. AI-based code completion is unlikely to have been used for deciders written before \CoqBB (see \url{https://wiki.bbchallenge.org/wiki/Code_repositories}) as most of them were developed before AI-based code completion was mature. AI-based code completion was used to study \CoqBB in preparation of this paper (for instance: accelerating the reproduction of some algorithms, or translating \Coq code to Python in order to understand it better).
    \item \textbf{Writing co-pilot.} AI-based copy editing was used (i) to verify the spelling and grammar of the text and, (ii) to compress three lengthy paragraphs of human-written text (marked with a \texttt{\%SIA} comment in the LaTeX source, original paragraphs are kept commented). AI was \textbf{extensively used} to make and improve figures, mainly through LaTeX and \texttt{tikz} code generation, often prompting the AI with a hand-drawn layout of what was wanted and then iterating together -- \eg Figure~\ref{fig:TNF}, Figure~\ref{fig:WFAR}, Figure~\ref{fig:sporadic}, Figure~\ref{fig:zoology}.

\end{itemize}

\end{document}